\documentclass[11pt]{article}

\usepackage[a4paper, margin=1in]{geometry}
\usepackage{enumitem}
\usepackage{subcaption}

\usepackage{amsmath, amssymb, amsfonts, mathtools, dsfont, bbm}
\usepackage{tikz}
\usepackage{pgfplots}
\pgfplotsset{compat=1.18}
\usepackage{amsthm}
\usepackage{booktabs, array, placeins}
\usepackage{etoolbox}
\appto\normalsize{%
  \setlength{\abovedisplayskip}{14pt}%
  \setlength{\belowdisplayskip}{14pt}%
  \setlength{\abovedisplayshortskip}{14pt}%
  \setlength{\belowdisplayshortskip}{14pt}%
}
\usepackage[singlelinecheck=false]{caption}
\usepackage{mathpazo}

\usepackage[dvipsnames]{xcolor}
\definecolor{TealBlue}{rgb}{0.311, 0.558, 0.733}
\definecolor{GreenBlue}{rgb}{0.00, 0.71, 0.42}
\definecolor{YaleBlue}{rgb}{0.00, 0.21, 0.52}
\definecolor{YellowGreen}{RGB}{140,180,65}
\definecolor{LinkBlue}{RGB}{0,0,200}
\usepackage[most]{tcolorbox}
\usepackage{listings}
\lstdefinestyle{mypython}{
    language=Python,
    backgroundcolor=\color{white},
    commentstyle=\color{gray},
    keywordstyle=\color{blue},
    numberstyle=\tiny\color{gray},
    stringstyle=\color{orange},
    basicstyle=\ttfamily\footnotesize,
    breaklines=true,
    numbers=left,
    numbersep=5pt,
    showspaces=false,
    showstringspaces=false,
    showtabs=false,
    frame=single,
    tabsize=4,
    captionpos=b
}
\usepackage{algorithm}
\usepackage{algpseudocode}
\usepackage{comment}
\usepackage{graphicx}
\usepackage{layout}
\usepackage[authoryear]{natbib}
\usepackage{hyperref}
\hypersetup{
  colorlinks=true,
  citecolor=YaleBlue,
  linkcolor=YellowGreen,
  urlcolor=YaleBlue,
  pdfborder={0 0 0}
}

\usepackage[capitalize]{cleveref}
\newtheorem{remark}{Remark}
\makeatletter
\newcommand{\labeltext}[2]{%
  \@bsphack
  \MakeLinkTarget*{#1}%
  \def\@currentlabel{#1}%
  (#1)\label{#2}%
  \@esphack
}
\makeatother
\newcommand{\subsubsubsection}[1]{\paragraph{#1}}
\newlength{\fwbleed}
\newtcolorbox{phenomenon}[1]{%
  enhanced,
  breakable,
  colback=black!3,
  colframe=black,
  boxrule=1pt,
  arc=3pt,
  left=10pt, right=10pt, top=8pt, bottom=8pt,
  fonttitle=\normalfont,
  title={\underline{\textsc{#1}}},
  coltitle=black,
  attach title to upper={\par\medskip},
  title style={opacity=0},
}

\title{
{Restoring Incentive Compatibility in Two-Stage Energy Markets with Prosumers}}

\author{%
  \centerline{%
  \begin{tabular}{cccc}
    \textbf{Nikolas Koumpis} &
    \textbf{Koushik Kar} &
    \textbf{Leandros Tassiulas} &
    \textbf{Manolis Zampetakis} \\
    \textbf{Yale University} &
    \textbf{RPI} &
    \textbf{Yale University} &
    \textbf{Yale University} \\
    {\footnotesize\href{mailto:nikolaos.koumpis@yale.edu}{\texttt{nikolaos.koumpis@yale.edu}}} &
    {\footnotesize\href{mailto:kark@rpi.edu}{\texttt{kark@rpi.edu}}} &
    {\footnotesize\href{mailto:leandros.tassiulas@yale.edu}{\texttt{leandros.tassiulas@yale.edu}}} &
    {\footnotesize\href{mailto:manolis.zampetakis@yale.edu}{\texttt{manolis.zampetakis@yale.edu}}}
  \end{tabular}}%
}

\date{}

\begin{document}

\include{styles}  

\maketitle

\begin{abstract}
\noindent A central challenge in modern energy market design is the formulation of a strategy-proof imbalance settlement layer that secures both the economic efficiency of the institution and the stability of the power grid. Public data reveals that the day-ahead market is strategically biased below actual consumer demand. Such empirical observations are explained by \textit{active prosumers} which provide implementable incentives for demand under-reporting. Active prosumers buy energy in the \textit{day-ahead market} and sell energy in the \textit{real-time market} for balancing real-time energy deviations. By under-reporting their demand for the day ahead they inflate real-time imbalances and, under uniform pricing, they dispatch their generation assets more profitably. We model the \textit{two-stage} institution under linear preferences and benchmark it against its associated competitive equilibria. We show that although consumers' incentives for demand under-reporting vanish when the day-ahead market scales, prosumers incentives remain lower bounded by a positive gain which depends only on the real-time market generation stack and their shares over it. To restore incentive compatibility under the existing informational constraints, we design a leave-one-out contrastive scoring rule-based penalty that is implemented by the day-ahead market operator, incentivizes prosumers to report their demand truthfully and ensures small charges when participating honestly. We illustrate these results with numerical simulations on synthetic data and evaluate our mechanism on real-market data by first rationalizing demand reports as subjective equilibria of the induced game. Our mechanism demonstrates strong incentive alignment while retaining a low cost for honest participation.

\end{abstract}
\newpage
{\hypersetup{linkcolor=YaleBlue}\tableofcontents}

\newpage

\section{Introduction}\label{sec:introduction}
Every day, a \textit{day-ahead market} operator decides how much energy has to be delivered for each of the twenty-four hours of the next day. The amount that will be needed depends on what households, factories, hospitals and shops will eventually draw from the grid, and nobody knows these numbers exactly in advance; for each hour tomorrow the actual total demand is random. The operator asks the consumers of the day-ahead market to report their demand, and \textit{schedules} producers for tomorrow to match the total of all reports \citep{stoft2002power,euphemia}. 
The next day, because at any given time the energy supplied to the grid should balance the energy demanded,
a \textit{real-time (balancing) market} corrects, at delivery time, any deviation between the scheduled generation and the realized demand. In the context of this work, we use the term \textit{energy market} to refer specifically to the sequential coupling of the day-ahead and real-time markets.
\smallskip

Although scheduling generation to match demand is the most sensible thing to do subject to the grid balance constraint, this arrangement relies on the assumption that the day-ahead market participants report their demands truthfully. The operator cannot independently measure how much each consumer believes to need and has to trust their reports. The same trust shows up in many other institutions. Online advertising exchanges ask advertisers what an impression is worth to them, and then allocate it accordingly \citep{Edelman}. Multi-tier supply chains ask each tier to forecast demand so that the next tier can plan \citep{bullwhip_effect}. Insurance underwriters ask applicants for risk information and use it to price the premium \citep{rothy}. 
\smallskip
 
Our work was motivated by striking empirical observations. \autoref{fig:hrs} is based on public market data fetched
from \citep{entsoe} and reports a persistent gap between the aggregate demand reported by participants in the
day-ahead market and the total demand actually realised the next day for three European countries on three different years.

\begin{figure}[h]
    \centering
    \begin{subfigure}{0.32\textwidth}
        \centering
\includegraphics[width=\linewidth]{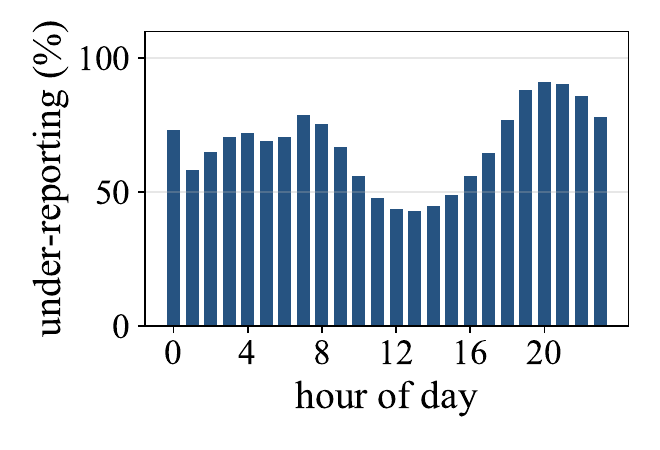}\par
        \vspace{0.4em}
        \includegraphics[width=\linewidth]{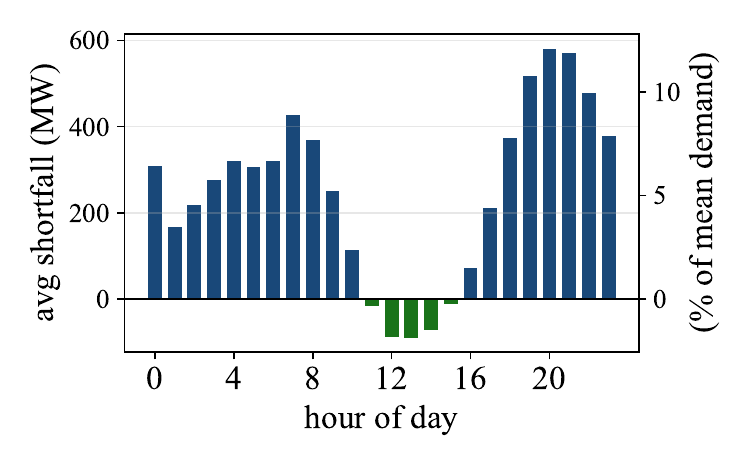}
        \caption{Hungary, 2023}
        \label{fig:dam-hu-2023}
    \end{subfigure}
    \hfill
    \begin{subfigure}{0.32\textwidth}
        \centering
\includegraphics[width=\linewidth]{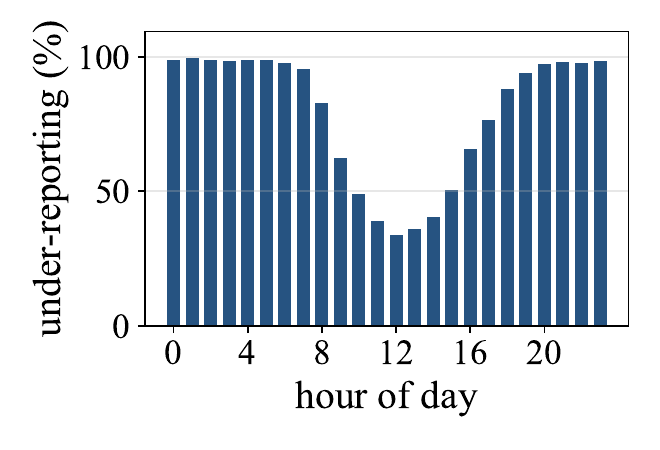}\par
        \vspace{0.4em}
        \includegraphics[width=\linewidth]{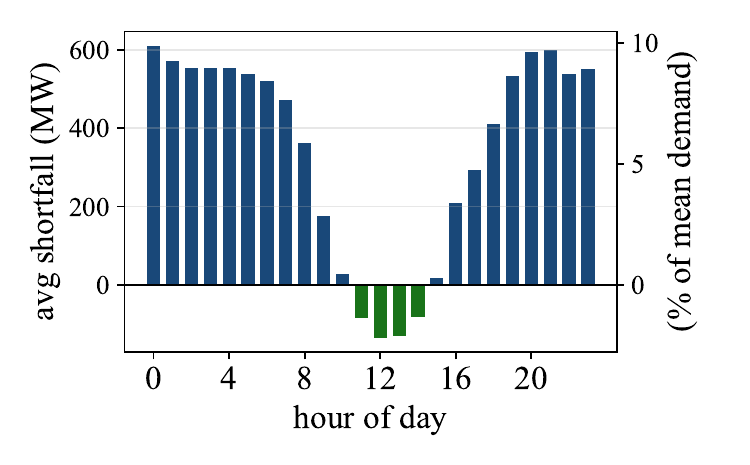}
        \caption{Romania, 2024}
        \label{fig:dam-ro-2024}
    \end{subfigure}
    \hfill
    \begin{subfigure}{0.32\textwidth}
        \centering
        \includegraphics[width=\linewidth]{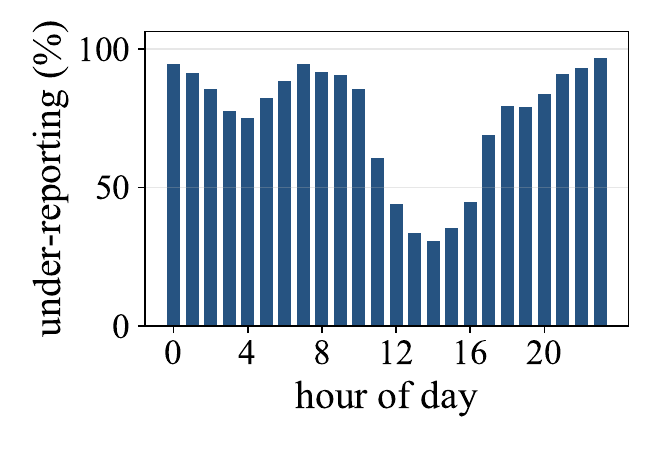}\par
        \vspace{0.4em}
        \includegraphics[width=\linewidth]{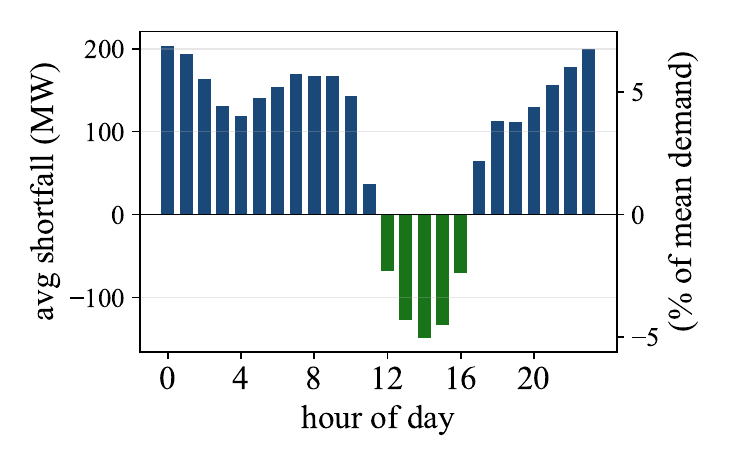}
        \caption{Slovakia, 2025}
        \label{fig:dam-sk-2025}
    \end{subfigure}
\vspace{5pt}
    \caption{Day-ahead market under-reporting across three European countries.}
    \label{fig:hrs}
\end{figure}
\noindent For each country we pick a recent year and show two views.
The top chart in each column gives, hour by hour, how often the volume
sold in the day-ahead market fell 10\% below the electricity people
actually used that hour. The bottom chart shows, again hour by hour,
how large that shortfall was on average -- in megawatts on the left
axis and as a share of the mean (actual) demand on the right axis. Blue bars mark
hours where the market scheduled less than what was needed; dark-green
bars mark hours where it scheduled more. The same pattern
shows up everywhere: the gap is largest overnight and in the evening,
and almost disappears around midday.
\smallskip 

This gap cannot be
attributed to prediction error: it is one-sided, it
persists across long observation windows, and survives changes in market conditions and
weather patterns. We argue that this gap reflects a strategic bias in the demand reports --- there exist energy market participants that systematically under-report their demand whenever is profitable under the \emph{existing} settlement rules. 
\smallskip

The economic consequence is direct. The day-ahead market collects reports from many self-interested players to recover the latent true distribution of next-day demand and shapes the very data it is trying to learn from. If participants do not report their demand honestly, the biased day-ahead market drives the real-time market towards a non-zero expected imbalance, and the real-time market clears more than just the irreducible random fluctuations in actual consumption, causing the market outcome to drift away from its competitive equilibrium, thus reducing social welfare.
\smallskip

\noindent That said, and given that in both the U.S. and the European energy markets, the settlement rules are an \emph{evolving design objective} (see \autoref{related-work}), the central question of this paper -- taking the demand reports as the strategic component -- is the following:
\begin{quote}
\itshape
     \textit{Can one design a settlement layer so that the institution estimates the latent demand belief and hence recovers the competitive equilibrium of the underlying economy?}   
\end{quote}

\subsection{Our Contributions.}
To answer this question, we start in \autoref{section-one} by first providing a unified institutional and formal model for the two-stage energy markets. This model is the common ground where all the subsequent results rest. Then, in \autoref{section-two} we identify and study the implementable incentives that explain the empirically observed strategic bias. While respecting institutional and informational constraints, we design in \autoref{sec:mechanism-design}, a contrastive scoring-rule-based mechanism implemented by the day-ahead market operator and added to all consumer's bill. \autoref{numerical-simulations} demonstrates our findings and evaluates our mechanism design via numerical simulations on synthetic and on real market data. 
\smallskip

\noindent Our theoretical and empirical contributions are as follows:
\paragraph*{Theoretical Contributions}
\begin{enumerate}[label=(\arabic*)]
     
    \item \textbf{Vanishing incentives to underreport for consumers} (\autoref{thm:eps-BIC}). For any consumer who does not own
    real-time generation, the maximum benefit from under-reporting her
    demand shrinks toward a non-positive value as the size $|I|$ of demand side of the day-ahead market grows, at a rate
    proportional to $1/\sqrt {|I|}$. Honesty is therefore an
    approximate best response for consumers in large day-ahead markets.

    \item \textbf{Non-vanishing incentives to underreport for prosumers}
    (\autoref{thm:not-eps-BIC}). For any prosumer --- a consumer who owns a strictly positive
    share of real-time generation assets --- the benefit from under-reporting
    is bounded below by a strictly positive constant that depends
    only on the size of her portfolio and the shape of the
    real-time supply stack. This constant does not depend on
    $|I|$. Therefore, under the standard institution not everyone can be made
    approximately honest simultaneously, no matter
    how large the day-ahead market grows.

    \item \textbf{Scoring-rule restoration of incentive compatibility} (\autoref{thm:main}, \autoref{cor:comparison} and \autoref{prop:tax-free}). While respecting the existing institutional  constraints, we modify the imbalance settlement by designing a
    contrastive scoring-rule-based payment that the day-ahead operator can compute from 
    observable quantities (the reports, the realised total consumption,
    and historical data) and add to each participant's bill. The
    extra payment is built from a strictly proper scoring rule
    and has a single tuning knob, the strong convexity parameter $\mu$ of the associated strongly convex function.
    We prove that, by choosing $\mu$ large enough, the prosumer's
    benefit from lying can be driven strictly below the persistent
    constant of (2). Honesty is therefore restored for everyone.
    \smallskip
    
    \noindent Turning $\mu$ up deters
    lying more strongly, but it also charges honest participation
    more heavily when their true demand happens to differ from the
    population average.~\autoref{prop:tax-free} in \autoref{trade-off} characterises this trade-off by providing the value $\mu^\star$ that
    balances the two effects and prove that, in a large market,
    both the residual incentive to lie and the residual cost paid
    by honest participants vanish at the same rate.
\end{enumerate}
\paragraph*{Empirical Contributions}
\begin{enumerate}[label=(\arabic*)]
    \item \textbf{Simulations on synthetic market data} (\autoref{synthetic results}).
    We illustrate the three qualitative phenomena presented by our theoretical results in a constructed synthetic market that focuses only on the incentive question. The day-ahead and real-time supply stacks are fixed and every player holds a uniform belief over private demand. We focus on a single player with fixed true demand across market sizes so that the comparisons have a common baseline.
In every experiment we evaluate the expected gain from under-reporting, with the expectation taken by Monte Carlo. The four experiments differ only in which mechanism is in force and which structural parameters are varied.

    \item \textbf{Evaluation on real market
    data} (\autoref{real market results}) 
    We demonstrate the performance of our mechanism by applying it to the strategic component that we infer from real market data. We work with hourly time series of low, medium and high voltage reported loads together with the corresponding realized demand.
We rationalize the day-ahead market reports as subjective equilibria by means of the prosumers parameters and real-time generation capacity bounds that make the observed day-ahead reports approximately optimal. Having inferred the model's parameters, we evaluate our scoring rule-based mechanism and 
demonstrate its effects. 
\end{enumerate}

\subsection{Related Work}\label{related-work}
This paper contributes to several literature domains. We organize our discussion of the related work around three themes \emph{(i)} the design of wholesale electricity markets and the institutional treatment of imbalance settlement, \emph{(ii)} mechanism design, subjective beliefs and scoring rules, \emph{(iii)} strategic machine learning. 

\paragraph*{Wholesale Electricity Market Design and Imbalance Settlement.}
The design of competitive wholesale electricity markets has been a subject of study since deregulation in the 1990s. \citep{stoft2002power}
is the first textbook treatment of the
two-stage structure --- day-ahead scheduling followed by real-time
balancing (what we previously called the energy market)--- that is the institutional setting of the present
paper. 
\smallskip

In Europe, the Third Energy Package (TEP) \citep{ThirdEnergyPackage2009}
and the most recent Clean Energy Package (CEP)
\citep{CleanEnergyPackage2019} are each, a set of laws passed in 2009 and 2019, respectively. The former established the framework for the internal electricity market and the most recent one changed the rules for two things: balancing (the real-time process ensuring that supply equals demand) and imbalance settlement (the rule for charging participants' deviations from schedule). In particular, Directive 2019/944 \citep{EU2019Directive944} from the CEP officially recognized two kinds of players that hadn't really existed in EU electricity law before -- active customers and aggregators. 
\begin{enumerate}[label=\textit{(\roman*)}]
    \item An \emph{active customer} is an energy market participant that isn't just passively buying electricity from the grid. They can produce and sell their own electricity. Directive 2019/944 defines them in Article 2(8) and gives them concrete rights in Article 15 (e.g. the right to sell their excess electricity).
    \item An \emph{aggregator} is an energy market participant that bundles together lots of small players and offers their combined flexibility to the electricity market as if it were one big power plant. The Directive 2019/944 defines aggregators in Article 2(18) and (19) and sets the rules for how they operate in Article 17.
\end{enumerate}
Article 5(1) of Regulation (EU) 2019/943 \citep{EU2019_943} establishes that all market participants are responsible for the imbalances they cause in the system. Since aggregators and active customers qualify as market participants under Article 2(25) of the same Regulation, they are subject to this obligation.
\smallskip

 Although both active customers and aggregators qualify as market participants under Regulation (EU) 2019/943, only aggregators are realistically positioned to engage in cross-market exchange (e.g., buying energy in the day-ahead market and selling it closer to real time, via the balancing market -- see Regulation (EU) 2019/943 \citep[Article 2(25)]{EU2019_943}). Active customers may benefit indirectly from such strategies by participating through an aggregator (Directive (EU) 2019/944, Article 15(2)(a)). 
\smallskip

That said, a central energy market entity of study in this paper
is the \emph{market-active prosumer}, or simply \emph{prosumer}: an entity that both consumes and produces electricity, qualifies as a market participant under Article 2(25) of Regulation (EU) 2019/943, and engages in cross-market exchange between the day-ahead and balancing markets. The prosumer thus sits at the intersection of the active customer (Directive (EU) 2019/944) and the aggregator function described in the same Directive, and is subject to the balance responsibility imposed by Article 5(1) of Regulation (EU) 2019/943. The formal definition we adopt for the player's economic role in our model is given in \autoref{subsec:actors}.
\smallskip
 
The EU regulatory framework signals that \emph{imbalance settlement is an active design objective}. Regulation (EU) 2019/943 establishes the general principle of balance responsibility. On top of that, recital 12 explicitly characterizes imbalance pricing as a behavioral incentive mechanism, stating that the pricing method for balancing energy "should create positive incentives for market participants in keeping their own balance or helping to restore the system balance [...], thereby reducing system imbalances and costs to society," and that such pricing approaches "should strive for the economically efficient use of demand response and other balancing resources." This framing is reinforced by the Electricity Balancing Guideline \citep{EC2017_Regulation_2017_2195}, to which recital 12 directly refers (Articles 18, 30 and 32). Taken together, these provisions reveal that the EU treats the design of imbalance settlement as an open, behaviorally consequential lever for integrating new market actors such as prosumers, rather than as a fixed technical detail. This paper contributes directly to this open design problem.
\smallskip

\noindent 
\paragraph*{Mechanism Design, Subjective Beliefs and Scoring Rules.} \noindent Mechanism design in its central insight, going back to Vickrey \citep{vickrey1961counterspeculation}, Myerson \citep{Myerson-optimal}, and a long line since, holds that honesty cannot be assumed -- it has to be \emph{engineered}. For mechanisms with money like energy markets, this often means designing a payment rule so that, for every participant, telling the truth is the most profitable thing to do. If lying is unprofitable, self-interest does the rest.
\smallskip

In our setting, participants report their demand for the next day. These reports can, in principle, be aggregated to reflect each player's \emph{subjective belief} about their consumption needs the following day. De Finetti \citep{definetti1937foresight} was the first who defined subjective belief axiomatically through betting: a person's subjective belief about an event is the price at which this person would be indifferent to sell or buy a \$1-contract written on that event. De Finetti's coherence argument shows that consistent beliefs are precisely those that admit no sure-loss bet. However, the question of how to \textit{elicit} a person's belief is not answered in his work.
\smallskip

\citep{glenn1950verification} independently, working in meteorology, likely without thinking about subjective probability foundations but rather about how to verify weather forecasters, was the first who proposed the quadratic scoring rule as an empirical tool. Two years later \citep{good1952rational} introduced the logarithmic scoring rule. De Finetti \citep[p. 357-364]{good1964scientist} observed that Brier's and Good's scores had the strict propriety property of \textit{measuring} -- in the sense of eliciting -- (true) subjective beliefs. This was a completely different insight from the question of forecasting skill: a strictly proper scoring rule controls the belief that a person communicates given what they know, and not the quality of what they know, which depends on their ability in information acquisition. Essentially, according to de Finetti, the scoring rule acts as an information-elicitation tool that filters out non-rational instincts that usually distort human judgment. By tying real-world consequences, the score aligns a person's self-interest with honesty. 

\smallskip

This insight was formalized into a general theory. \citep{mccarthy1956measures} and \citep{savage1971elicitation} characterized the class of strictly proper scoring rules and established their equivalence to strongly convex functions defined over distributions, while \citep{gneiting2007strictly} and \citep{GneitingRaftery2007_StrictlyProperScoringRules} extended the framework to continuous and multivariate forecasts and provided the modern synthesis on which most  applications rely.
\smallskip

Although in principle individual demand reports can be mapped to the player's subjective belief about their consumption needs, scoring these beliefs requires knowing their actual consumption levels. Such information can be obtained only via metering which introduces measurement errors and it is often of questionable reliability. That said, any settlement modification is intended to be implemented by the day-ahead market operator. Accordingly, the design should be subject to the information available to the day-ahead market stage (see \autoref{sec:mechanism-design}) which does not include the individual consumed energy levels. 
\smallskip

For each player, our mechanism builds on two \emph{population-level
subjective beliefs} -- one constructed from all reports including that player's and one constructed without it, thus charging based on their marginal contribution to the day-ahead market's forecasting ability. The first work where scoring-rule theory is applied to collective beliefs constructed from many players' inputs is \citep{hanson2003combinatorial, Hanson2007_LogarithmicMSR} where a \emph{market scoring rule} aggregates individual reports into a single distribution, and each player's payoff depends on how their action moves that collective belief. 
\smallskip

Our mechanism aggregates past demand reports and past actual consumption realizations, forms two market-level subjective beliefs for the real-time imbalance and scores each player based on their marginal contribution to the market's forecasting ability of the imbalance. Our mechanism can be viewed as a scoring-rule formulation of ideas previously developed in \citep{cai2015optimum} on truthful elicitation from strategic data sources, and \citep{ghorbani2019data} on marginal-contribution valuation. Casting these ideas in scoring-rule language allows us to establish approximate Bayesian strategy-proofness under the information constraints (see \autoref{sec:mechanism-design}). 
\smallskip

\noindent 
\paragraph*{Strategic Machine Learning.} The problem we study is, at its core, a strategic statistical estimation problem: an estimator -- here, the day-ahead market -- is fitted to data submitted by self-interested players whose reports depend on how the estimator will use them. This places our work within the growing literature on strategic machine learning, which examines learning and estimation problems in which data are produced by strategic players rather than drawn from a fixed distribution of nature. Foundational contributions established the basics of strategy-proof regression and classification \citep{dekel2010incentive, meir2012algorithms}, followed by the strategic classification framework of \citep{hardt2016strategic} and the performative prediction framework of \citep{perdomo2020performative}. Closest to our setting are the strategic-data-source framework of \citep{cai2015optimum}, which designs payments so that the resulting statistical estimate is optimal in expectation under truthful reporting, and the data valuation framework of \citep{ghorbani2019data}, which uses leave-one-out marginal contributions as a fair measure of each data point's value. We continue in this line, but apply its ideas to a domain not previously studied from this perspective -- electricity imbalance settlement -- in which the events being scored are endogenously defined by the real-time supply stack, a feature absent from both prior works.

\section{Institutional Background and Model}\label{section-one}
Our analysis uses the simple timeline that is common in energy markets:
(i) energy is scheduled ahead of delivery,
and any real-time deviations between scheduled and actual consumption are corrected by the system operator and settled financially afterwards. 
\subsection{Key Actors and Terms}\label{subsec:actors}
Terminology differs across regions. Therefore, before describing the institutional sequence, we briefly list the actors and terms that appear in
it. We map EU terms (NEMO, TSO, BRP, ENTSO-E) to
their closest U.S. counterparts (ISO/RTO, Balancing Authority, NERC).
\begin{description}
\item [ACTIVE PROSUMERS.] A \emph{prosumer} is an entity that (a) both consumes and produces electricity, (b) qualifies as a market participant under Article 2(25) of Regulation (EU) 2019/943 \citep{EU2019_943}, and (c) engages in cross-market exchange, in particular between the day-ahead and balancing markets. The prosumer thus sits at the intersection of the active customer defined in Directive (EU) 2019/944~\citep{EU2019Directive944} and the aggregator function described in the same Directive, and is subject to the balance responsibility imposed by Article 5(1) of Regulation (EU) 2019/943.
    \item[NEMO (EU) / ISO--RTO (US).]
In the EU, a \emph{Nominated Electricity Market Operator (NEMO)} is the entity designated to perform tasks
related to single day-ahead or single intraday coupling \citep[Art.~2(23)]{EU2015R1222}.
\smallskip

\noindent In the U.S., the closest counterpart is the \emph{Independent System Operator (ISO)} or
\emph{Regional Transmission Organization (RTO)}, which operates the system and runs day-ahead and
real-time wholesale markets \citep{FERC_IntroGuide_2025,DOE_BA_Backgrounder_2022}.

    \item[TSO (EU) / system operator and Balancing Authority (US).]
A \emph{Transmission System Operator \break (TSO)} is responsible for operating, maintaining, and (if needed)
developing the transmission system in a given area \citep[Art.~2(35)]{EU2019L0944}.
\smallskip

\noindent In the U.S., the operational responsibility to keep supply and demand balanced in real time belongs to a
\emph{Balancing Authority (BA)} (often an ISO/RTO in organized-market regions) \citep{NERC_Glossary_2025,DOE_BA_Backgrounder_2022}.

\item[ENTSO-E (EU) / NERC (US).]
ENTSO-E is the European network that brings together TSOs to support secure and coordinated operation of
Europe's interconnected grid \citep{ENTSOE_MissionStatement}.
\smallskip

\noindent In North America, bulk power system reliability standards are developed and overseen through the North American Electric Reliability Corporation (NERC) under
FERC oversight \citep{FERC_ReliabilityExplainer_2023}.

\item[LSE:] \emph{Load Serving Entity.} The retail-facing entity that is responsible for serving electricity to end-use consumers. An LSE forecasts consumption, purchases energy in wholesale markets (e.g., day-ahead),
    and is responsible for deviations between scheduled and actual consumption through imbalance settlement
    \citep[p.~24 (definition of ``Load-Serving Entity'')]{NERCGlossary2025}.

 \item[BSP:] \emph{Balancing Service Provider}. In the EU, a market participant with
reserve-providing units (or groups) able to provide balancing services to TSOs
\citep[Reg.~(EU)~ 2017/2195, Art.~2(6) (definition of ``balancing service provider''), OJ~L~312/9]{EU2017R2195}.
\smallskip

\noindent In U.S.\ ISO/RTO markets, the closest counterpart is a \emph{resource/market participant} that
offers \emph{ancillary services} (e.g., regulation and operating reserves) to the ISO/RTO and is
compensated according to the ancillary-service market rules when providing these reliability
services \citep{ISO_NE_Regulation_Item,PJM_AncillaryServices_FactSheet_2025}.

    \item[BRP (EU) / responsible market participant (US).]
A \emph{Balance Responsible Party (BRP)} is a market participant (or its representative) responsible for its
imbalances in the electricity market \citep[Art.~2(14)]{EU2019R0943}.
\smallskip

\noindent In U.S. ISO/RTO markets, the ISO/RTO financially settles real-time deviations through its real-time market and
settlement rules; the responsible party is the market participant (or scheduling entity) behind the schedule
\citep{FERC_IntroGuide_2025}.

\item[Buy/sell orders:] Instructions submitted to the market operator for a given delivery period.
In EU terminology these are orders that state how much
energy a participant wants to buy or sell for a given delivery period, and at what monetary rate (price) limit. \citep[Reg.~(EU)~2015/1222, Art.~2(21) (definition of ``order'')]{EU2015R1222}. A buy order specifies a quantity the participant is willing to purchase up to
a maximum price; a sell order specifies a quantity the participant is willing to supply above a
minimum price. The market clearing process aggregates these orders to determine the accepted schedules
and the market price.
\smallskip

\noindent In U.S.\ ISO/RTO terminology they are typically called \emph{bids} (to buy) and \emph{offers}
(to sell). A buy order/bid specifies a quantity and usually a maximum price; a sell order/offer
specifies a quantity and usually a minimum price. The market clearing process matches bids and
offers to determine the accepted schedules and the market price(s)
\citep{FERC_ISONE_Guide_2025,PJM_Manual11_Redline_2023}.

\end{description}

\subsection{Institutional Sequence: From Schedules to Imbalances}\label{subsec:institution}

We describe the institutional sequence at a level sufficient for incentive analysis. While the precise design differs across markets, the institutional functioning of energy markets can often be represented as the following sequence:\\[10pt]
\textbf{Day-Ahead Scheduling (DAM).}
Ahead of delivery, market participants submit buy/sell orders to the Day-Ahead Market (DAM)
operated by the NEMO. Market operator selects accepted demand and generation
schedules and computes uniform prices.
\\[10pt]
\textbf{Intraday Adjustments (IDM).}
After DAM, participants may adjust schedules closer to delivery in the Intraday Market. For clarity we skip intraday adjustments and treat the DAM schedule as the final schedule for the day ahead.
\\[10pt]
\textbf{Balancing and Activation (TSO layer).}
In real time, the TSO activates balancing resources to maintain system balance
(generation equals demand). Balancing resources can be integrated cross-border through certain
platforms (e.g. ENTSO-E).
\\[10pt]
\textbf{Imbalance settlement (financial layer).}
After delivery, BRPs are charged/credited for deviations between scheduled and realized positions using an imbalance price derived from balancing actions (design-specific details).

\subsection{Energy Mechanisms}\label{subsec:types}

Fix one delivery slot. An outcome of an energy market for this slot consists of:
(i) an electricity allocation $x$ and (ii) monetary transfers $t$.
A participant of an energy market has a private type that encodes their preference over the market outcomes. In energy markets, participants communicate their type directly to the mechanism and derive utility based on the market outcomes.
\begin{definition}[Mechanism with money]\label{def: mech-with-money}
A direct mechanism with money is a pair of maps
\[
(x(\cdot),t(\cdot)),
\]
where $x(\cdot)$ selects allocations and $t(\cdot)$ selects transfers as functions of the communicated types.
\end{definition}
\noindent Energy markets are direct mechanisms with money and in this section, we define the sequential coupling of the day-ahead and real-time institutions explicitly as such.

\subsubsection{The Day-Ahead Mechanism}\label{sec: DA mechanism}
The day-ahead mechanism is a direct mechanism with money executed by the NEMO. 
\smallskip

\noindent
\textbf{Consumers.} Individual households and businesses do not participate in the energy market themselves. Instead, an LSE bundles their total demand and buys the electricity for them. Therefore, within the day-ahead market, LSEs function as actual consumers. 
\smallskip

\noindent
\textbf{Producers.} Power plants act as the market’s producers. Rather than selling directly to the end-consumers, they submit supply orders to the day-ahead market based on their nominal generation capacity. Therefore, within this market, power plants function as the primary producers.
\smallskip

Each participant has a private type that encodes their preference over the mechanism's outcomes, the latter consisting of allocations and transfers (see  \autoref{def: mech-with-money}). The day-ahead market operates as a double auction, where both LSEs and producers simultaneously submit price-quantity orders (bids/offers) to the NEMO (ISO). Within the framework of mechanism design, these orders constitute the \textit{reported} types of the participants.
\smallskip

\noindent
\textbf{Simple Orders.} In this study, we focus on simple orders. Simple orders consist of straightforward price-quantity pairs thus yielding quasi-linear preferences (for more details about the different orders submitted in energy markets, see e.g. \citep{euphemia}). 
\smallskip

\noindent
Let $I$ be the set of consumers and $J$ the set of producers.
Each buyer $i\in I$ has a type
\[
\theta_i \equiv (v_i,\mathrm{d}_i)\in \mathbb{R}_{>0}\times \mathbb{R}_{\ge 0},
\]
where $v_i$ is their marginal value, i.e. willingness to pay, measured in \$/MWh and $\mathrm{d}_i$ is $i$'s \emph{true} demand (MWh) for the particular time slot.
Given an allocation $x_i\in[0,\mathrm{d}_i]$ and transfer $t_i\in\mathbb{R}$,
buyer $i$ has quasi-linear utility 
\[
U_i(x_i,t_i;\theta_i)=v_i x_i + t_i.
\]
Each producer $j\in J$ is characterized by a type
\[
\phi_j\equiv(c_j,g_j)\in\mathbb{R}_{>0}\times \mathbb{R}_{\ge 0},
\]
where $c_j$ is constant marginal cost (\$/MWh) and $g_j$ is the nominal generation capacity (MWh) for the hour.
Given dispatch $y_j\in[0,g_j]$ and transfer $t_j$, producer $j$ has quasi-linear utility
\[
U_j(y_j,t_j;\phi_j)=t_j - c_j y_j.
\]
Participants communicate their types to the mechanism by reporting their types to the NEMO. On the basis of these reports then, the NEMO determines allocations by solving a social welfare maximization program. On top of that, given these allocations, the NEMO assigns payments to the participants.
\smallskip

\noindent
In what follows we define the allocation and payment rules for the day-ahead mechanism under the assumption of simple orders (linear utilities).

\subsubsubsection{Day-Ahead Allocation Rule.}\label{sec:da-allocation rule} For a given reported profile $(\theta,\phi)$, where $\theta = (v,\mathrm{d}) \in \mathbb{R}^{2|I|}$ collects buyers' values and demands and $\phi = (c,g) \in \mathbb{R}^{2|J|}$ collects sellers' costs and nominal generation capacities, the NEMO's role is to find the allocation of energy that creates the most value for society. The NEMO solves
\begin{equation}\label{eq:DA-LP}
    (x^{DA},y^{DA}) \in \arg\max _{(x,y) \in \mathbb{R}^{|I|}_+ \times \mathbb{R}^{|J|}_+}
    \left\{\; v^\top x - c^\top y\;:\;\; x \leq \mathrm{d},\; y \leq g,\; \sum_{i \in I} x_i = \sum_{j \in {J}} y_j \right\},
\end{equation}
where the inequalities are component-wise and the equality constraint ensures that delivered energy matches demanded energy. 
\smallskip

\noindent \textbf{Market Restriction.} The welfare maximization program (\autoref{eq:DA-LP}) can be solved greedily by matching the highest-value consumers with the lowest-cost producers while maintaining the power balance constraint. 
\smallskip

Program (\autoref{eq:DA-LP}) has non-empty and bounded feasible region and hence a primal optimal $(x^{\star},y^{\star})$ exists and the optimal value is finite. The equality constraint can be relaxed and because strong duality holds it guarantees the existence of a dual optimal $p^{\star}$. The convex and piece-wise linear dual function reads:
\begin{align}\label{eq:decentralized}
    g(p)=\sum_{i \in I} \;\max _{x_i \in[0, \mathrm{d}_i]}\left(v_i-p\right) x_i+\sum_{j \in J} \;\max _{y_j \in[0, g_j]}\left(p-c_j\right) y_j,
\end{align}
\noindent
the best-responses to some feasible $p^{\star}$ are
\begin{equation}\label{eq:best-responses}
    \begin{aligned}
x_i^\star(p^\star) \in \arg\max_{x\in[0,\mathrm d_i]} (v_i - p^\star)\,x,
\qquad
y_j^\star(p^\star) \in \arg\max_{y\in[0,g_j]} (p^\star - c_j)\,y,
\end{aligned}
\end{equation}
and $p^{\star}$ is dual optimal if 
\begin{align}\label{eq:sub-gradient-dual}
    0\in \partial g\left(p^\star\right)=\left\{\;\sum_{i \in I} x_i = \sum_{j \in {J}} y_j \;:\; y_j \in y_j^\star\left(p^\star\right), x_i \in x_i^\star\left(p^\star\right)\right\}.
\end{align}
Then, $(x^{\star},y^{\star}, p^{\star})$, where $x^{\star}=x^{\star}(p^{\star})$ and $y^{\star}=y^{\star}(p^{\star})$, is primal-dual optimal.
According to (\autoref{eq:best-responses}), for some fixed dual optimal $p^{\star}$ the set of maximizers of (\autoref{eq:DA-LP}) is not a singleton and therefore, the allocations determined are in general multi-valued correspondences.
\smallskip

Because of this non-uniqueness, even small variations in reported types could shift market outcomes in an unknown manner. Then, a participant is not able to determine the consequence of their reports and subsequently the mechanism designer cannot understand incentives and hence establish desired incentive properties. To this end, we impose the following rather non-restrictive assumption over reported types.

\begin{assumption}[Sufficient conditions for uniqueness of DA-allocation]\label{as:DA-regular}
We restrict to order profiles such that $\forall i, i^{\prime} \in I, \forall j, j^{\prime} \in J: v_i \neq v_{i^{\prime}},\; c_j \neq c_{j^{\prime}},\; v_i \neq c_j$
\end{assumption}
We adopt  \autoref{as:DA-regular} to ensure the uniqueness of the market-clearing outcome. While this excludes exact ties between orders, such cases are negligible (of measure zero). This restriction allows us to treat the allocation rule as a well-defined function rather than a multi-valued correspondence, providing the necessary foundation for our incentive analysis and mechanism design in \autoref{section-two} and  \autoref{sec:mechanism-design}, respectively.
\begin{proposition}[Uniqueness of the DA allocation]\label{prop:DA-unique}
Under ~\autoref{as:DA-regular},
the DA allocation problem (\autoref{eq:DA-LP}) has a unique minimizer.
\end{proposition}

\begin{proof}[Proof]
Fix some given dual optimal $p^{\star}$. Then, under ~\autoref{as:DA-regular}, a consumer and a producer cannot be marginal simultaneously since \( v_i \neq c_j \) for all \( i \in I \) and \( j \in J \). As a result, one of the summations in~(\autoref{eq:sub-gradient-dual}) is constant and uniquely determines the allocation of the marginal participant. In the case where a producer is marginal, the existence of another producer with cost equal to the marginal cost would place both in the same merit order and can therefore be exchanged, while being
assigned one another's allocation. The same argument holds in the case where consumer is marginal. 
\end{proof} 
Each maximization problem in~(\autoref{eq:decentralized}) is solved independently by the participants with the dual variable \(p^{\star}\) driving their individual incentives so that the market clears, i.e.~(\autoref{eq:sub-gradient-dual}) holds. The dual optimal that allows players to act based only on their private information while still trading with one another, is called market-clearing price.
\smallskip

The allocations obtained by the market and the allocations assigned by the mechanism are in general different since the latter depend on the reported than the private types. However, if day-ahead mechanism market participants report their true types, then the two allocations coincide. 
\begin{definition}[Day-Ahead Market] A day-ahead market is a primal-dual optimal under private types. \end{definition}
\noindent \textbf{Fairness of Competitive Allocations.} We refer to market allocations as competitive allocations. A well-known fact from microeconomics is that competitive allocations are fair in the sense of Pareto: everyone is allocated what they deserve (according to their private type). 
That being said, it makes sense to benchmark our mechanism against the corresponding competitive allocation. This of course translates to incentivizing truthful reporting of the market participants.
\smallskip

In what follows, we define the day-ahead payments. Because energy markets clear through uniform pricing\footnote{Uniform price means that each participant is charged/payed a constant per unit of energy.}, we need first to define the pricing rule -- how the reported types determine the day-ahead price.
If we wish our mechanism to be able to recover competitive allocations (and hence be Pareto fair), we have to make sure that the pricing rule spans market clearing prices. First, we introduce some notation
\smallskip

\noindent
\textbf{Merit-Order Notation.}
Let $(c_{1},g_{1}),\dots,(c_{|J|},g_{|J|})$ denote producers ordered by increasing cost:
$c_{1}<\cdots<c_{|J|}$. Define cumulative capacities
\[
G_0\equiv 0,\qquad G_j\equiv \sum_{l\le j} \;g_{l}\quad(j=1,\dots,|J|).
\]
Let $G_{\vert J\vert}$ be the 
total generation capacity available in the day-ahead market and assume that it always suffices for covering the total schedule
$q^{\star}\equiv\sum_{i}x^{\star}_{i}\in[0,G_{|J|}]$. The marginal producer index is defined as the cheapest marginal cost that covers that total (see also
\autoref{fig:market curves})
\[
m(q^\star)\equiv \min\{l: q^\star\le G_l\}.
\]
Then, the cost of the marginal producer is $c({q^\star})$
where we slightly abuse the notation by identifying the cost of the $m-$th producer (which is the marginal) $c_{m(q^\star)}$, with the cost funtion at that quantity $c({q^\star})$. In a similar way, we define the marginal consumer. 
\smallskip

Let $(v_{1},d_{1}),\dots,(v_{|J|},d_{|J|})$ denote consumers ordered by decreasing willingness to pay :
$v_{1}>\cdots>v_{|I|}$. Define cumulative demands 
\[
D_0\equiv 0,\qquad D_l\equiv \sum_{i\le l}\; \mathrm{d}_{i}\quad(l=1,\dots,|I|).
\]
Then, the marginal consumer is defined as $n(q^\star)\equiv \max\big\{l: q^\star\le D_l\big\}$ and the valuation of the marginal consumer is $v(q^{\star})$, where once again, by $v(q^{\star})$ we mean $v_{n(q^{\star})}$.
\begin{proposition}[Competitive prices form a closed interval]\label{prop:cp-interval}
Let $(x^{\star},y^{\star})$ be a competitive allocation and $q^{\star}\equiv \sum_{i}x^{\star}_{i}$ the total quantity exchanged. Then, 
\[
p^{\star}
=
\Big[\max\{c(q^{\star}),\, v(q^{\star}+1)\},\ \min\{v(q^{\star}),\, c(q^{\star}+1)\}\Big].
\]
\end{proposition}

\begin{proof}
Because the pair $(x^{\star},y^{\star})$ is the projection of the competitive equilibrium into its first two coordinates, it automatically satisfies market clearing and therefore, its pre-image under the best responses \eqref{eq:best-responses} gives the competitive prices that are compatible with this allocation \footnote{We can take the best responses over all positive reals (prices). This gives us a set of only \emph{optimal} allocations. From that set of optimal allocations, only a subset corresponds to allocations (pairs) that satisfy clearing (non-empty), and we take the pre-image of that subset under the best responses.}.
We have
\begin{align}
p^{\star}
\in
\left(\bigcap_{i\in I}\{p\geq 0 : x_{i}=x^{\star}_{i}(p)\}\right)
\cap
\left(\bigcap_{j\in J}\{p\geq 0 : y_{j}=y^{\star}_{j}(p)\}\right).
\label{eq:competitive-prices}
\end{align}
By splitting consumers into those who are at least partially allocated and those who are allocated zero $I = I_{1}\cup I_{0}$,  $I_{1}=\{i\in I : x_{i}\in (0,\mathrm{d}_{i}]\}$,  $I_{0}=\{i\in I : x_{i}=0\}$ we obtain
\begin{align}
    &\cap_{i\in I}\{p\geq 0 : x_{i}=x^{\star}_{i}(p)\}=
\bigl(\cap_{i\in I_{1}}\{p\geq 0 : x_{i}=x^{\star}_{i}(p)\}\bigr)
\cap
\bigl(\cap_{i\in I_{0}}\{p\geq 0 : x_{i}=x^{\star}_{i}(p)\}\bigr).\nonumber
\end{align}
With $x_{i}\in (0,\mathrm{d}_{i}]$, we have $\{p\geq 0 : x_{i}=x^{\star}_{i}(p)\}\iff p\leq v_{i}$ and hence, the prices supporting those that are at least partially allocated are
\[
\cap_{i\in I_{1}}\{p\leq v_{i}\} = [0, v(q)].
\]
For $x_{i}=0$, we obtain $\{p\geq 0 : x^{\star}_{i}(p)=0\}\iff p\geq v_{i}$,
and 
\[
\cap_{i\in I_{0}}\{p\geq v_{i}\} = [v(q+1),+\infty).
\]
Therefore, the set to the left of the intersection in (\autoref{eq:competitive-prices})
reads $[\,v(q+1),\, v(q)\,]$. By splitting producers into those that are at least partially
allocated and non-allocated and subsequently working in a similar way as we did
for the consumers, the set to the right of the intersection in
(\autoref{eq:competitive-prices}) reads $[\,c(q),\, c(q+1)\,]$, and the result follows.
\end{proof}

\subsubsubsection{Day-Ahead Market Pricing Rule} Let 
\begin{align}\label{eq:cleared-quantity}
q(\theta,\phi)\;\equiv\;\sum_{i}x^{DA}_{i}(\theta,\phi)\in [0,{G}_{|J|}],
\end{align}
be the total quantity demanded for tomorrow. We follow the uniform-pricing convention and select the day-ahead price as the marginal cost of the
marginal producer (in merit order) supplying the cleared quantity:
\begin{equation}\label{eq:DA-price-rule}
p^{DA}(\theta,\phi)\;\equiv\; c({q(\theta,\phi)}).
\end{equation}
According to  \autoref{prop:cp-interval}, under  \autoref{as:DA-regular} the marginal cost of the marginal producer can be realized as a competitive price only if strictly 
\begin{align}\label{eq: cut-by-consumers}
c(q^{\star}) > v(q^{\star}+1).
\end{align}
This condition is rather qualitative. It implies that the decline in surplus is driven by a reduction in valuations
rather than an increase in costs. In particular, this is the case in \autoref{fig:market curves} where the market clears with the inverse demand curve ``cutting'' the inverse supply curve (see \autoref{fig:market curves}). In addition, (\autoref{eq: cut-by-consumers}) implies that non-zero allocated consumers are fully allocated their demand. This fact will simplify our incentive analysis later on in \autoref{section-two}.
\begin{assumption}\label{as: marginal cost price} We assume that the NEMO is always able to adjust market curves so that (\autoref{eq: cut-by-consumers}) holds.
\end{assumption}
\noindent Under  \autoref{as: marginal cost price} the pricing rule (\autoref{eq:DA-price-rule}) supports competitive allocations and hence renders possible benchmarking the day-ahead mechanism against competitive equilibria. As a result, it establishes a formal baseline for evaluating and achieving fairness. Having established the pricing rule, we are ready to define the day-ahead payment rule.
\smallskip

\noindent  \subsubsubsection{Day-Ahead Payment Rule} Energy markets clear through a uniform price. The payments are 
\begin{equation}\label{eq:DA-transfers}
\begin{aligned}
    t_i^{DA}(\theta,\phi)&\;\equiv\; -\,p^{DA}(\theta,\phi)\,x_i^{DA}(\theta,\phi)\quad(i\in I),
\\[5pt]
t_j^{DA}(\theta,\phi)
&\;\equiv\; +\,p^{DA}(\theta,\phi)\,y_j^{DA}(\theta,\phi)\quad(j\in J).
\end{aligned}
\end{equation}
\smallskip

\noindent Let $x \equiv (x_i)_{i=1}^{|I|}$, $y \equiv (y_j)_{j=1}^{|J|}$, and $t = \big((t_i)_{i=1}^{|I|}, (t_j)_{j=1}^{|J|}\big)$ with $t_i = p\,x_i$ and $t_j = p\,y_j$, denote the outcomes of the day-ahead mechanism.

\begin{definition}[Day-ahead mechanism]\label{def:da-mech}
Under \autoref{as:DA-regular} and \autoref{as: marginal cost price}, the day-ahead mechanism with money is defined as
\[
\rM^{DA}(\cdot) \equiv \big(x^{DA}(\cdot),\, y^{DA}(\cdot),\, t^{DA}(\cdot)\big) : (\theta,\phi) \mapsto (x,y,t),
\]
where $\big(x^{DA}(\cdot),\, y^{DA}(\cdot)\big)$ is given by (\autoref{eq:DA-LP}), $t^{DA}(\cdot) \equiv \big((t^{DA}_i)_{i\in I},\, (t^{DA}_j)_{j\in J}\big)$ is given by (\autoref{eq:DA-transfers}), and the clearing price $p^{DA}(\cdot)$ entering $t^{DA}$ is given by (\autoref{eq:DA-price-rule}).
\end{definition}
\subsubsection{The Real-Time Mechanism}\label{subsec:rt-mech}
The day-ahead producers that are dispatched to cover the actual consumption in the day-ahead  commit to supply $y$ based on the scheduled demand $x$ reported in the day-ahead market. However, the actual consumption $D$ (MWh), demanded in real-time is uncertain. The difference between the total scheduled demand $q(\theta,\phi)$
from the day-ahead market and the actual consumption places the grid in imbalance
\begin{align}\label{system-imbalance}
    \rho(D,(\theta,\phi))\;\equiv\; D-q(\theta,\phi),
\end{align}
and must be balanced with an appropriate real-time supply adjustment. The role of the real-time (balancing) market is therefore to provide this required quantity.
\smallskip

To safeguard operational security around the day-ahead plan, the TSO procures reserve capacity in advance
(FCR/FRR/RR) \citep[Chapter~1, Article~14]{EU2017Regulation2195} for maintaining physical balance
between generation and demand. Operationally, it does so by activating balancing resources (BSPs), the latter being the actors of the real-time market. A BSP can be either a producer, or market actor with assets capable of rapidly adjusting electricity production. 
\smallskip

Each BSP $j\in \overline J$ is characterized by a type $\overline{\theta}_{j}\equiv (\overline{c}_{j},\overline{g}_{j})\in \mathbb{R}_{>0}\times \mathbb{R}_{\geq 0}$, where $\overline c_k>0$ is the marginal cost and
$\overline g_k\ge 0$ the awarded capacity. 
For simplicity, we assume that the set of balancing resources awarded upward activation in the hour is the same with the BSPs.
\subsubsubsection{Real-Time Market Allocation Rule.} Given required balancing energy $\rho \geq 0$ and reported profile $\overline{\theta} = (\overline{c}, \overline{g}) \in \mathbb{R}^{|\overline{J}|}_{\geq 0}\times \mathbb{R}^{|\overline{J}|}_{\geq 0}$ collecting BSPs' costs and reserved capacities, the role of the TSO is to dispatch BSPs for balancing real-time imbalance at minimum cost. 
\begin{equation}\label{eq:RT-LP}
    y^{RT} \in \arg\min_{\overline{y} \in \mathbb{R}^{|\overline{J}|}_+}
    \left\{\; \overline{c}^\top \overline{y} \;:\;\; \overline{y} \leq \overline{g},\; \sum_{j \in \overline{J}} \overline{y}_j = \rho \right\},
\end{equation}
where the inequality is component-wise and $\sum \overline{y}_j = \rho$ is the balancing constraint ensuring that dispatched reserves clear the system imbalance. 
Let $\overline{G} \equiv \sum \overline{g}_j$ denote the total reserved capacity. For a consumption realization $D$, the LP \eqref{eq:RT-LP} has a solution only if $\overline{G}$ suffices for supplying the emerged system imbalance.
\begin{assumption}[RT feasibility]\label{as: RT-feas}
We assume that the TSO reserves enough capacity for supplying the real-time imbalances.    
\end{assumption}
\noindent  \autoref{as:RT-regular} ensures that the RT allocation rule (\autoref{eq:RT-LP}) remains well-defined. Failure to satisfy this condition poses a significant risk of grid instability; consequently, the TSO is required to secure real-time reserves to account for tail events. This approach trades grid safety with increased system costs, as the TSO must procure reserve capacity that may ultimately remain idle with high probability. Therefore, accurate consumption forecasting is vital for ensuring both grid reliability and market cost-effectiveness.
\smallskip

We can easily link the structure of the real-time allocation to the day-ahead allocation by modeling the TSO's commitment to balancing the grid as a fictitious consumer with finite willingness to pay greater than any real-time cost, $v_o > \overline{c}_j$ for all $j \in \overline{J}$, and demand equal to the imbalance, $\mathrm{d}_o \equiv \rho$
\begin{equation}\label{eq:RT-LP-2}
    y^{RT} \in \arg\max_{(x_o, \overline{y}) \in \mathbb{R}_+ \times \mathbb{R}^{|\overline{J}|}_+}
    \left\{\; v_o\, x_o - \overline{c}^\top \overline{y} \;:\;\; x_o \leq \mathrm{d}_o,\; \overline{y} \leq \overline{g},\; \sum_{j \in \overline{J}} \overline{y}_j = x_o \right\}.
\end{equation}
Similarly to the discussion in \autoref{sec:da-allocation rule}, we impose the following assumption for ensuring uniqueness of the real-time allocation.
\begin{figure}
    \centering
    \begin{minipage}{0.48\textwidth}
        \centering
        \includegraphics[width=\textwidth]{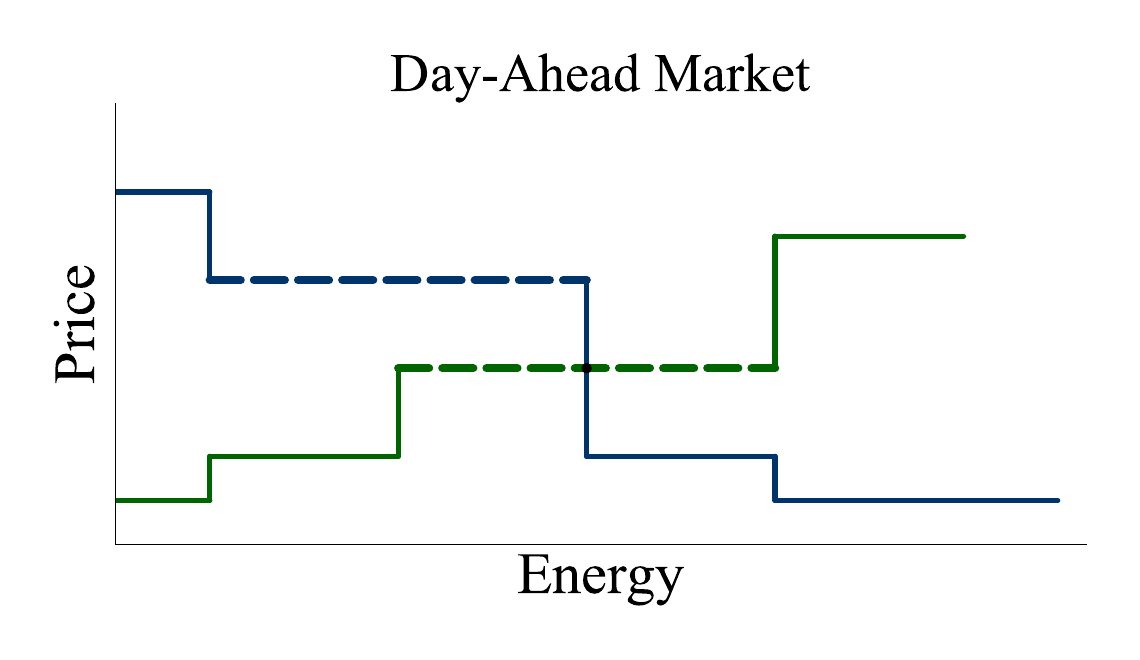}
    \end{minipage}
    \hfill
    \begin{minipage}{0.48\textwidth}
        \centering
        \includegraphics[width=\textwidth]{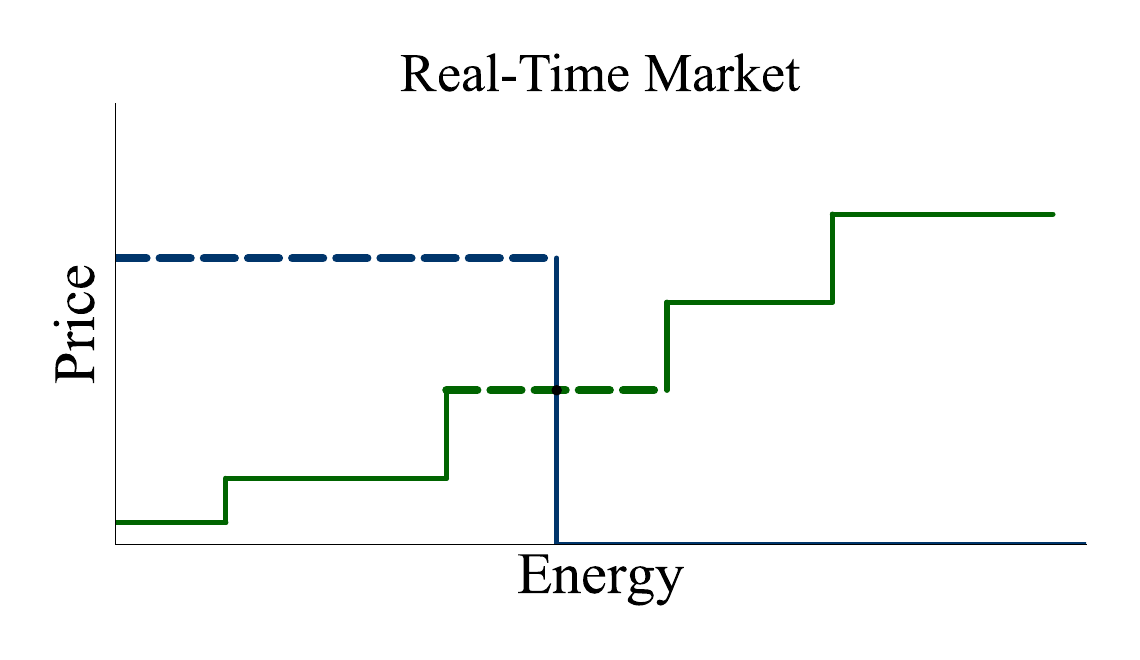}
    \end{minipage}
    \caption{(Left) Day-ahead market schedule for a particular time slot. Costs and valuations are placed in increasing and decreasing order, respectively. Each step corresponds to the best response of the participant (see also equation (\autoref{eq:best-responses})). The height of a step corresponds to marginal values whereas the width to quantities. The marginal consumer and producer correspond to those steps where the market curves ``meet''. The vertical gap at clearing encodes the range of market clearing prices. (Right) Real-time market curves. The dark-blue curve represents the TSO's demand for balancing energy, while the dark-green curve represents the awarded capacities.}
    \label{fig:market curves}
\end{figure}
\begin{assumption}[RT strict merit order]\label{as:RT-regular}
We assume $\overline c_{j^{\prime}}\neq \overline c_{j}$ for $j^{\prime}\neq j$.
\end{assumption}
\begin{proposition}[Uniqueness of real-time allocation]\label{prop:RT-competitive}
Under \autoref{as: RT-feas} and \autoref{as:RT-regular}, and for any realized consumption, the real-time problem (\autoref{eq:RT-LP}) has a unique
minimizer.
\end{proposition}
\begin{proof}
Given that the TSO's valuation is always greater than any real-time cost, under distinct marginal costs, no more than one real-time producers can be marginal and hence, for each $D$, allocation is pinned down by market clearing (see also proof of  \autoref{prop:DA-unique}) 
\end{proof}
\smallskip

Given that for each consumption realization the real-time allocation rule is structurally the same with the day-ahead allocation rule, the definition of real-time competitive equilibria, and pricing rules are also follow directly from the discussion of \autoref{sec:da-allocation rule}. Accordingly, under  \autoref{as: RT-feas}, for each realized consumption $D$, strong duality holds for (\autoref{eq:RT-LP}) and there exists dual optimal $\overline{p}^{\star}(D)$, that we call real-time market clearing price, that regulates real-time producers' incentives so that the real-time system imbalance is supplied
\begin{align}
        \overline{y}^{\star}_{j}\in \underset{\overline{y}\in [0,\overline{g}_{j}]}{\mathrm{argmax}} \;\;\overline p^{\star}(D)\;\overline{y}\;-\;\overline c_{j}\overline{y} \label{rt-gen-opt1}
\end{align}
and
\begin{align}
\rho\;=\;\sum\overline y^{\star}_{j}. \label{rt-gen-opt}
\end{align}
Since the TSO commits to delivering the day-ahead schedule, $v_o>\overline{p}^{\star}(D)$ and hence, there exists large enough budget $w_{o}>0$ such that for each $D$, TSO is allocated the grid imbalance $\rho(D,x^{\star})$,
\begin{align}
        \overline{x}^{\star}_o(D)\; =\; \rho(D,x^{\star}),\quad w_{o}\;\geq\; \overline p^{\star}(D)\rho(x^{\star},D),\quad Q-\text{a.e.}\label{rt-budget}
\end{align}
\begin{definition}[Real-time market]
Given day-ahead consumers' allocations $(x^{\star}_{i})_{i=1}^{|I|}$, a real-time market is a price $\overline{p}^{\star}(D)$ and allocations $\big(\overline{x}^{\star}_o(D),\, (\overline{y}^{\star}_{j}(D))_{j=1}^{|\overline{J}|}\big)$ such that (\autoref{rt-gen-opt1}), (\autoref{rt-gen-opt}), and (\autoref{rt-budget}) hold.
\end{definition}
\noindent Once again, $\left(x^{\star}_o(D),\overline{y}^{\star}(D),\overline{p}^{\star}(D)\right)$ are primal-dual optimal for (\autoref{eq:RT-LP}). Moving forward, fairness for the real-time mechanism is feasible only if the real-time pricing rule spans real-time market clearing prices.

\subsubsubsection{Real-Time Market Pricing Rule.}
Order balancing units by increasing costs,
with cumulative capacities $\overline G_r=\sum_{\ell\le r}\overline g_{(\ell)}$, and define
$\overline m(\rho)\equiv\min\{r:\rho\le \overline G_r\}$.
The real-time uniform pricing policy is defined as the marginal cost of the marginal unit:
\begin{equation}\label{eq:RT-price-rule}
p^{RT}(\theta,\phi,\overline \theta,D)\;\equiv \;\overline c{(\rho)},
\end{equation}
where again, we slightly abuse notation by writing $\overline c{(\rho)}$ instead of $\overline c_{(\overline m(\rho))}$. Similarly to (\autoref{eq: cut-by-consumers}) the marginal cost of marginal producer supports real-time competitive allocation only if 
\begin{align}\label{eq:RT-comp-price}
    \overline{c}(\rho)\;>\;v(\rho+1).
\end{align}
This condition can be easily satisfied assuming the existence of a second real-time consumer with zero valuation.
\begin{assumption}\label{as: rt-marginal cost price} We assume that the TSO is always able to adjust market curves so that (\autoref{eq:RT-comp-price}) holds.
\end{assumption}
\subsubsubsection{Real-Time Market  Payment Rule.} On top of their operational role, TSO's are granted financial role to act as a real-time settlement counterparty: they pay BSPs for the balancing energy
they deliver and charge (or credit) Balance Responsible Parties (BRPs) for their deviations from
schedule.  
\smallskip

\noindent Let $o_i(\theta,\phi)\equiv \mathrm{d}_i-x_i^{DA}(\theta,\phi)$ be BRP $i$'s deviation. With one-price settlement at the
real-time price $ p^{RT}(D,\theta,\phi)$, BRPs pay their imbalance 
\begin{equation}\label{eq: BRPs-payment}
   \hspace{59pt} t_i^{RT}(D,\theta,\phi)\;\equiv\;- \;p^{RT}(\theta,\phi,\overline \theta,D)\;o_{i}(\theta,\phi)\;\mathbbm{1}\left[x^{DA}_{i}(\theta,\phi)\neq 0\right],
\end{equation}
and BSPs receive 
\begin{equation}\label{eq: BSPs-payment}
   t_k^{RT}(D,\theta,\phi)\;\equiv\;+\; p^{RT}(\theta,\phi,\overline \theta,D)\;y_k^{RT}(D,\theta,\phi). 
\end{equation}
The real-time settlement process requires knowledge of $\mathrm{d}_i$ in order to allocate imbalance charges to each BRP, and this information is typically provided via metering. A BRP can either be the participant itself or an appointed representative that takes on the financial responsibility for imbalances. Conceptually, BRPs serve as the TSO's financial counterparties in the same way that buyers serve as the NEMO's counterparties in the day-ahead market.
\smallskip

\noindent
Let $x_o(D)$, $\overline{y}(D) \equiv (\overline{y}_j(D))_{j=1}^{|\overline J|}$, and $\overline{t} \equiv \big(\overline{t}_o(D),\, (\overline{t}_j(D))_{j=1}^{|J|}\big)$ denote the real-time outcomes, where $\overline{t}_o(D) = \overline{p}(D)\,x_o(D)$ and $\overline{t}_j(D) = \overline{p}(D)\,\overline{y}_j(D)$.

\begin{definition}[Real-time mechanism]\label{def: rt-mech} Under \autoref{as:RT-regular}, the real-time mechanism with money is defined as
\[
\rM^{RT}(\cdot,\;\theta, \phi, D)\;\equiv\; \left(\;y^{RT}(\cdot,\theta,\phi,\overline \theta,D),\;t^{RT}(\cdot,\theta,\phi,\overline \theta,D)\;\right) : \overline{\theta}\longmapsto \left(\;\overline x_o,\;\overline y(D),\;\overline t(D)\;\right)
\]
with $y^{RT}(\cdot)$ given by (\autoref{eq:RT-LP}), $t^{RT}(\cdot)\equiv \left(t^{RT}_{i}(\cdot),\, t^{RT}_{k}(\cdot)\right)$ given by (\autoref{eq: BRPs-payment}) and (\autoref{eq: BSPs-payment}), and $p^{RT}(\cdot)$ given by (\autoref{eq:RT-price-rule}).
\end{definition}

\noindent The day-ahead market operator (NEMO) implements only the day-ahead allocation and uniform payments as a function of submitted bids. The imbalance settlement is implemented by the BRP/TSO layer, which observes metered consumption and applies an exogenous settlement rule. Incentives must therefore be analyzed in the composed mechanism that we now define:

\begin{definition}[Energy mechanism]\label{def: energy-mechanism}
Under~\autoref{as:DA-regular}--\autoref{as:RT-regular}, the energy mechanism is a direct mechanism with money defined as
\[
\rM^{E}(D) \equiv
\left(x^{E}(\cdot),\; t^{E}(\cdot,\;D)\right) : (\theta,\phi,\overline{\theta})
\;\longmapsto\; \left(\;x,\;y,\;t,\;\overline x_o,\;\overline y(D),\;\overline t(D)\;\right),
\]
where $x^{E}(\cdot)\equiv \left(x^{DA}(\cdot),\;x^{RT}(\cdot)\right)$ and $t^{E}(\cdot,\;D)\equiv \left(t^{DA}(\cdot),\;t^{RT}(\cdot,\;D)\right)$.
\end{definition}
\noindent Observe that the day-ahead and real-time markets form two sequential equilibria, where the real-time equilibrium is driven by the day-ahead schedule. Moving forward, it is important to characterize the defined institutions w.r.t.\ whether they preserve money.
\begin{definition}[Ex-post budget balance]
A mechanism is \emph{ex-post budget balanced} if
\begin{equation}
\sum t_k(\cdot,\;D) = 0 \quad \mathrm{for~all~type~profiles}. ~\nonumber
\end{equation}
\end{definition}
\noindent Neither the day-ahead nor the real-time mechanism is ex-post budget balanced in isolation due to imbalance settlement. Nevertheless, the energy mechanism $M^{E}(D)$ is ex-post budget balanced by construction since
\begin{equation}\label{energy-budget-balanced}
    \begin{aligned}
    &\sum_{i \in I} t_i^{D A}(\theta,\phi)+\sum_{j \in J} t_j^{D A}(\theta,\phi)+\sum_{i\in I} t_i^{RT}(\theta,\phi,\overline \theta,D) \\[5pt]
&= p^{RT}(\theta,\phi,\overline \theta,D)\left(\;\sum_{j\in|\overline J|}\overline y_j^{RT}(\theta,\phi,\overline \theta,D)-\sum_{i\in I}\delta_i(D,\theta,\phi)\right)\\[5pt]
&\hspace{200pt}
+p^{D A}(\theta,\phi)\left(\;\sum_{j \in J} y^{D A}_j(\theta,\phi)-\sum_{i \in I} x^{D A}_i(\theta,\phi)\right)\\[5pt]
&=0.
\end{aligned}
\end{equation}

\section{Incentives}\label{section-two}

Market participants have preferences over the market outcomes, typically expressed as utility functions derived from their private types. On the basis of their private information, the publicly available institutional framework and their knowledge about the world and the others, they 
best respond to the induced game so as to influence the social choice for their own benefit.
\smallskip

As we already saw in \autoref{section-one}, the response of the real-time mechanism, driven by the day-ahead schedule, determines energy costs to the final consumers and affects grid stability. In particular, any bias of the day-ahead schedule from the actual consumption increases imbalances which later on are bought by the end consumers via the imbalance settlement (\autoref{eq: BRPs-payment}). Beyond financial damage, large biases threaten operational security by destabilizing the power grid.

\subsection{Social Choice}\label{sub-sub-sec: social-choice}

As we saw in \autoref{sec:introduction}, empirical observations in European energy markets—and, in particular, inspection of consumption data from \citep{entsoe} reveal a systematic underestimation of realized demand in the day-ahead stage. While forecast errors are always present in practice, the severity and persistence of this underestimation motivates the hypothesis that day-ahead market consumers are strategically misreporting their demands. 
\smallskip

To investigate this hypothesis, we assume that valuations, costs and nominal capacities are publicly available or at least, they can be inferred, and focus on the demands. Let $\mathrm{d}\equiv(\mathrm{d}_i)_{i\in I}$ denote the vector of true (private) demands. Building on the discussion in \autoref{sec:da-allocation rule}, if day-ahead players were acting in isolation, the resulting day-ahead allocation would be Pareto fair relative to the private types and in the absence of strategic bias, the TSO would have to safeguard only for the inherent uncertainty 
of the actual consumption. 
\smallskip

Consequently, Pareto fairness in the Day-Ahead stage is a prerequisite for minimizing real-time costs and preventing grid failures. Given that all types except demands are public, truthful reporting of demands in the day-ahead mechanism becomes the necessary and sufficient condition for Pareto fairness.
\smallskip

Given that the day-ahead and real-time institutions are linked, incentives for demand under-reporting in the day-ahead stage might be induced by preferences that extend beyond the day-ahead outcomes. Therefore, we benchmark the entire energy institution (see  \autoref{def: energy-mechanism}) against the corresponding market equilibria 
\begin{align}\label{eq: energy-mech}
    \rM^{E} : \mathrm{d} \longmapsto \Big(\, (x^{\star},\, y^{\star},\, t^{\star}),\; \big(x^{\star}_{o}(D),\, \overline{y}^{\star}(D),\, \overline{t}^{\star}(D)\big) \,\Big).
\end{align}

Intuitively, when the day-ahead allocation is truthful, the real-time imbalance varies around zero, thus minimizing risk of grid instability, the real-time price and the overall energy costs.
A nonzero average deviation means that the day-ahead plan inherits some strategic component (strategic bias). That said, observe that under public values, costs, and generation capacities, if 
\begin{align}\label{eq:da-benchmark}
    \rM^{DA}~:~ \mathrm{d}\mapsto \left(x^{\star},\;y^{\star},\;t^{\star}\right), 
\end{align}
is truthfully implemented, then so is $\rM^{RT}$ and hence $\rM^{E}$.
Crucially, because existence and uniqueness conditions are independent of the demands $(\mathrm{d})$, we can analyze incentives without concern that varying the demand reports will compromize the well-definedness of (\autoref{eq: energy-mech}).

\subsection{Setup}\label{sec:setup}
Consumers in the day-ahead-mechanism privately observe demands from some distribution. They also maintain a belief about the others types $\mathrm{d}_{-i}$ and actual consumption $D$, usually formed by aggregating past profiles and states of the world. 
\smallskip

\noindent Let $(\mathrm{d}_{-i}, D)\sim P_{i}$ be such a belief maintained player $i$ and let $P_{-i}$ denote the associated marginal w.r.t.\ $\mathrm{d}_{-i}$. Then, for a given sampled private type, they decide based on that subjective belief and play a subjective equilibrium (see \citep{sorin1992information, aumann1992perspectives, kalai1995subjective}). 
\smallskip

Unlike Bayes-Nash, subjective equilibria are inefficient since there is no guarantee that honest forecasting will ever emerge. If, however, under the energy mechanism, all consumers were required to report their demands truthfully, and this requirement were common knowledge, then the induced game would admit a well-defined common prior.
\smallskip

In order the mechanism participants to report truthfully on the basis of their knowledge, mechanism should be such that they rationalize truth-telling --approximately-- as their best response. 

\begin{definition}[$\varepsilon$-Bayesian Strategy-Proofness]\label{def:e-bic}
If bidders draw their types from some distribution $P=\times_{i=1}^m P_i$, then a direct mechanism $\rM$ is $\varepsilon$-Bayesian strategy-proofness, or ($\varepsilon$-Bayesian Incentive Compatible ($\varepsilon$-BIC)) with respect to $P$ if for each bidder $i \in I$
$$
\mathbb{E}_{\theta_{-i} \sim P_{-i}}\Big[U_i\left(\theta_i, \rM\left(\theta_i, \theta_{-i}\right)\right)\Big]\; \geq\; \mathbb{E}_{\theta_{-i} \sim P_{-i}}\Big[U_i\left(\theta_i, \rM\left(\theta_i^{\prime}, \theta_{-i}\right)\right)\Big]-\varepsilon,
$$
for all potential misreports $\theta_i^{\prime}$, in expectation over all other bidders bid $\theta_{-i}$. A mechanism is Bayesian strategy-proof if it is 0-Bayesian strategy-proof (resp. 0-BIC).
\end{definition}
\noindent The implementable preferences over the mechanism's outcomes span positive combinations of consumers in the day-ahead market, producers in the day-ahead market and producers in the real-time market. However, under the defined institutions (see \autoref{as:DA-regular} and \autoref{as:RT-regular}), day-ahead producers' components would constrain the demand under-reporting incentives and hence act against what we empirically observe in the European electricity markets. For that reason, we rule them out of the incentive study. 
\smallskip

The incentive question studied in this section is whether, under only cross-market ownership (prosumers), the rational behavior under the defined institutions is compatible with the market equilibria (see \autoref{section-one}).

\subsection{Consumers}
We now study the incentives of consumers in the energy institution. Typically, preferences over the day-ahead consumption outcomes are implemented by (LSEs), the latter being associated with their own clientele. Thus naturally, such entities are modeled as a single consumer. Consumer $i$'s preferences over outcomes $(x,t)$ of the energy mechanism are given by 
\begin{equation}
    U_{i}(x,t)\;=\;v_{i}\;x_{i}+t_{i}.\nonumber
\end{equation}
Being aware of the institutional framework, each consumer affects the social choice by acting rationally and reporting their demand $\mathrm{d}_{i}\in \mathbb{R}_{\geq 0}$ to the day-ahead institution so that their preferences are maximized in expectation.
Let $( \mathrm{d}^{\prime}_{i},\mathrm{d}_{-i})\in \mathbb{R}^{|I|}_{\geq 0}$ be the 
reported demand profile in the day-ahead mechanism. At that profile, under mechanism $\rM^E$, consumer $i\in I$ derives net utility
\begin{equation}\label{consumers-utility} 
   \hspace{-8pt} U_{i}\big(\mathrm{d}_{i},\;\rM^E(\mathrm{d}^{\prime}_{i},\mathrm{d}_{-i})\big)
    =
    v_{i}\mathrm{d}_i+t^{DA}_i(\mathrm{d}^{\prime}_i,\mathrm{d}_{-i})+t_i^{RT}(D,\mathrm{d}^{\prime}_{i},\mathrm{d}_i,\mathrm{d}_{-i}),
\end{equation}
where, according to (\autoref{eq:DA-transfers})
\[
t^{DA}_i(\mathrm{d}^{\prime}_i,\mathrm{d}_{-i})=-p^{DA}(\mathrm{d}^{\prime}_{i},\;\mathrm{d}_{-i})\;x^{DA}_{i}(\mathrm{d}^{\prime}_{i},\;\mathrm{d}_{-i}),
\]
and according to (\autoref{eq: BRPs-payment}) 
\begin{align}
\hspace{-5pt}t_i^{RT}(D,\mathrm{d}^{\prime}_{i},\mathrm{d}_i,\mathrm{d}_{-i})    = -p^{RT}\big(\rho(D,x^{DA}(\mathrm{d}^{\prime}_{i},\mathrm{d}_{-i}))\big)\big(\mathrm{d}_{i}-x^{DA}_{i}(\mathrm{d}^{\prime}_{i},\;\mathrm{d}_{-i})\big)\mathbbm{1}\!\big[x^{DA}_{i}\neq 0\big].\nonumber
\end{align}

\noindent
The first term in (\autoref{consumers-utility}) is consumer $i$'s utility from both the day-ahead and the real-time stage. The second term is the payment charged to the consumer in the day-ahead stage. If consume more than bought in the day-ahead stage, consumers have to buy the remainder at the real-time price. The third term in (\autoref{consumers-utility}) expresses the payment for these imbalances as being defined by the real-time settlement (\autoref{eq: BSPs-payment}). When $i\in I$ is zero day-ahead-allocated the real-time settlement is not applied (see (\autoref{eq: BRPs-payment})) and the net utility vanishes.
To set the stage for our incentive analysis, given report $\mathrm{d}^{\prime}_{i}\in[0,\mathrm{d}_i]$ let 
\begin{align}\label{eq: dispatch event}
    A_i\;\equiv\; \left\{\;\mathrm{d}_{-i}\in \mathbb{R}^{|I|-1}_{\geq 0}\;\;\Big|\;\; v_i>p^{DA}(\mathrm{d}^{\prime}_{i},\mathrm{d}_{-i})\right\}, 
\end{align}
be the event within which $i$ is allocated when reporting $\mathrm{d}^{\prime}_i$. Let $A_i=A(\mathrm{d}_i)$ be the event where $i$ is allocated when reporting truthfully, and $A^{\prime}_i=A(\mathrm{d}^{\prime}_i)$ the event where $i$ is allocated when under-reporting.  
\smallskip

Under \autoref{as: marginal cost price} a consumer that is allocated in the day-ahead market is (fully) allocated their report. Therefore,
\begin{align}\label{eq:da-alloc-under-assumption-1}
x^{DA}_{i}(\mathrm{d}^{\prime}_{i},\mathrm{d}_{-i}) \;=\; \mathrm{d}^{\prime}_{i}\;\mathbbm{1}_{A^{\prime}_i}(\mathrm{d}_{-i}).
\end{align}
If consumer $i$ is allocated when reporting truthfully, so is when under-reports
\begin{align}\label{eq:da-alloc-ind}
\mathbbm{1}_{A_i}(\mathrm{d}_{-i})=1\implies\mathbbm{1}_{A^{\prime}_i}(\mathrm{d}_{-i})=1,
\end{align}
since due to merit order, under-reporting in the day-ahead mechanism can only drop the day-ahead price. In addition, observe that for some report $\mathrm{d}^{\prime}_{i}$
\begin{align}\label{eq: da-alloc-eq-ind}
\mathbbm{1}\left[x^{DA}_i(\mathrm{d}^{\prime}_{i},\mathrm{d}_{-i})\neq 0\right]=\mathbbm{1}_{A^{\prime}_i}(\mathrm{d}_{-i}).
\end{align}
The point-wise gain obtained by under-reporting $\mathrm{d}^{\prime}_{i}=\mathrm{d}_i-\delta_i$ is defined as
\begin{align}\label{eq: gain-under-reporting}
    G_{i}\big(\mathrm{d}_{i},\;\rM^E(\mathrm{d}_i-\delta_i,\mathrm{d}_{-i})\big)\;\equiv\; U_i\big(\mathrm{d}_{i},\;\rM^E(\mathrm{d}_i-\delta_i,\mathrm{d}_{-i})\big)-U_i\big(\mathrm{d}_{i},\;\rM^E(\mathrm{d}_i,\mathrm{d}_{-i})\big).
\end{align}

\noindent Given their true (private) demand, consumers decide about their report by maximizing their utility in expectation over all other reports and actual consumption. 
By plugging (\autoref{eq:da-alloc-under-assumption-1}), (\autoref{eq:da-alloc-ind}), (\autoref{eq: da-alloc-eq-ind}) into (\autoref{eq: gain-under-reporting}) we obtain:

\begin{equation}
    \begin{aligned}
G_{i}\big(\mathrm{d}_{i},\;\rM^E(\mathrm{d}_i-\delta_i,\mathrm{d}_{-i})\big)&= \;\Big( 
 p^{DA}(\mathrm{d}_i,\mathrm{d}_{-i})\mathbbm{1}_{A_i}-p^{DA}(\mathrm{d}_i-\delta_i,\mathrm{d}_{-i})\mathbbm{1}_{A^{\prime}_i}
\Big)\mathrm{d}_i\nonumber\\[5pt]
 & + \Big( p^{DA}(\mathrm{d}_i-\delta_i,\mathrm{d}_{-i})-p^{RT}(\rho)\Big)\delta_i\mathbbm{1}_{A^{\prime}_i},\nonumber
\end{aligned}\label{eq:random}
\end{equation}
where $\rho=\rho(\delta_i)= 
 D - \sum_{l} x^{DA}_{l}\!\big(\mathrm{d}_i-\delta_i,\mathrm{d}_{-i}\big)$. At this point we assume that consumer $i$ believes that she is certainly allocated.
\begin{assumption}[Certain allocation]\label{ass:certain-alloc}
When reporting truthfully, total demand does not exceed total supply capacity: $\mathrm{d}_i + \sum_{j \neq i}\mathrm{d}_j \leq G_{|J|}$ almost surely, so that $\Pr(A_i) = 1$.
\end{assumption}
\noindent Under~\autoref{ass:certain-alloc}, $\Pr(A^{\prime}_i) = 1$. By taking expectations~(\autoref{eq:random}) reads

\begin{align}\label{eq:Delta-consumer-33}
\mathbb{E}_{(D,\mathrm{d}_{-i})\sim P_i}\;G_{i}\big(\mathrm{d}_{i},\;\rM^E(\mathrm{d}_i-\delta_i,\mathrm{d}_{-i})\big)
&=
\mathbb{E}\!\Big[
p^{DA}(\mathrm{d}_i,\mathrm{d}_{-i})-p^{DA}(\mathrm{d}_i-\delta_i,\mathrm{d}_{-i})
\Big]\mathrm{d}_i \nonumber\\[5pt]
&\hspace{40pt}+\mathbb{E}\!\Big[
p^{DA}(\mathrm{d}_i-\delta_i,\mathrm{d}_{-i})
-p^{RT}(\delta_i,\mathrm{d}_{-i})
\!\Big]\delta_i\;
\nonumber\\[5pt]
&\leq 
\mathbb{E}\!\Big[
-p^{DA}(\mathrm{d}_i-\delta_i,\mathrm{d}_{-i})
+p^{DA}(\mathrm{d}_i,\mathrm{d}_{-i})
\!\Big]\mathrm{d}_i \nonumber\\[5pt]
&\hspace{40pt}+\mathbb{E}\!\Big[
p^{DA}(\mathrm{d}_i,\mathrm{d}_{-i})
-p^{RT}(0,\mathrm{d}_{-i})
\!\Big]\delta_i,
\end{align}
where $\mathbb{E}\!\left[\, p^{RT}\!\big(\rho(\delta_i)\big) \,\middle|\, \mathrm{d}_{-i} \right]
= p^{RT}(\delta_i,\mathrm{d}_{-i}),$
and the inequality follows because the day-ahead price increases with the demand report of $i$ and the real-time price is increasing with imbalance $\rho$. 
According to the first term of (\autoref{eq:Delta-consumer-33}), by requesting less energy, consumer $i$ pulls the day-ahead price down. For handling the second term of (\autoref{eq:Delta-consumer-33}) we state the following

\begin{assumption}[Real-time market is more expensive than the day-ahead market]\label{ass:DA-RT}
When consumer $i$ reports truthfully, in expectation over all others bid $\mathrm{d}_{-i}$
\begin{align}
\mathbb{E}\Big[p^{DA}(\mathrm{d}_i,\mathrm{d}_{-i})
-p^{RT}(0,\mathrm{d}_{-i})\Big] \;\leq\; 0.
\end{align}
\end{assumption}
\smallskip

\noindent Moving forward, because  the inverse supply curve is piece-wise constant, if the reduction stays within the same step, the same producer remains marginal and the price does not change at all. Further, the price drops only if the reduced  demand crosses a step boundary in the supply curve. 
\smallskip

\noindent The following Lemma extends Lipschitz continuity to piece-wise constant inverse supply.

\begin{lemma}[Extending Lipschitz continuity to the piecewise constant inverse supply curve]
\label{lem:local-cost-gap}
For all $\mathrm{d}_{-i}$ in the support of $P_{-i}$, all feasible demands $\mathrm{d}_i$, and all under-reports $\delta_i\in [0,\mathrm{d}_i]$,
\begin{align}
 \Big\vert p^{DA}(\mathrm{d}_i,\mathrm{d}_{-i})-p^{DA}(\mathrm{d}_i-\delta_i,\mathrm{d}_{-i})\Big\vert\leq L \vert \delta_i \vert +2H,
\end{align}
where
\[
L=\max_{0\leq j \leq |J|-1}\frac{c_{j+1}-c_{j}}{G_{j+1}-G_{j}},\qquad 2H=2\max_{0\leq j\leq |J|-1} \vert c_{j+1}-c_{j}\vert.
\]
\end{lemma}

\begin{proof}
Let $D = \mathrm{d}_i + \sum_{j\neq i}\mathrm{d}_j$ denote the total demand under truthful reporting and $D' = (\mathrm{d}_i - \delta_i) + \sum_{j\neq i}\mathrm{d}_j$ the total demand under the under-report, so that $D - D' = \delta_i$. From (\autoref{eq:cleared-quantity}) and (\autoref{eq:DA-price-rule}), within the event $A_i$ the day-ahead market clearing price equals the inverse supply function $c:[G_0,G_{|J|}]\to \mathbb{R}_{\geq 0}$ evaluated at total demand:
\[
p^{DA}(\mathrm{d}_i,\mathrm{d}_{-i})=c(D),\qquad p^{DA}(\mathrm{d}_i-\delta_i,\mathrm{d}_{-i})=c(D'),
\]
so the price difference we seek to bound is $|c(D)-c(D')|$.
\smallskip

\noindent 
The inverse supply function is piecewise constant over the partition $G_0<G_{1}<\dots<G_{|J|}$:
\[
c(x)=c_j\quad \text{for}\quad x\in [G_{j},G_{j+1}),\quad j=0,1,\ldots,|J|-1.
\]
It has $|J|$ jump discontinuities at the partition points, with jump sizes $|c_{j+1}-c_j|$.

\smallskip
\noindent\emph{Step 1: Piecewise linear interpolant.}
We construct a continuous approximation of $c$ by connecting the corners
\[
(G_0,c_0),\;(G_1,c_1),\;\ldots,\;(G_{|J|},c_{|J|})
\]
with straight line segments. This defines the piecewise linear interpolant $g:[G_0,G_{|J|}]\to \mathbb{R}_{\geq 0}$, given on each subinterval $[G_{j},G_{j+1}]$ by
\[
g(x)=c_j+\frac{c_{j+1}-c_{j}}{G_{j+1}-G_{j}}(x-G_j).
\]
By construction, $g$ is continuous on $[G_0,G_{|J|}]$, linear on each subinterval $[G_{j},G_{j+1}]$, and agrees with $c$ at every partition point: $g(G_{j})=c(G_j)=c_j$ for each $j=0,1,\ldots,|J|$.

\smallskip
\noindent\emph{Step 2: Lipschitz constant of the interpolant.}
Since $g$ is piecewise linear and continuous, it is Lipschitz continuous with global Lipschitz constant
\[
L=\max_{0\leq j \leq |J|-1}\frac{\vert c_{j+1}-c_{j}\vert }{G_{j+1}-G_{j}}.
\]
That is, for all $x,y \in [G_0,G_{|J|}]$,
\[
\vert g(x)-g(y)\vert \leq L \vert x-y\vert.
\]

\smallskip
\noindent\emph{Step 3: Pointwise approximation error.}
Fix $x\in [G_j,G_{j+1})$. On this interval, $c(x)=c_j$, while the interpolant satisfies $g(x)\in [\min(c_j,c_{j+1}),\;\max(c_j,c_{j+1})]$. Thus
\[
\vert c(x)-g(x)\vert = \vert c_j-g(x) \vert \leq \vert c_{j+1}-c_{j}\vert.
\]
Taking the supremum over all $x$ and all subintervals,
\[
H\equiv \Vert c-g\Vert_{\infty}\leq \max_{0\leq j\leq |J|-1} \vert c_{j+1}-c_j\vert.
\]
The quantity $H$ is bounded by the largest jump in the inverse supply function.

\smallskip
\noindent\emph{Step 4: Combining.}
Applying the triangle inequality to the decomposition $c = (c - g) + g$,
\begin{align}
    \vert c(D)-c(D')\vert &= \vert c(D)-g(D)+g(D)-g(D')+g(D')-c(D') \vert \nonumber\\[5pt]
    &\leq \vert c(D)-g(D)\vert + \vert g(D)-g(D')\vert + \vert g(D')-c(D')\vert \nonumber\\[5pt]
    &\leq \Vert c-g\Vert_{\infty} + L\vert D-D'\vert + \Vert c-g\Vert_{\infty}\nonumber\\[5pt]
    &= L\vert \delta_i\vert+2H,\nonumber
\end{align}
where the last line uses $|D - D'| = |\delta_i|$ and $\|c - g\|_\infty \leq H$.
\end{proof}
\noindent Empirical evidence suggest that the day-ahead market clears within a sub-range of costs where the Lipschitz part and the price gap are both very small.

\subsubsection{Anti-Concentration of Aggregate Demand.} While Lemma \autoref{lem:local-cost-gap} alone yields a deterministic bound on the gain from the day-ahead price drop, this bound can be further tightened by accounting for the fact that the price change occurs only when the deviation causes total demand to cross a supply step boundary. Whether or not this crossing occurs depends on where the aggregate demand of other players $\sum_{j\neq i}\mathrm{d}_j$, happens to fall relative to the step boundaries of the supply curve. If player $i$'s beliefs $P_{-i}$ are sufficiently spread out ---that is, if no small interval of possible demand values is assigned too much probability---then the chance of such a crossing is small, and the expected gain shrinks accordingly.
\smallskip

As mentioned earlier in \autoref{sec:setup}, players form their beliefs about others' total demand from historical market data. The condition we require is that these beliefs do not concentrate excessively in any narrow interval. This is a rather mild requirement: it holds whenever the historical demand data is not pathologically concentrated at a single point and it is directly verifiable from past data.
\smallskip

The appropriate tool for formalizing this condition is the \emph{L\'{e}vy concentration function}, a classical object in probability theory introduced by L\'{e}vy~\citep{Levy1937} and developed systematically by Petrov~\citep[Chapter~3]{Petrov1975}.
\smallskip

For a random variable $X$ and a width $\delta \geq 0$, the L\'{e}vy concentration function is defined as
\begin{align}\label{eq:levy}
Q(X,\,\delta) \;\equiv\; \sup_{x \in \mathbb{R}}\; \Pr\big(X \in [x,\; x+\delta]\big).
\end{align}
It measures the maximum probability that $X$ places on any interval of width $\delta$. A small value of $Q(X,\delta)$ means that the distribution of $X$ is spread out and does not concentrate too much mass in any narrow region. As we show below, the L\'{e}vy concentration function of the aggregate demand $\sum_{j\neq i}\mathrm{d}_j$ is small whenever enough individual players exhibit variability in their bids. This follows from the Kolmogorov--Rogozin inequality, which we now recall.
\begin{proposition}[Kolmogorov--Rogozin inequality~{\citep{Rogozin1961}; see also~\citep[Chapter~3, $\S 2$]{Petrov1975}}]\label{thm:rogozin}
Let $\xi_1,\ldots,\xi_n$ be independent random variables. For any $l>0$, define
\begin{align}
Q_k(l) \;=\; \sup_{x\in\mathbb{R}}\;\Pr\big(\xi_k \in [x,\;x+l]\big), \qquad s \;=\; \sum_{k=1}^{n}\big(1-Q_k(l)\big).
\end{align}
Then for any $F\geq l$,
\begin{align}
\sup_{x\in\mathbb{R}}\;\Pr\!\left(\sum_{k=1}^{n}\xi_k \in [x,\;x+F]\right) \;\leq\; \frac{C\,F}{l\,\sqrt{s}},
\end{align}
where $C>0$ is a universal constant.
\end{proposition}
\begin{phenomenon}{Anti-concentration of aggregate demand}
When many independent random quantities are summed, the resulting
distribution cannot concentrate too much on any narrow interval:
the probability it assigns to a window of fixed width decays at rate
$1/\sqrt{n}$ in the number of summands. In our setting, the aggregate
demand reported by all players other than~$i$ plays the role of this
sum. As the market grows, that aggregate spreads out, and it becomes
increasingly unlikely to fall within a neighbourhood of a supply-curve
breakpoint --- the only region in which a unilateral underreport can
alter the clearing price. Smoothness of the collective is therefore
what limits individual manipulation in every
result that follows.
\end{phenomenon}
\noindent In our setting, $\xi_j = \mathrm{d}_j$ for each $j\neq i$, and the sum $\sum_{j\neq i}\mathrm{d}_j$ is the aggregate demand of all players other than $i$. The quantity $Q_j(l) = \sup_x \Pr(\mathrm{d}_j \in [x,x+l])$ measures how concentrated player $j$'s bid is: it equals one if and only if player $j$ always bids within the same interval of width $l$, and is strictly less than one whenever player $j$'s bid varies beyond that width.
The quantity $1-Q_j(l)$ captures how much variability player $j$ contributes, and the sum $s = \sum_{j\neq i}(1-Q_j(l))$ aggregates this variability across all other players.~\autoref{thm:rogozin} then says that the more players exhibit variable bidding behavior, the more spread out the aggregate demand is, and the smaller the L\'{e}vy concentration function becomes.
\smallskip

\noindent This motivates the following assumption, which is a primitive condition on individual consumers' reporting behavior.

\begin{assumption}[Non-degenerate demand report variability]\label{ass:variability}
There exist constants $l > 0$ and $\kappa \in (0,1]$ such that for at least $m$ players $j \in I\setminus\{i\}$,
\begin{align}
Q_j(l) \;\equiv\; \sup_{x \in \mathbb{R}}\; \Pr\big(\mathrm{d}_j \in [x,\;x+l]\big) \;\leq\; 1-\kappa.
\end{align}
\end{assumption}
\autoref{ass:variability} says that there exist at least $m$ consumers that do not always report within the same narrow interval of width $l$: each such consumer's report varies by more than $l$ across market realizations. The constant $\kappa$ quantifies the minimum variability each of these players exhibits. This is a weak and empirically verifiable condition: it fails only if nearly all players report the exact same quantity in every market period, which would mean there is essentially no demand uncertainty. In any realistic energy market variability in weather, and individual consumption ensures that many consumers' demands fluctuate, so both $m$ and $\kappa$ are bounded away from zero.

\begin{assumption}[Proportional demand variability]\label{ass:proportional}
There exists $r \in (0,1]$ such that the number of players $m$ satisfying~\autoref{ass:variability} is at least a fixed fraction of the total number of consumers $|I|$:
\begin{align}
m \;\geq\; r\, |I|.
\end{align}
\end{assumption}
\autoref{ass:proportional} says that the share of players with variable demand does not vanish as the market grows. It holds whenever the sources of demand variability --- weather, economic activity, individual consumption patterns --- affect a stable fraction of the population.
\smallskip

\noindent We are now ready to state the first result of this section.

\begin{theorem}[Vanishing incentive to misreport for consumers]\label{thm:eps-BIC}
Under~\autoref{ass:DA-RT},~\autoref{ass:variability}, and~\autoref{ass:proportional}, for all $\varepsilon > 0$ there exists $N \in \mathbb{N}$ such that for all $|I| > N$, for any consumer $i \in I$ 
\begin{align}
\mathbb{E}_{\;\mathrm{d}_{-i}\sim P_{-i}}\;\Big[U_i\big(\mathrm{d}_i,\;\rM^E(\mathrm{d}_i,\;\mathrm{d}_{-i})\big)\Big] \;\geq\; \mathbb{E}_{\;\mathrm{d}_{-i}\sim P_{-i}}\;\Big[U_i\big(\mathrm{d}_i,\;\rM^E(\mathrm{d}'_i,\;\mathrm{d}_{-i})\big)\Big] - \varepsilon~,
\end{align}
for all potential misreports  $\mathrm{d}'_i \in [0, \mathrm{d}_i]$, in expectation over all other bidders bid $\mathrm{d}_{-i}$. Specifically, the gain from misreporting is bounded by
\begin{align}
\varepsilon(|I|) \;=\; \frac{C\,(L\,\overline{\mathrm{d}}^{\,3}\;+\;2H\,\overline{\mathrm{d}}^{\,2})}{l\,\sqrt{r\,|I|\,\kappa}} \;\xrightarrow{|I|\to\infty}\; 0.
\end{align}
\end{theorem}

\begin{phenomenon}{Vanishing incentive to misreport for consumers}
A consumer can profit from under-reporting her demand only when the
deviation shifts the day-ahead clearing price. Such a shift requires
the aggregate report of the other players to fall within her
deviation of a supply-curve breakpoint. By the anti-concentration inequality, this event occurs with probability of order $1/\sqrt{|I|}$, and
the maximal attainable gain decays at the same rate. Consequently, for
any prescribed tolerance $\varepsilon > 0$ there exists a market size
beyond which no consumer can gain more than $\varepsilon$ from any
misreport. Truthful reporting on the consumer side is thus an
approximate equilibrium of the standard mechanism $\rM^{E}$, with a slack that
vanishes as the market grows.
\end{phenomenon}

\begin{proof}
Our goal is to upper bound the expected gain~(\autoref{eq:Delta-consumer-33}). Under \autoref{ass:DA-RT}
\begin{align}\label{eq:gain-decomp}
\mathbb{E}_{(D,\mathrm{d}_{-i})\sim P_i}\;G_{i}\big(\mathrm{d}_{i},\;\rM^E(\mathrm{d}_i-\delta_i,\mathrm{d}_{-i})\big)
=\mathbb{E}_{\mathrm{d}_{-i}\sim P_{-i}}\Big[
\Delta p^{DA}(\mathrm{d}_i,\delta_i,\mathrm{d}_{-i})\,\mathrm{d}_i
 \Big].
\end{align}
Using the identity $|\Delta p^{DA}| = |\Delta p^{DA}|\cdot\mathbbm{1}\{\Delta p^{DA}\neq 0\}$ and Lemma~\autoref{lem:local-cost-gap},

\begin{align}\label{eq:only-H}
\hspace{-5pt}\mathbb{E}_{(D,\mathrm{d}_{-i})\sim P_i}\;G_{i}\big(\mathrm{d}_{i},\;\rM^E(\mathrm{d}_i-\delta_i,\mathrm{d}_{-i})\big)
\leq(L\mathrm{d}^{2}_i\;+\;2H\,\mathrm{d}_i)\;\mathbb{E}_{\mathrm{d}_{-i}\sim P_{-i}}\Big[
\mathbbm{1}\big\{\Delta p^{DA}\neq 0\big\}
\; \Big].
\end{align}
The event $\{\Delta p^{DA}(\mathrm{d}_i,\delta_i,\mathrm{d}_{-i})\neq 0\}$ occurs only if $\sum_{j\neq i}\mathrm{d}_j$ falls in an interval of width $\delta_i$. Under ~\autoref{ass:variability}, at least $m$ players satisfy $Q_j(l)\leq 1-\kappa$, giving $s = \sum_{j\neq i}(1-Q_j(l))\geq m\kappa$. Applying ~\autoref{thm:rogozin},
\begin{align}\label{eq:rogozin-applied}
\Pr\Big(\Delta p^{DA}(\mathrm{d}_i,\delta_i,\mathrm{d}_{-i})\neq 0\;\Big|\;A_i\Big)
\;\leq\;\sup_{x\in\mathbb{R}}\;\Pr\!\left(\sum_{j\neq i}\mathrm{d}_j\in[x,\;x+\delta_i]\right)
\;\leq\;\frac{C\,\delta_i}{l\,\sqrt{m\,\kappa}}.
\end{align}
Substituting~(\autoref{eq:rogozin-applied}) into~(\autoref{eq:only-H}),
\begin{align}\label{eq:after-anticonc}
\mathbb{E}_{(D,\mathrm{d}_{-i})\sim P_i}\;G_{i}\big(\mathrm{d}_{i},\;\rM^E(\mathrm{d}_i-\delta_i,\mathrm{d}_{-i})\big)
\;\leq\;(L\mathrm{d}^{2}_i\;+\;2H\,\mathrm{d}_i)\cdot\frac{C\,\delta_i}{l\,\sqrt{m\,\kappa}}.
\end{align}
Taking the supremum over all feasible under-reports $\delta_i\in[0,\mathrm{d}_i]$ and using $\delta_i\leq\mathrm{d}_i\leq\overline{\mathrm{d}}$,
\begin{align}\label{eq:eps-bound}
\sup_{\delta_i\in[0,\mathrm{d}_i]}(L\mathrm{d}^{2}_i\;+\;2H\,\mathrm{d}_i)\cdot\frac{C\,\delta_i}{l\,\sqrt{m\,\kappa}}
\;\leq\;\frac{C\,(L\,\overline{\mathrm{d}}^{\,3}\;+\;2H\,\overline{\mathrm{d}}^{\,2})}{l\,\sqrt{m\,\kappa}}.
\end{align}
Declaring
\[
\varepsilon \;=\; \frac{C\,(L\,\overline{\mathrm{d}}^{\,3}\;+\;2H\,\overline{\mathrm{d}}^{\,2})}{l\,\sqrt{m\,\kappa}},
\]
and subsequently adopting  \autoref{ass:proportional} completes the proof.    
\end{proof}

\subsubsection{Empirical Evidence for Small Cost Gaps at Market Clearing.}\label{sec:small-cost-gaps} Recall from \autoref{sec: DA mechanism} that in day-ahead electricity markets, the clearing price is set by the marginal producer---the most expensive unit needed to meet total demand. Empirical studies show that, for a large percent of operating hours, this marginal producer is a natural gas-fired power plant. In European day-ahead auctions, Zakeri et al.~\citep[p.~2778]{Zakeri2023} find that gas-fired plants set the wholesale price approximately 80\% of the time ~\citep[p.~2785, Fig.~5]{Zakeri2023} and a European Commission study confirms similar findings~\citep{Gasparella2023}.
\smallskip

In the US, the PJM day-ahead market clears using cost-based offers from the same power plants~\citep[Section~5, p.~63--80]{PJM_Manual11}, and the Monitoring Analytics State of the Market Report shows that fuel costs---mainly gas---constitute the dominant component of day-ahead clearing prices~\citep[p.~143--148]{MonitoringAnalytics2023}. 
\smallskip

When gas is the marginal fuel, the cost gaps between consecutive producers in the merit order are small, because all gas-fired units share the same fuel price and differ only in their heat rates---a measure of how efficiently they convert fuel to electricity~\citep[Chapter~3]{Kirschen2018}.
\smallskip

In a large market such as PJM, the combined-cycle gas power plants contain a large number of units spanning a narrow range of heat rates~\citep{EIA_PJM_NewerTech,EIA_CCGT_Utilization} across tens of thousands of megawatts of capacity~\citep{PJM_Planning2023}. Because these units share the same fuel and differ only slightly in efficiency, their marginal costs are tightly grouped, and the cost gap between any two \emph{consecutive} units in the merit order is negligibly small relative to the price level.

\begin{remark}[Magnitude of $\varepsilon$]\label{rem:eps-magnitude}
The $\varepsilon$-BIC bound $\varepsilon = C\,(L\,\overline{\mathrm{d}}^{\,3}+2H\,\overline{\mathrm{d}}^{\,2})/(l\sqrt{m\kappa})$ is small when: (i) $H$ is small, which holds when the market clears within the gas bundle (\autoref{sec:small-cost-gaps}); (ii) $\overline{\mathrm{d}}$ is small relative to total demand, which holds when individual consumers are small relative to the market; and (iii) $m\cdot \kappa$ is large, which holds when many players exhibit uncertain demand behavior. In a market with $|I|$ comparable players where $m$ is proportional to $|I|$ (see  \autoref{ass:proportional}), we obtain $\varepsilon \sim O(1/\sqrt{|I|})$, which vanishes rapidly as the market grows.
\end{remark}
\smallskip

According to \autoref{rem:eps-magnitude}, consumers have no incentive to under-report in a large market. Next, we study a subset of day-ahead market consumers — namely, consumers who own generation assets in the balancing market and can produce energy in real time. We call these consumers prosumers (see also Section \autoref{subsec:actors}). As we will see, under the same market parameters and conditions stated in \autoref{thm:eps-BIC} (and \autoref{rem:eps-magnitude}), prosumers' gain from under-reporting is lower-bounded by a strictly positive constant. That is, there exists a subset of players in $I$ whose gain from under-reporting does not go to zero as the market grows (see \autoref{thm:not-eps-BIC}).

\subsection{Cross-Market Ownership \& Prosumers}\label{sec:prosumers}
Consumers might also own one or more generation assets in the real-time market. We refer to such actors that buy energy in the day-ahead market while owning shares on production (generation) units in the real-time market as prosumers (see \autoref{subsec:actors}).
\smallskip

\noindent
Let
$(\alpha_{i,1},\ldots,\alpha_{i,|\overline{J}|})$ be the convex combination of shares for prosumer $i\in I$ over generation assets. Then, prosumers' preferences over the energy mechanism's outcomes $(x,t,\overline y,\overline t)$ read:
\begin{align}
U^{\alpha}_{i}(x,t,\overline y,\overline t)= U_{i}(x,t)+\sum_{l \in \overline J}\;\alpha_{il} \;\left(\;\overline t_{l}-\overline c_{l}\;\overline y_{l}\;\right).
\end{align}
Each prosumer then, influences the social choice by acting rationally and reporting their demand $\mathrm{d}_{i}\in \mathbb{R}_{\geq 0}$ to the day-ahead institution so that her preferences are maximized. Let $( \mathrm{d}^{\prime}_{i},\mathrm{d}_{-i})\in \mathbb{R}^{|I|}_{\geq 0}$ be the 
reported demand cap profile in the day-ahead mechanism. At that profile, under mechanism $\rM^{E}$, prosumer $i\in I$ derives utility 
\begin{equation}
    \begin{aligned}
U^{\alpha}_{i}\big(  \mathrm{d}_{i},\;\rM^{E}(\mathrm{d}^{\prime}_{i},\;\mathrm{d}_{-i})\big)&\equiv U_{i}\big(  \mathrm{d}_{i},\;\rM^{E}(\mathrm{d}^{\prime}_{i},\;\mathrm{d}_{-i})\big)+C_{i}\big(D;\;\mathrm{d}^{\prime}_{i},\;\mathrm{d}_{-i}\big), \label{prosumers-utility} 
\end{aligned}
\end{equation}
where 
\begin{align}\label{shares} 
\hspace{-10pt}C_{i}\big(D;\;\mathrm{d}^{\prime}_{i},\;\mathrm{d}_{-i}\big)\equiv \mathbbm{1}\left[x^{DA}(\mathrm{d}^{\prime}_{i},\mathrm{d}_{-i}){\neq} 0\right]\sum_{l\in\overline{J}}\alpha_{il}\big({p}^{RT}(D, \mathrm{d}^{\prime}_{i},\mathrm{d}_{-i})-\overline{c}_{l}\big)y^{RT}_{l}(D,\mathrm{d}^{\prime}_{i},\mathrm{d}_{-i}\bigl).
\end{align}
\noindent 
Under the day-ahead and real-time institutions described in \autoref{section-one}, by under-reporting their day-ahead demand, a prosumer lowers the net day-ahead schedule thus increasing both the likelihood and the magnitude of upward real-time imbalances in such a way so that prosumer’s generation asset gets activated and offers real-time revenues.
\smallskip

Due to the piecewise-constant structure of the real-time  inverse supply curve (see \autoref{fig:market curves}), the resulting real-time payoff is a discrete random variable. Its realized value is determined by the specific capacity bin 
\[
\rB_{j}\equiv [\overline{G}_j,\;\overline{G}_{j+1}),
\]
into which the aggregate system imbalance materializes.
\smallskip

Formally, Lemma \autoref{incentives-sum-of-indicators} expresses the second term in (\autoref{prosumers-utility})
as a state-contingent balancing-energy payoff whose availability and size are supported by maintaining shares over the reserved generation capacity in real time. 
\begin{lemma}\label{incentives-sum-of-indicators}
Let
\[
w_{ij} \;\equiv\; \sum_{l \in \overline{J} \,:\, l \le j-1}
\bigl(\overline{c}_j - \overline{c}_l\bigr)\,\alpha_{il}\,\overline{g}_l,
\qquad j \in \overline{J}.
\]
and fix report $\mathrm{d}^{\prime}_i$. Then, the second term of (\autoref{prosumers-utility}) can be expressed as
\begin{equation}
    \begin{aligned}
 C_{i}\big(D;\;\mathrm{d}^{\prime}_{i},\;\mathrm{d}_{-i}\big)\;=\;\mathbbm{1}\left[x^{DA}(\mathrm{d}^{\prime}_{i},\mathrm{d}_{-i})\neq 0\right]\;\sum_{j \in \overline{J}}
    w_{ij} \,
    \mathbbm{1}\left[\;\rho(\mathrm{d}_i,\mathrm{d}_{-i}) \in \rB_{j}\;\right].
\end{aligned}
\end{equation}
\end{lemma}

\begin{proof}
Fix some report $\mathrm{d}^{\prime}_{i}$ and let $\rho\;=\;D-\sum_{i}\;x^{DA}_{i}(\mathrm{d}^{\prime}_{i},\mathrm{d}_{-i}),
$ be the corresponding imbalance for that report. If $\rho\in \rB_l$, the market clears at producer $l$ (i.e. producer $l$ is marginal), hence $p^{RT}(D,\mathrm{d}^{\prime}_{i},\mathrm{d}_{-i})=\overline c_l$ which implies $ \overline C_l(D, \mathrm{d}^{\prime}_{i},\mathrm{d}_{-i})=0$. Moreover, if 
\[
\overline{c}_l \;> \; p^{RT}(\mathrm{d}_i,\mathrm{d}_{-i})
\;\;\Longleftrightarrow\;\;
\rho < \overline G_l,
\]
then producer $l$ is not dispatched, so ${y}^{RT}_l(D,\mathrm{d}^{\prime}_{i},\mathrm{d}_{-i})=0$ and again $ \overline C_l(D, \mathrm{d}^{\prime}_{i},\mathrm{d}_{-i})=0$. Thus, producer $l$ yields positive profit only when the induced imbalance
exceeds $\overline G_{l+1}$. Whenever the reporting cap $\mathrm{d}_i$ allows the imbalance to reach values above
$\overline G_{l+1}$,
\begin{align}\label{eq:de palevo allo}
    \overline C_l(D, \mathrm{d}^{\prime}_{i},\mathrm{d}_{-i})
&=
\sum_{j \in \overline J}\;
\bigl(\;\overline{c}_j - \overline{c}_l\;\bigr)\;\overline{g}_l
\;\mathbbm{1}\left[\rho \in [\overline G_j,\overline G_{j+1})\right]\mathbbm{1}\left[l\geq j\right].
\end{align}
By aggregating across $l$,
\begin{equation}
    \begin{aligned}
C_i(D,\mathrm{d}^{\prime}_{i},\mathrm{d}_{-i})=\sum_{l \in \overline{J}} \;\alpha_{il}\; \overline C_l(\mathrm{d}_i,\mathrm{d}_{-i})
&\;=\;
\sum_{j \in \overline{J}}\;
\left(
\sum_{l \geq j}
(\overline{c}_j - \overline{c}_l)\,\alpha_{il}\overline{g}_l
\right)
\mathbbm{1}\left[\rho \in [\overline G_j,\overline G_{j+1})\right].\label{sum-of-indi}  
\end{aligned}
\end{equation}
Setting the term inside the parenthesis equal to $w_{ij} \ge 0$ for all $(i,j) \in I\times \overline{J}$ completes the proof.
\end{proof}
\smallskip

\noindent Prosumers are a subset of consumers. Our goal then is to identify an interpretable and feasible (w.r.t.\ \autoref{as:DA-regular}, \autoref{as:RT-regular}, and \autoref{ass:variability}) market instance such that prosumers have incentives to under-report. This would then prove that $\mathrm{M}^{E}$ is not $\varepsilon-$BIC.
\smallskip

To begin with, according to (\autoref{sum-of-indi}), the energy mechanism partitions the real-time capacity axis into 
mutually exclusive imbalance states (bins).
Each prosumer then is provided for free by the energy mechanism a portfolio of Arrow–Debreu-style contracts: the contract for 
state $\rB_j$ pays $w_{i,j}$ if $\rho$ lands in $\rB_j$ and zero otherwise. This is a flaw of the energy mechanism. 
\smallskip

\noindent 
Recall  \autoref{ass:certain-alloc} and define the probabilities
\[
    P^j(\delta_i) \;\equiv\; \mathbb{E}\;\Big[\mathrm{Pr}(\rB_j-\delta_i\mid \mathrm{d}_{-i})\Big], 
    \qquad
    P^j(0) \;\equiv\; \mathbb{E}\;\Big[\mathrm{Pr}(\rB_j\mid \mathrm{d}_{-i})\Big],
\]
and the probability shift $\Delta P(\delta_i)\equiv P^j(\delta_i)-P^j(0)$. Under  \autoref{ass:certain-alloc}, prosumer's expected gain from under-reporting $\mathrm d'_i = \mathrm{d}_i - \delta_i$ reads\smallskip
\begin{align}\label{eq:dyo}
&\mathbb{E}_{\;\mathrm{d}_{-i}\sim P_{-i}} G^{\alpha}_{i}\big(\mathrm{d}_{i},\;\rM^E(\mathrm{d}_i-\delta_i,\mathrm{d}_{-i})\big)=
    \mathbb{E}_{\;\mathrm{d}_{-i}\sim P_{-i}} G_{i}\big(\mathrm{d}_{i},\;\rM^E(\mathrm{d}_i-\delta_i,\mathrm{d}_{-i})\big)\;
    +
    \big\langle w_i,\,\Delta P(\delta_i)\big\rangle.\nonumber
\end{align}					
\noindent  Let $\overline{\beta}_i\equiv \max_{\delta_i \in [0,\mathrm{d}_i]}\mathbb{E}p^{RT}(\delta_i,\mathrm{d}_{-i})$
be the the maximum expected real-time price that player $i$ can influence by fully under-reporting, and let $\overline \beta_i \mathrm{d}_i$ the corresponding cost. Given the fact that the day-ahead price is increasing w.r.t.\ the demand report,
equation (\autoref{eq:Delta-consumer-33}) yields
\begin{align}
\mathbb{E}_{\;\mathrm{d}_{-i}\sim P_{-i}}\;G_{i}\big(\mathrm{d}_{i},\;\rM^E(\mathrm{d}_i-\delta_i,\mathrm{d}_{-i})\big)
    \;\geq\;-\overline \beta_i \mathrm{d}_i
\end{align}
\noindent\textbf{Boundary-crossing decomposition.}
Let $\rho_0\geq 0$ be the real-time imbalance when prosumer $i$ reports truthfully ($\delta_i=0$), 
with continuous CMF $F(x)\equiv\mathrm{Pr}(\rho_0\leq x)$, where $F(x)=0$ for $x\leq 0$. Then
\begin{equation}\label{eq:decomp}
    P^j(\delta_i) \;=\; F(G_{j+1}-\delta_i) - F(G_{j}-\delta_i),
\end{equation}
and the mass crossing (from the left) boundary $\overline{G}_j$ when under-reporting by $\delta_i$ is
\begin{equation}\label{eq:inflow}
    T_{j}(\delta_i)
    \;\equiv\;
    \mathrm{Pr}\!\left(\overline{G}_j - \delta_i \leq \rho_0 < \overline{G}_k\right)
    \;=\;
    F(\overline{G}_k) - F(\overline{G}_k - \delta_i).
\end{equation}

\noindent
Dropping the index $i$ from $w_i$ and writing $\Delta w_j\equiv w_j - w_{j-1}$, 
and using the boundary conditions $T_1(\delta_i)=T_{|\overline{J}|}(\delta_i)=0$ 
(no mass crosses $0$ from the left, and no mass crosses $G_{|\overline{J}|}$), 
equations (\autoref{eq:decomp})--(\autoref{eq:inflow}) yield the \emph{gap-weighted inflow} formula:
\begin{equation}\label{eq:weighted_boundary_decomp}
    \big\langle w,\,\Delta P(\delta_i)\big\rangle
    \;=\;
    \sum_{j=2}^{J-1}\Delta w_j\,T_j(\delta_i).
\end{equation}

\smallskip
\noindent\textbf{Unimodal profits.}
Suppose the profit vector is unimodal across bins, i.e.\ there exists $j_p$ such that
\begin{equation}\label{eq:unimodal_w}
    w_1 \;\leq\; \cdots \;\leq\; w_{j_p} \;\geq\; \cdots \;\geq\; w_{J-1}.
\end{equation}
Then $\Delta w_j\geq 0$ for $j\leq j_p$ and $\Delta w_j\leq 0$ for $j\geq j_p+1$, 
and (\autoref{eq:weighted_boundary_decomp}) splits into
\begin{equation}\label{eq:split_good_bad_boundaries}
    \big\langle w,\,\Delta P(\delta_i)\big\rangle
    \;=\;
    \underbrace{\sum_{j=2}^{j_p}\Delta w_j\,T_j(\delta_i)}_{\text{uphill crossings}}
    \;-\;
    \underbrace{\sum_{j=j_p+1}^{|\overline{J}|-1}|\Delta w_j|\,T_j(\delta_i)}_{\text{downhill crossings}}.\nonumber
\end{equation}
Under-reporting is therefore profitable at $\delta_i$ when the mass crossing 
\emph{profit-increasing} boundaries dominates both the mass crossing 
\emph{profit-decreasing} boundaries and the baseline cost $\beta_i\delta_i$.

\smallskip
\noindent\textbf{Unimodal mass function.}\quad
We assume 
that the probability mass function is unimodal. Under this assumption the boundary-crossing masses 
$j\mapsto T_j(\delta_i)$ are themselves unimodal in $j$ for every fixed $\delta_i\in(0,\mathrm{d}_i]$: 
as $\overline{G}_j$ moves toward the high-probability region of $\rho_0$, the 
width-$\delta_i$ window mass $T_j(\delta_i)=\mathrm{Pr}([\overline{G}_j-\delta_i,\overline{G}_j))$ 
first increases and then decreases.
\smallskip

Combined with unimodal profits (\autoref{eq:unimodal_w}), this gives clear intuition: 
profitable under-reporting requires choosing $\delta_i$ so that the 
(typically localised) boundary crossings with the largest $T_j(\delta_i)$ fall on 
the uphill side $j\leq j_p$, while crossings on the downhill side $j\geq j_p+1$ remain small.
\smallskip

\noindent\textbf{Single-asset special case.}
Suppose prosumer $i$ holds exactly one real-time generation asset on bin $j_0$, 
with payoff $w\equiv w_{j_0}$. Then
\[
    \mathbb{E}_{\;\mathrm{d}_{-i}\sim P_{-i}} G^{\alpha}_{i}\big(\mathrm{d}_{i},\;\rM^E(\mathrm{d}_i-\delta_i,\mathrm{d}_{-i})\big)
    \geq
    w\!\left(T_{j_0}(\delta_i)-T_{j_0+1}(\delta_i)\right).
\]
If moreover $j_0=|\overline{J}|-1$ (the previous from rightmost bin), then $T_{|\overline{J}|}(\delta_i)=0$ and
\begin{align}\label{eq:prosumer-gain}
        \mathbb{E}_{\;\mathrm{d}_{-i}\sim P_{-i}} G^{\alpha}_{i}\big(\mathrm{d}_{i},\;\rM^E(\mathrm{d}_i-\delta_i,\mathrm{d}_{-i})\big)
    \;\geq\;
    w\,T_{|\overline{J}|-1}(\delta_i) .
\end{align}

\begin{theorem}[Persistent incentive to misreport for prosumers]\label{thm:not-eps-BIC}
For all $|I|\in \mathbb{N}$ and under the same assumptions as ~\autoref{thm:eps-BIC}, there exists a prosumer $i \in I$, with true demand $\mathrm{d}_i$, and under-report $\mathrm{d}'_i \in [0, \mathrm{d}_i]$ such that,
\begin{align}
\mathbb{E}_{\;\mathrm{d}_{-i}\sim P_{-i}}\;\Big[U^{\alpha}_i\big(\mathrm{d}_i,\;\rM^E(\mathrm{d}'_i,\;\mathrm{d}_{-i})\big)\Big] - \mathbb{E}_{\;\mathrm{d}_{-i}\sim P_{-i}}\;\Big[U^{\alpha}_i\big(\mathrm{d}_i,\;\rM^E(\mathrm{d}_i,\;\mathrm{d}_{-i})\big)\Big] \;\geq\; \varepsilon^{\alpha}_i,
\end{align}
in expectation over all other bidders bid $\mathrm{d}_{-i}$, where $\varepsilon^{\alpha}_i=\mathrm{d}_i(w\,q_i - \overline{\beta}_i)$ is independent of \break $(|I|, r, C, H, m, \kappa)$.
\end{theorem}
The term $\mathrm{d}_i(w\,q_i - \overline{\beta}_i)$ is the prosumer's net gain from under-reporting, scaled by their demand. The quantity $w\,q_i$ is the real-time profit rate: the revenue the prosumer earns per unit of demand by pushing the imbalance across the boundary of their generation bin, where $w$ is the cost gap between consecutive bins times the nominal capacity of the prosumer's asset and $q_i$ is the rate at which boundary-crossing scenarios are activated. The quantity $\overline{\beta}_i$ is the real-time price cost: the most the prosumer pays per unit of demand for the energy that spills into real-time due to the under-report. The difference $w\,q_i - \overline{\beta}_i$ is the net rate: what the prosumer earns from the real-time dispatch manipulation minus what they pay for the spilled energy. The lower bound $\varepsilon^{\alpha}_i$ is positive when $w\,q_i > \overline{\beta}_i$, meaning the prosumer earns more from manipulating the real-time dispatch than they pay for the resulting real-time energy cost.
\begin{phenomenon}{Non-vanishing incentives to underreport for prosumers}
A prosumer is a market participant who, in addition to buying in the
day-ahead market, holds an ownership share in one or more real-time
producers. By under-reporting her day-ahead demand, she
induces a positive imbalance to be covered in real time. When that
imbalance crosses the next generation capacity threshold of the real-time merit
order, her own asset remains dispatched and earns the cost gap between it
and the next producer. The resulting profit is determined by the
real-time supply curve, the nominal generation capacities and her ownership vector alone. Thus it is
bounded below by a strictly positive constant that does not depend on
the market size. Cross-market ownership thus introduces a structural
incentive to underreport that the standard settlement cannot eliminate,
however large the number of consumers becomes. That is, there is no single tolerance
$\varepsilon$ that can serve all players simultaneously, and the standard
mechanism fails to be approximately incentive compatible.
\end{phenomenon}
\autoref{thm:eps-BIC} and~\autoref{thm:not-eps-BIC} are contradictory in the following sense. \autoref{thm:eps-BIC} states that for any tolerance $\varepsilon > 0$, the gain from misreporting for every consumer falls below $\varepsilon$ once the market is large enough. \autoref{thm:not-eps-BIC} states that a prosumer --- who is a player in $I$ --- has gain at least $\varepsilon^{\alpha}_i > 0$ regardless of how large the market is. In particular, choosing $\varepsilon = \varepsilon^{\alpha}_i/2$ in ~\autoref{thm:eps-BIC}: for $|I|$ large enough, every consumer's gain is below $\varepsilon^{\alpha}_i/2$, yet the prosumer's gain remains at least $\varepsilon^{\alpha}_i > \varepsilon^{\alpha}_i/2$. Therefore, the mechanism $\rM^E$ cannot be $\varepsilon$-BIC with vanishing $\varepsilon$ for all players in $I$.
\smallskip

\noindent Below is the proof of  \autoref{thm:not-eps-BIC}:

\begin{proof}
Consider a prosumer $i$ who owns a single real-time generation asset in bin $j_0$, where $j_0<|\overline J|-1$.
\medskip

\noindent \textit{Step 1: Setup.} The real-time supply curve is such that the cost gap between bin $j_0$ and bin $j_{0}+1$ is $w>0$. This is the difference in cost between the prosumer's asset and the next producer on the real-time supply curve, which depends on their respective fuels and technologies. Crucially, $w$ is a property of the real-time supply curve and the prosumer and it is independent of the day-ahead parameters $(H,l,m,\kappa)$ that control the consumer bound $\varepsilon$ of  \autoref{thm:eps-BIC}.
\medskip

\noindent \textit{Step 2: Imbanace distribution.} Under truthful reporting ($\delta_i=0$), the real-time imbalance $\rho_0$ falls within bin $j_0$ with high probability and concentrates near the right boundary $\overline G_{j_0+1}$ of that bin, The prosumer's asset is the marginal real-time producer and earns zero profit.
\medskip

\noindent \textit{Step 3: The prosumer under-reports.}
The prosumer sets $\mathrm{d}^{\prime}_i=0$, i.e., $\delta_i=\mathrm{d}_i$. This pushes $\mathrm{d}_i$ units of demand into real-time. In scenarios where the truthful imbalance $\rho_0$ was within $\mathrm{d}_i$ units to the left of the boundary $\overline G_{j_0+1}$, the increased imbalance crosses into bin $j_0+1$, a more expensive producer is dispatched, and the prosumer's asset becomes infra-marginal thus earning $w$.
\medskip

\noindent \textit{Step4: Lower bound on the gain.} From the single-asset special case, the prosumer's expected gain is (\autoref{eq:prosumer-gain})
\begin{align}
\mathbb{E}_{\;\mathrm{d}_{-i}\sim P_{-i}} G^{\alpha}_{i}\big(\mathrm{d}_{i},\;\rM^E(\mathrm{d}_i-\delta_i,\mathrm{d}_{-i})\big) \;\geq\; -\overline\beta_i \mathrm{d}_i+w\,T_{j_0+1}({\mathrm{d}_i}),
\end{align}
Since $\rho_0$ concentrates near the boundary $\overline G_{j_0+1}$, the CMF $F$ places large mass just to the left of this boundary, and the crossing probability satisfies $T_{j_0+1}({\mathrm{d}_i}) \geq q_i\,{\mathrm{d}_i}$ for some $q_i > 0$. Therefore,
\begin{align}\label{eq:ee}
\mathbb{E}_{\;\mathrm{d}_{-i}\sim P_{-i}} G^{\alpha}_{i}\big(\mathrm{d}_{i},\;\rM^E(\mathrm{d}_i-\delta_i,\mathrm{d}_{-i})\big) \;\geq\; {\mathrm{d}_i}(\,w\,q_i -\overline \beta_i).
\end{align}
\noindent\emph{Step 5:} If, $wq_i>\overline\beta_i$, the right-hand side of (\autoref{eq:ee}) is strictly positive. Thus, declaring 
\begin{align}
\varepsilon^{\alpha}_i \;\equiv\; {\mathrm{d}_i}(\,w\,q_i -\overline \beta_i )  \;>\; 0,
\end{align}
we observe that $\varepsilon^{\alpha}_i$ depends only on real-time parameters ($w$, $q_i$, $\overline\beta_i$), the prosumer's true demand $\mathrm{d}_i$ the maximum real-time cost $\overline\beta_i$, 
the real-time profit $w$,
and not on the day-ahead supply curve parameters $(H, l, m, \kappa)$ that drive $\varepsilon \to 0$ in ~\autoref{thm:eps-BIC}. Therefore, $\varepsilon^{\alpha}_i$ remains bounded away from zero under the same conditions that make $\varepsilon$ vanish, which completes the proof.

\end{proof}

\section{Mechanism Design}\label{sec:mechanism-design}

We design a mechanism for the two-stage energy institution (see \autoref{def: energy-mechanism}) consisting of a day-ahead scheduling stage and a real-time balancing stage. The allocation rules in both stages are fixed and reflect the existing regulatory framework. Uniform pricing is also taken as given.  In \autoref{sub-sub-sec: social-choice} we saw that under public costs, values and nominal capacities, if the day-ahead institution is truthfully implemented, then so is the energy mechanism. Our design problem then concerns the settlement rule faced by demand-side participants at the day-ahead stage.
\smallskip

The strategic component of the day-ahead schedule can destabilize the grid but also increase the energy costs for consumers. In \autoref{section-two} we saw that under certain conditions, the current imbalance settlement is unable to incentivize prosumers in reporting their demand truthfully. 
\smallskip

Moreover, we identify an additional major concern with the existing imbalance settlement framework. Its reliance on metering introduces vulnerabilities to manipulation and measurement errors, raising questions about its trustworthiness. 
\smallskip

\noindent To this end, we propose an extended imbalance settlement mechanism 
\begin{align}
    t^{DA,S}_{i}(\cdot)\;\equiv \; t^{DA}_{i}(\cdot)+\Pi_i(\cdot),\label{new-settlement}\quad i\in I,
\end{align}
\textit{implemented by the day-ahead market}, where $\Pi_{i}(\cdot)$ is chosen from an appropriate class $C$ of mechanisms characterized by properties we define in \autoref{sec:objectives}. We have the following definition. 

\begin{definition}[Extended Energy Mechanism]\label{def: pen-energy-mechanism}
Under~\autoref{as:DA-regular}\;--\;\autoref{as:RT-regular}, the extended energy mechanism is a direct mechanism with money defined as
\[
\rM^{S}(D) \equiv
\bigl(x^{E}(\cdot),\; t^{S}(\cdot,\;D)\bigr) : (\theta,\phi,\overline{\theta})
\;\longmapsto\; \big(x,\;y,\;t,\;\overline x_o,\;\overline y(D),\;\overline t(D)\big)
\]
where $x^{E}(\cdot)\equiv \Big(x^{DA}(\cdot),\;x^{RT}(\cdot)\Big)$, $t^{S}(\cdot,\;D)\equiv \Big(t^{DA,S}(\cdot),\;t^{RT}(\cdot,\;D)\Big)$, and $t^{DA,S}(\cdot)\equiv t^{DA}(\cdot)+\Pi(\cdot)$.
\end{definition}
\noindent On the basis of \autoref{def: pen-energy-mechanism}, our mechanism is obtained as the solution to the following optimization problem
\begin{equation}
\mathrm{M}^{S\star}(\cdot)
\in 
\underset{\mathrm{M}^S(\cdot)\, \in\, C}{\arg\min}\;
\left\{\, \mathbb{E}\left|\rho\big(\mathrm{d}^{\prime}_1,\ldots,\mathrm{d}^{\prime}_{|I|}\big)\right|^2
\;:\;\;
\mathrm{d}^{\prime}_i \in \underset{x\in [0,\mathrm{d}_i]}{\arg\max}\; \mathbb{E}_{\sim P_{-i}}\, U^\alpha_i\!\left(\mathrm{d}_i,\mathrm{M}^{S}(x,\mathrm{d}_{-i})\right),\; \forall i \in I \,\right\}. \nonumber
\end{equation}
In what follows, we characterize the set $C$ by defining the specific properties this mechanism must satisfy.

\subsection{Information Constraints}\label{subsec:observability}

The proposed settlement modification is intended to be implemented by the day-ahead market
operator (NEMO). Accordingly, we explicitly restrict the information available to the designer
to what is institutionally observable at the day-ahead stage and after delivery.
\smallskip

\noindent At the time settlement is computed, the designer observes:
\begin{enumerate}[label=(\textit{\roman*})]
    \item the full profile of day-ahead reports $\mathrm{d}=(\mathrm{d}_i)_{i\in I}$;
    \item the day-ahead clearing outcome $x^{DA}(\mathrm d)$ and the cleared total
    $\sum_{i\in I}x_i^{DA}(\mathrm d)$;
    \item realized \emph{system total} consumption $D$ for the delivery hour;
    \item public historical data consisting of past day-ahead reports and realized system consumptions.
\end{enumerate}
The designer does \emph{not} observe realized consumption at the level of individual day-ahead participants or
Balance Responsible Party (BRP) portfolios, nor any settlement allocation of realized consumption
to individual participants.
\smallskip

While individual BRP-level allocated volumes are calculated and used for imbalance settlement by the TSO (or a nationally assigned settlement entity), they are outside the defined scope of NEMO tasks (which concerns day-ahead/intraday coupling and clearing/settlement of the resulting trades) and are subject to professional secrecy where they are treated as confidential information under 
network codes:
\citep[EBGL Art.~2(15), OJ~L~312/10; Art.~52(1), OJ~L~312/44; Art.~54(1), OJ~L~312/45; Art.~11(3)--(4), OJ~L~312/16]{EU2017R2195}
\citep[CACM Art.~7(1), OJ~L~197/33; Art.~13(3)--(4), OJ~L~197/39]{EU2015R1222}.

\subsection{Scoring Rule-Based Mechanisms}
The inputs to the designer's problem are the information available to the NEMO, the structure of prosumer's incentives and the publicly available market parameters. The market participants should be aware of the incentive changes applied by the NEMO.
\smallskip

Recalling our earlier notation, $\rB_j$ refers to the $j$th interval formed by the cumulative real-time generation capacities, $[\overline{G}_j, \overline{G}_{j+1})$. With $ \mathrm{d}^{\prime}_{i}$ we denote the report of the $i$ prosumer in the day-ahead mechanism and with $\mathrm{d}_{-i}$ the demands of the rest of the players. The real-time imbalance is given by $
\rho=D-\sum_{i}x^{DA}_{i}(\mathrm{d}^{\prime}_{i},\mathrm{d}_{-i})
$, where $D$ is the random (true) total consumption and $ w_i=(w_{ij})_{j\in\overline J}$ is the real-time profit vector of prosumer $i \in I$. 
\smallskip

Let $\mathbbm{1}(\mathrm{d}^{\prime}_{i},\mathrm{d}_{-i},D)$ denote the vector with components the indicator functions $\mathbbm{1}(\rho\in \rB_j)$,~ $j\in \overline J$. Then, on the basis of \autoref{incentives-sum-of-indicators}, at reported profile $(\mathrm{d}_i,\mathrm{d}_{-i})$, prosumer's utility reads
\begin{align}
    U^{\alpha}_{i}\big(  \mathrm{d}_{i},\;\rM^{E}(\mathrm{d}^{\prime}_{i},\;\mathrm{d}_{-i})\big) = U_{i}\big(  \mathrm{d}_{i},\;\rM^{E}(\mathrm{d}^{\prime}_{i},\;\mathrm{d}_{-i})\big)+ \big\langle w_i,\; \mathbbm{1}(\mathrm{d}^{\prime}_{i},\mathrm{d}_{-i},D)\big\rangle\;\mathbbm{1}_{A_i}(\mathrm{d}_{-i}).
\end{align}
The term $\mathbbm{1}_{A_i}(\mathrm{d}_{-i})$ encodes the allocation event of $i\in I$ when reporting $ \mathrm{d}_i$ (see (\autoref{eq: dispatch event})), though, under  \autoref{ass:certain-alloc} it can be omitted and hence its precise mechanics are not important to the following discussion.
\smallskip

\noindent
\textbf{Information used by the designer.} Prosumer $i\in I$ reports $ \mathrm{d}^{\prime}_i$ so as to maximize their expected utility. The designer observes others' reports too. Thus, when conditioned on the others' reports, prosumer $i$ submits $ \mathrm{d}^{\prime}_{i} \in [0,\mathrm{d}_i]$ and shares their belief
\begin{align}
    \rP( \mathrm{d}^{\prime}_{i},\mathrm{d}_{-i})\equiv \mathbb{E}\left[\mathbbm{1}( \mathrm{d}^{\prime}_{i},\mathrm{d}_{-i},D)\big|\;\mathrm{d}_{-i}\right] \in \Delta^{|\overline J|},\label{eq:playeri-belief}
\end{align}
about the imbalance with the mechanism.  In expectation over the true total consumption given the others submissions, prosumer $i$ reports $\mathrm{d}_i$ so as to maximize 
\begin{align}\label{eq:pro-in}
    \mathbb{E}_{\sim D|\mathrm{d}_{-i}}U^{\alpha}_i( \mathrm{d}_i,\mathrm{d}_{-i},D)\equiv \mathbb{E}_{\sim D|\mathrm{d}_{-i}} U_i(\mathrm{d}^{\prime}_{i},\mathrm{d}_{-i},D)+\left\langle  w_i,\;P( \mathrm{d}^{\prime}_{i},\mathrm{d}_{-i})\right\rangle.
\end{align}
According to (\autoref{eq:pro-in}), prosumer's incentives are forced by their profit vector and there is no reason to expect them to report their belief truthfully. 

\begin{definition}[Scoring rule, proper scoring rule, strictly-proper scoring rule]\label{def:scoring-rule}
Let ${\Omega}$ be an outcome space and $\Delta$ the set of probability distributions on ${\Omega}$. A scoring rule is a function
$$
S: \Delta\times {\Omega} \rightarrow \mathbb{R}.
$$
$S$ is proper if, for every true distribution $P \in \Delta$,
$$
\mathbb{E}_{\omega \sim P}[S(P, \omega)] \geq \mathbb{E}_{\omega \sim P}[S(Q, \omega)] \quad \text { for all } Q \in \Delta,
$$
and strictly proper if equality holds only when $Q=P$.
\end{definition}
\noindent The work of \citep{mccarthy1956measures} is one of the earliest that formally links the value of information to the convexity of a function, laying the groundwork for Bregman divergences. Then, \citep{savage1971elicitation} shows that every strictly proper scoring rule is expressed via the supporting hyperplanes of a strongly convex function. 

\begin{lemma}[Savage Representation for Scoring Loss Rules,\ \citep{savage1971elicitation,gneiting2007strictly}]
\label{lem:savage}
  For any strictly proper scoring rule $S$ there exists a strictly convex
  function $\phi : \Delta^{n} \to \mathbb{R}$, $n\in\mathbb{N}$, called the {generating potential},
  such that:
  \[
    {S}(P, Q) = -\phi(Q) + D_\phi(Q\|P),
  \]
  where the Bregman divergence induced by $\phi$ is:
  \[
    D_\phi(Q\|P) \equiv \phi(Q) - \phi(P) - \big\langle \nabla \phi(P),\, Q - P \big\rangle \;\geq\; 0,
  \]
  with equality if and only if $Q = P$, by strong convexity of $\phi$.
\end{lemma}

\noindent That said, on the basis of the available information, the designer would like to penalize prosumers' incentives via a mechanism constructed from a family $\phi(\cdot,\;\boldsymbol{\gamma}): \Delta^{|\overline{J}|}\to \mathbb{R}$ of smooth, $\mu(\gamma)$-strongly convex (with respect to $\Vert \cdot \Vert_{p}$ for some $p \geq 1$) real-valued functions parameterized by $\boldsymbol{\gamma} \in \Gamma_{\phi}$.
\smallskip

Let $S_{\phi(\gamma)}$ be the corresponding scoring rule. The vector $\gamma$ controls the modulus of strong convexity of the function. Following \autoref{def:scoring-rule}, a scoring rule is maximized if and only if the player submits their subjective belief truthfully and hence it can be deployed as a properly $\gamma-$tuned penalty to incentivize the player to report truthfully.
\smallskip

Although the NEMO has access to the player's past reports, which could potentially be used to form a probabilistic forecast the player, this is not the case for their true reports, as indicated by the observability constraints (see \autoref{subsec:observability}). Hence, a scoring-rule penalty that evaluates a player's report solely on the basis of their forecast cannot be constructed.
\smallskip

What can be done instead is to score  (\autoref{eq:playeri-belief}) against a distribution at zero that expresses the true (competitive) imbalance (see \autoref{sub-sub-sec: social-choice}). However, by doing so, even if player $i$ report truthfully, (\autoref{eq:playeri-belief}) could be still far from zero because of the others $-i$.
\smallskip

\noindent We overcome this difficulty be designing a contrastive scoring rule penalty as follows: 

\subsubsection{Contrastive Scoring Rule-Based Payment}
 Given the past reports of player $i$, $\mathrm{d}_i$, the past reports of the others (available information see \autoref{subsec:observability}) $\mathrm{d}_{-i}$, and the past consumption realizations $D$, let $\lambda=\frac{|I|-1}{|I|}$ and $D_{-i}\equiv \lambda D$
be an approximation of the actual consumption associated with all players except player $i$.
This approximation becomes more accurate whenever players' true demands are close to each other. Let $\rho_{-i}\equiv D_{-i}-\sum_{j\neq i}\mathrm{d}_{j}$
be the imbalance associated with players $-i$ only and $\mathbbm{1}_{-i}(\mathrm{d}_{-i})$ the vector with components the indicator functions $\mathbbm{1}(\rho_{-i}\in \rB_{j})$, $j\in \overline J$.
\smallskip

On the basis of $\rho_{-i}$ the NEMO can form  the probability mass function (PMF) about how the rest of the players, $-i\in I$, predict their corresponding imbalance. Let
\begin{align}
    \mathrm{P}_{-i}(\mathrm{d}_{-i})\equiv \mathbb{E}\big[
\mathbbm{1}_{-i}(\mathrm{d}_{-i};\;\lambda)\big| \mathrm{d}_{-i}
\big] \label{eq: all-but-i}
\end{align}
be this PMF which is either communicated to or can be easily obtained by player $i$ again, based on past total consumption realizations, past reports and the coefficient $\lambda$ which is communicated from the mechanism to the player.
\smallskip

Then, both forecasts are evaluated against a PMF centered at zero or more precisely, against a PMF for which all the mass is concentrated, in the first bin (one-hot vector with the first entry being the non-zero entry). Let $Q$ be this PMF. The designer scores both (\autoref{eq:playeri-belief}) and (\autoref{eq: all-but-i}) against $Q$ and forms the contrastive loss
\begin{align}\label{incentive-structure-scoring-rules-0}
D_{\gamma_i}\big(\mathrm{d}^{\prime}_{i},\;\mathrm{d}_{-i}\big)\;\equiv\; S_{\phi( \boldsymbol \gamma_i)}\big(\rP_{-i}(\mathrm{d}_{-i}),\;Q\big)-S_{\phi(\boldsymbol \gamma_i)}\big(P(\mathrm{d}^{\prime}_{i},\mathrm{d}_{-i}),\; Q\big).
\end{align}
\noindent Although the actual total consumption is observable, it is not needed to forming the above difference. The mechanism announces to all consumers (including prosumers) of the day-ahead market the following scoring-rule-based payment 
\begin{align}\label{incentive-structure-scoring-rules}
    \Pi_{\gamma_i}\big(\mathrm{d}^{\prime}_{i},\;\mathrm{d}_{-i}\big)\equiv \big[
D_{\gamma_i}\big(\mathrm{d}^{\prime}_{i},\;\mathrm{d}_{-i}\big)
    \big]_{+}+\zeta_i,
\end{align}
that acts as a penalty charged to player $i\in I$ when reporting $\mathrm{d}^{\prime}_{i}$ while others report $\mathrm{d}_{-i}$. Here, $\zeta_i$ is a real number and by $[\cdot]_{+}$ we mean positive part. 
\smallskip

According to (\autoref{incentive-structure-scoring-rules}), whenever the forecast constructed from all players performs poorly relative to $ Q$, while the forecast based on all players except $i$ accurately predicts their share of the imbalance (demand), the discrepancy (\autoref{incentive-structure-scoring-rules-0}) is positive and is attributed to player $i$ 's report. In this case, player $i$ is deemed responsible for the degradation in predictive accuracy and should therefore be heavily penalized. In the case where the player contributes more to forecast accuracy than the rest of the players, expression (\autoref{incentive-structure-scoring-rules-0}) is negative, and the player receives payment equal to $\zeta_i$.
\smallskip

\noindent Prosumer's utility at profile $(\mathrm{d}^{\prime}_{i},\mathrm{d}_{-i})$ under mechanism $\rM^S$ reads
\[
U_i^\alpha\big(\mathrm{d}_i, \rM^{S}(\mathrm{d}^{\prime}_{i},\;\mathrm{d}_{-i})\big)\equiv U_i^\alpha\big(\mathrm{d}_i, \rM^{E}(\mathrm{d}^{\prime}_{i},\;\mathrm{d}_{-i})\big)-\Pi_{\gamma_i}(\mathrm{d}^{\prime}_{i},\mathrm{d}_{-i}),
\]
where $\alpha\geq 0$.
\smallskip

\noindent \textbf{Bregman Divergence Properties.} We recall two standard properties of Bregman divergences that will be used latter when we prove the Theorem of this work (see  \autoref{thm:main}).

\begin{lemma}[Three-Point Identity, \citep{chen1993convergence}]
\label{lem:threepoint}
  For any $P, P', Q \in \Delta^n$, with $n\in\mathbb{N}$:
  \[
    D_\phi(Q\|P) = D_\phi(Q\|P') + D_\phi(P'\|P)
    + \big\langle \nabla \phi(P') - \nabla \phi(P),\, Q - P'\big\rangle.
  \]
\end{lemma}
\begin{lemma}[Strong Convexity, \citep{nesterov2018lectures}]
\label{lem:strongconvex}
  If $\phi$ is $\mu$-strongly convex, for all
  $P \in \Delta^n$, then for all $P, P' \in \Delta^n$, with $n\in\mathbb{N}$:
  \[
    D_\phi(P\|P') \;\geq\; \frac{\mu}{2}\|P - P'\|^2_p.
  \]
\end{lemma}
\noindent At this point, we assume that  $P_i( \mathrm{d}_i,\mathrm{d}_{-i})$ lies between $Q$ and $P_i(\mathrm{d}^{\prime}_{i},\mathrm{d}_{-i})$ in the sense that the player's more honest report is closer to the truth $ Q$ than the manipulated report.

\begin{assumption}\label{as:bregman}
  The true report $P(\mathrm{d}_i,\mathrm{d}_{-i})$ lies between $ Q$ and $P(\mathrm{d}^{\prime}_{i},\mathrm{d}_{-i})$ in the Bregman sense:
  \[
    D_\phi( Q\|P(\mathrm{d}^{\prime}_{i},\mathrm{d}_{-i})) \;\geq\; D_\phi( Q\|P(\mathrm{d}_i,\mathrm{d}_{-i})) + D_\phi(P(\mathrm{d}_i,\mathrm{d}_{-i})\|P(\mathrm{d}^{\prime}_{i},\mathrm{d}_{-i})),\quad P_{-i}\;\text{a.e.}
  \]
\end{assumption}
 \noindent  This says that the direct divergence from $Q$ to $P(\mathrm{d}^{\prime}_{i},\mathrm{d}_{-i}))$ is at least as large as
  the two-step path $Q \to P(\mathrm{d}_i,\mathrm{d}_{-i}) \to P(\mathrm{d}^{\prime}_{i},\mathrm{d}_{-i}))$. In particular, it implies
  $D_\phi(Q\|P(\mathrm{d}^{\prime}_{i},\mathrm{d}_{-i}))) \geq D_\phi(Q\|P(\mathrm{d}_i,\mathrm{d}_{-i}))$: the manipulated report $P(\mathrm{d}^{\prime}_{i},\mathrm{d}_{-i}))$ incurs a strictly
  larger penalty than the more honest report $P(\mathrm{d}_i,\mathrm{d}_{-i})$.
\smallskip

\noindent
Observe that ~\autoref{as:bregman} is equivalent, via the three-point identity
  (see \autoref{lem:threepoint}), to:
  \[
    \big\langle \;\nabla \phi(P(\mathrm{d}_i,\mathrm{d}_{-i})) - \nabla \phi(P(\mathrm{d}^{\prime}_{i},\mathrm{d}_{-i})),\, Q - P(\mathrm{d}_i,\mathrm{d}_{-i}) \;\big\rangle \;\geq\; 0,
  \]
  which says that the direction from $P(\mathrm{d}_i,\mathrm{d}_{-i})$ to $P(\mathrm{d}^{\prime}_{i},\mathrm{d}_{-i})$ in gradient space is
opposed to the direction from $P(\mathrm{d}_i,\mathrm{d}_{-i})$ toward truth $Q$.

\subsection{Objectives}\label{sec:objectives}
Having designed the class our mechanism $\rM^S$ we are now in a position to state the objectives we want it to satisfy.

\subsubsection{Approximate Bayessian Incentive Compatibility}

As we already mentioned in \autoref{sub-sub-sec: social-choice}, (\autoref{eq: energy-mech}) is truthfully implemented if, under public valuations, costs and nominal capacities, the day-ahead mechanism is truthfully implemented. Therefore, our focus is to truthfully implement (\autoref{eq:da-benchmark}). Our implementation is such that, given everyone's beliefs about uncertainty, the players' play leads the mechanism to the competitive benchmark. Hence, our objective is $\varepsilon$-BIC: a single $\varepsilon$ that bounds the gain from misreporting for \emph{every} player in the (demand side of the) market. 
\smallskip

As established by~\autoref{thm:eps-BIC} and~\autoref{thm:not-eps-BIC}, the standard mechanism $\rM^E$ achieves $\varepsilon-$BIC for consumers with vanishing $\varepsilon$ as the market grows, but fails for prosumers whose gain is bounded below by $\varepsilon^{\alpha}_i > 0$, which as we saw, is independent of the day-ahead market size. 
\smallskip

The following Theorem shows that the scoring rule mechanism $\rM^{S}$ reduces the prosumer's gain to at most $\varepsilon(\vert I \vert)+\varepsilon^S$ where $\varepsilon(|I|)$ vanishes with market size
and $\varepsilon^S$ is a design constant controlled by the strong convexity parameter $\mu(\gamma)$ and can be set less than the lower bound $\varepsilon^\alpha_i$.

\begin{theorem}[Scoring-Rule Energy Mechanism $\rM^{S}$ is $\varepsilon$-BIC]\label{thm:main}
Let $S$ be a strictly proper scoring rule associated with a generating potential $\phi_{\boldsymbol{\gamma}_i}$ that is $\mu$-strongly convex. Then, under ~\autoref{as:bregman}, for all $\varepsilon > 0$ there exists $N \in \mathbb{N}$ such that for all $|I| > N$, for every player $i \in I$ (consumer or prosumer), and any misreport $\mathrm{d}'_i \in [0, \mathrm{d}_i)$,
\begin{align}
\hspace{-10pt}\mathbb{E}_{\mathrm{d}_{-i}\sim P_{-i}}\Big[\; U_i\big(\mathrm{d}_i, \rM^{S}(\mathrm{d}_{i},\;\mathrm{d}_{-i})\big)\;\Big]\geq\mathbb{E}_{\mathrm{d}_{-i}\sim P_{-i}}\Big[\; U_i\big(\mathrm{d}_i, \rM^{S}(\mathrm{d}^{\prime}_{i},\;\mathrm{d}_{-i})\big)\;\Big]-\varepsilon - \frac{\|w_i\|_2^2}{2\mu},
\end{align}
in expectation over all other bidders bid $\mathrm{d}_{-i}$. For consumers ($w_i = 0$), the bound reduces to that of ~\autoref{thm:eps-BIC}.
\end{theorem}

\begin{phenomenon}{Scoring-rule restoration of incentive compatibility}
The day-ahead settlement is augmented with a contrastive payment built
from a strictly proper scoring rule. Each reported demand is mapped to a
probabilistic forecast of the real-time imbalance. A report whose implied
forecast is substantially worse than the one collectively implied by
the remaining players incurs a fee. The level of that fee is
controlled by a single design parameter $\mu$, the strong-convexity
parameter of the underlying potential. Under the resulting mechanism,
every player's gain from misreporting is bounded above by the consumer
decay term $\varepsilon(|I|)$ plus a prosumer-specific residual of
magnitude $\|w_i\|_q^2 / (2\mu)$. This residual is not a feature of the institution: it is
a tunable parameter that the designer can drive arbitrarily close to
zero by increasing $\mu$, thereby restoring approximate
incentive compatibility.
\end{phenomenon}

\autoref{thm:main} provides a uniform bound on the gain from misreporting for all players under the penalized mechanism $\rM^S$. For consumers, the bound is $\varepsilon(|I|)$, which vanishes with market size --- the same as under $\rM^E$. For prosumers, the bound has an additional term $\|w_i\|_q^2/(2\mu)$ that reflects the real-time profit. This term does not vanish with market size, but the designer controls it through the strong convexity parameter $\mu$. A natural question is whether $\mu$ can be chosen large enough to make this term smaller than the prosumer's gain under the original mechanism $\rM^E$. We first prove ~\autoref{thm:main} and then state a Corollary (see ~\autoref{cor:comparison}), which confirms that such a choice is always possible.

\begin{proof}
To simplify the notation, let $P=P(\mathrm{d}_i,\mathrm{d}_{-i})$ and $P'=P(\mathrm{d}^{\prime}_{i},\mathrm{d}_{-i})$ denote the PMFs over real-time bins under truthful reporting and under the under-report, respectively, and let $Q$ denote the true PMF.
We treat consumers and prosumers separately.

\smallskip

\noindent\textbf{Case 1: Consumers ($w_i = 0$).} The consumer's gain under $\rM^S$ equals their gain under $\rM^E$ minus the penalty difference. Then, from Lemma \autoref{lem:savage}, and under  \autoref{ass:certain-alloc}, the gain from under-reporting under mechanism $\rM^{S}$ reads 
\begin{align}\label{eq:consumer-drop}
\mathbb{E}_{ \mathrm{d}_{-i}\sim P_{-i}}\;G_{i}\big(\mathrm{d}_{i},\;\rM^S(\mathrm{d}^{\prime}_i,\mathrm{d}_{-i})\big)&\leq\mathbb{E}_{\mathrm{d}_{-i}\sim P_{-i}}\;G_i\big(\mathrm{d}_i,\;\rM^E(\mathrm{d}^{\prime}_i,\mathrm{d}_{-i})\big) {-}\left|
 D_{\phi}(Q||P')-D_{\phi}(Q||P)
\right|\nonumber\\[5pt]
&\leq \varepsilon(|I|),
\end{align}
where the last inequality follows from~\autoref{thm:eps-BIC} --- the penalty does not worsen the consumer's bound.
\smallskip

\noindent\textbf{Case 2: Prosumers ($w_i \neq 0$).} The prosumer's gain under $\rM^S$ is
\begin{align}\label{eq:prosumer-MS}
\hspace{-8pt}\mathbb{E}_{\mathrm{d}_{-i}\sim P_{-i}}\;G^{\alpha}_i\big(\mathrm{d}_i,\;\rM^S(\mathrm{d}^{\prime}_i,\mathrm{d}_{-i})\big) 
\leq
\varepsilon(|I|) + \langle w_i,\, P' - P \rangle - \left|D_\phi(Q\|P') - D_\phi(Q\|P)\right|,
\end{align}
where the first term is the consumer component bounded by ~\autoref{thm:eps-BIC}, the second is the real-time profit, and the third is the penalty difference.
\smallskip

\noindent\emph{Step 1: Apply ~\autoref{as:bregman}.}
\begin{align}\label{eq:after-bregman}
\mathbb{E}_{\mathrm{d}_{-i}\sim P_{-i}}\;G^{\alpha}_i\big(\mathrm{d}_i,\;\rM^S(\mathrm{d}_i,\mathrm{d}_{-i})\big) 
\;\leq\; 
\varepsilon(|I|) \;+\; \langle w_i,\, P' - P \rangle \;-\; D_\phi(P\|P').
\end{align}
\smallskip
\noindent\emph{Step 2: Apply Lemma~\autoref{lem:strongconvex}.}
\begin{align}\label{eq:after-bregman-ii}
\mathbb{E}_{\mathrm{d}_{-i}\sim P_{-i}}\;G^{\alpha}_i\big(\mathrm{d}_i,\;\rM^S(\mathrm{d}_i,\mathrm{d}_{-i})\big) 
&\;\leq\; 
\varepsilon(|I|) \;+\; \langle w_i,\, P' - P \rangle \;-\; \frac{\mu}{2}\|P - P'\|^2
\end{align}
\smallskip
\noindent\emph{Step 3: Conjugate duality.}
\begin{align}\label{eq:after-bregman-iii}
\mathbb{E}_{\mathrm{d}_{-i}\sim P_{-i}}\;G^{\alpha}_i\big(\mathrm{d}_i,\;\rM^S(\mathrm{d}_i,\mathrm{d}_{-i})\big) 
&\;\leq\; 
\varepsilon(|I|) \;+\; \langle w_i,\, P' - P \rangle \;-\; \frac{\mu}{2}\|P - P'\|^2_p\nonumber\\[5pt]
&\;\leq\; \varepsilon(|I|) \;+\; \sup _{\Delta \in \mathbb{R}^{|\overline J|}}\left[\left\langle w_i, \Delta\right\rangle-\frac{\mu}{2}\|\Delta\|_p^2\right]\nonumber\\[5pt]
&\;\leq\;\varepsilon(|I|) \;+\; \frac{\left\|w_i\right\|_q^2}{2 \mu} .
\end{align}
Note that for $w_i=0$ the bound reduces to that of  \autoref{thm:eps-BIC}.  
\end{proof}
\autoref{thm:main} establishes that the prosumer's gain under $\rM^S$ is at most $\varepsilon(|I|) + \|w_i\|_q^2/(2\mu)$. ~\autoref{thm:not-eps-BIC} established that prosumer's gain under $\rM^E$ can be at least $\varepsilon^{\alpha}_i > 0$. The penalty is effective if the upper bound under the new mechanism is strictly below the lower bound under the old mechanism. We now state that the designer can always achieve this.
\begin{corollary}[The penalty eliminates the prosumer's advantage]\label{cor:comparison}
Consider the prosumer from ~\autoref{thm:not-eps-BIC} with lower bound $\varepsilon^{\alpha}_i  > 0$. For all $|I|$ large enough that $\varepsilon(|I|) < \varepsilon^{\alpha}_i$, any $\mu$ satisfying
\begin{align}
\mu \;>\; \frac{\|w_i\|_q^2}{2\big(\varepsilon^{\alpha}_i - \varepsilon(|I|)\big)}
\end{align}
ensures that the prosumer's gain under $\rM^S$ is strictly less than $\varepsilon^{\alpha}_i$:
\begin{align}
\varepsilon(|I|) \;+\; \frac{\|w_i\|_q^2}{2\mu} \;<\; \varepsilon^{\alpha}_i.
\end{align}
Such a $\mu$ always exists, since the designer controls $\mu$ through $\boldsymbol{\gamma}_i$ (see also Table~\autoref{tab:scoring_rules})
\end{corollary}
\noindent Of course, $\mu$ can be increased further to infinity but as we are going to see in the next section, $\mu$ poses a trade-off between eliminating prosumer's incentives to under-reporting and a large penalty when the same prosumer happens to report truthfully.

\subsubsection{Approximate Penalty Guarantee for Honest Participants}

Ideally, we would like the mechanism to impose zero expected penalty on any player who reports truthfully:
\begin{align}\label{eq:tax-free-objective}
    \mathbb{E}_{\mathrm{d}_{-i}\sim P_{-i}}\;\Pi_{\gamma_i}\big(\mathrm{d}_{i},\;\mathrm{d}_{-i}\big) \;=\; 0,
\end{align}
for all true demands $\mathrm{d}_i$ and all $i\in I$. That is, an honest participant should expect zero penalty in expectation over the others bid. To achieve (\autoref{eq:tax-free-objective}) exactly, the designer would set  $\zeta_i = -\mathbb{E}_{\mathrm{d}_{-i}}\big[D_{\gamma_i}(\mathrm{d}_i, \mathrm{d}_{-i})\big]_+$, which offsets the expected penalty at the truthful report. However, this requires knowledge of the true demand $\mathrm{d}_i$, which is private information. If the designer knew $\mathrm{d}_i$, the contrastive penalty structure would be unnecessary --- one could simply score each player's report against a belief constructed from the true demands. The entire motivation for the contrastive nature of the penalty mechanism is precisely that $\mathrm{d}_i$ is not observed (see \autoref{subsec:observability}).
\smallskip

We resolve this issue with a simple approximation. The designer does not observe individual true demands, but does observe the total realized demand $D = \sum_{i \in I} \mathrm{d}_i$ ex post. In a large market with comparable consumers, the average demand $D/|I|$ is a reasonable proxy for each player's true demand. Then, the designer sets the offset as
\begin{align}\label{eq:zeta-approx}
\zeta_i \;=\; -\mathbb{E}_{\mathrm{d}_{-i}\sim P_{-i}}\big[D_{\gamma_i}(D/|I|,\; \mathrm{d}_{-i})\big]_+,
\end{align}
which is computable only from observable quantities. Under this choice, the expected penalty at the truthful report $\mathrm{d}_i$ is no longer exactly zero, but approximately zero:
\begin{align}\label{eq:approx-tax-free}
\mathbb{E}_{\mathrm{d}_{-i}\sim P_{-i}}\;\Pi_{\gamma_i}\big(\mathrm{d}_{i},\;\mathrm{d}_{-i}\big)
\;=\;
\mathbb{E}_{\mathrm{d}_{-i}}\big[D_{\gamma_i}(\mathrm{d}_i, \mathrm{d}_{-i})\big]_+
\;-\;
\mathbb{E}_{\mathrm{d}_{-i}}\big[D_{\gamma_i}(D/|I|, \mathrm{d}_{-i})\big]_+.
\end{align}

\begin{proposition}[Approximate tax-free guarantee]\label{prop:tax-free}
Set the offset as in (\autoref{eq:zeta-approx})
where $D = \sum_{i \in I}\mathrm{d}_i$ is the total realized demand. Then, for any consumer $i$ reporting truthfully,
\begin{align}
\mathbb{E}_{\mathrm{d}_{-i}\sim P_{-i}}\;\Pi_{\gamma_i}\big(\mathrm{d}_i,\;\mathrm{d}_{-i}\big)
\;\leq\;
\kappa_{\phi}\cdot\mu \cdot \mathbb{E}_{\mathrm{d}_{-i}\sim P_{-i}}\big\|P(\mathrm{d}_i,\mathrm{d}_{-i}) - P(D/|I|,\mathrm{d}_{-i})\big\|_1,
\end{align}
where $\kappa_{\phi} \;\equiv\;  \frac{\sup_{t\in[0,1]}|\phi_0''(t)|}{\inf_{t\in[0,1]}|\phi_0''(t)|}.$
depends only on the designer's choice of base potential.
\end{proposition}

\begin{proof}
Under the offset $\zeta_i = -\mathbb{E}_{\mathrm{d}_{-i}}[D_{\gamma_i}(D/|I|, \mathrm{d}_{-i})]_+$, the expected penalty at the truthful report $\mathrm{d}_i$ is
\begin{align}
\mathbb{E}_{\mathrm{d}_{-i}}\;\Pi_{\gamma_i}(\mathrm{d}_i, \mathrm{d}_{-i})
&\;=\;
\mathbb{E}_{\mathrm{d}_{-i}}\big[D_{\gamma_i}(\mathrm{d}_i, \mathrm{d}_{-i})\big]_+
\;-\;
\mathbb{E}_{\mathrm{d}_{-i}}\big[D_{\gamma_i}(D/|I|, \mathrm{d}_{-i})\big]_+\nonumber\\[5pt]
&
\;\leq\;
\mathbb{E}_{\mathrm{d}_{-i}}\big|D_{\gamma_i}(\mathrm{d}_i, \mathrm{d}_{-i}) - D_{\gamma_i}(D/|I|, \mathrm{d}_{-i})\big|.\label{eq:lip-pos}
\end{align}

\noindent Recall that the first term of $D_{\gamma_i}(\mathrm{d}'_i, \mathrm{d}_{-i})$ does not depend on the first argument. Therefore,
\begin{align}\label{eq:ref-cancel}
\big|D_{\gamma_i}(\mathrm{d}_i, \mathrm{d}_{-i}) - D_{\gamma_i}(D/|I|, \mathrm{d}_{-i})\big|
{=}
\big|S_{\phi(\boldsymbol{\gamma}_i)}(P(\mathrm{d}_i, \mathrm{d}_{-i}), Q) - S_{\phi(\boldsymbol{\gamma}_i)}(P(D/|I|, \mathrm{d}_{-i}), Q)\big|.
\end{align}

\noindent
For the separable potential $\phi_{\boldsymbol{\gamma}_i}(P)$, the Savage representation gives
\begin{align}
S_{\phi(\boldsymbol{\gamma}_i)}(P, Q) \;=\; \sum_l \gamma_{i,l}\big[\phi_0(P_l) + \phi_0'(P_l)(Q_l - P_l)\big].
\end{align}
Fix $\mathrm{d}_{-i}$ and declare $P_l = P_l(\mathrm{d}_i, \mathrm{d}_{-i})$ and $\tilde{P}_l = P_l(D/|I|, \mathrm{d}_{-i})$. Then, for each component $l$, the function $p \mapsto \phi_0(p) + \phi_0'(p)(Q_l - p)$ has derivative $\phi_0''(p)(Q_l - p)$, which is bounded in absolute value by $\sup_{t \in [0,1]}|\phi_0''(t)| \cdot |Q_l - p| \leq M_{\phi_0}$ since $Q_l, p \in [0,1]$. By the mean value theorem applied to each component,
\begin{align}\label{eq:score-lip}
\big|S_{\phi(\boldsymbol{\gamma}_i)}(P, Q) - S_{\phi(\boldsymbol{\gamma}_i)}(\tilde{P}, Q)\big|
\;\leq\;
\sum_l \gamma_{i,l}\, M_{\phi_0}\, |P_l - \tilde{P}_l|
\;\leq\;
\left\|\boldsymbol{\gamma}_i\right\|_{\infty} M_{\phi_0}\, \|P - \tilde{P}\|_1,
\end{align}

\noindent
\noindent
Substituting (\autoref{eq:ref-cancel}) and (\autoref{eq:score-lip}) into (\autoref{eq:lip-pos}),
\begin{align}\label{eq:taxfree-gamma}
\mathbb{E}_{\mathrm{d}_{-i}}\;\Pi_{\gamma_i}(\mathrm{d}_i, \mathrm{d}_{-i})
\;\leq\;
\left\|\boldsymbol{\gamma}_i\right\|_{\infty}\,M_{\phi_0}\cdot\mathbb{E}_{\mathrm{d}_{-i}}\big\|P(\mathrm{d}_i, \mathrm{d}_{-i}) - P(D/|I|, \mathrm{d}_{-i})\big\|_1.
\end{align}
For uniform weights $\gamma_{i,l} = \gamma_i$ for all $l$, we have $\|\boldsymbol{\gamma}_i\|_{\infty} = \gamma_i$ and the strong convexity parameter reads $\mu = \gamma_i\cdot\inf_{t\in[0,1]}|\phi_0''(t)|$. Substituting $\gamma_i = \mu\,/\,\inf_{t}|\phi_0''(t)|$ into (\autoref{eq:taxfree-gamma}),
\begin{align}\label{eq:taxfree-bound-mu}
\mathbb{E}_{\mathrm{d}_{-i}\sim P_{-i}}\;\Pi_{\gamma_i}\big(\mathrm{d}_i,\;\mathrm{d}_{-i}\big)
&\;\leq\;
\frac{\sup_{t\in[0,1]}|\phi_0''(t)|}{\inf_{t\in[0,1]}|\phi_0''(t)|}\cdot\mu \cdot \mathbb{E}\big\|P(\mathrm{d}_i,\mathrm{d}_{-i}) - P(D/|I|,\mathrm{d}_{-i})\big\|_1\nonumber\\[5pt]
&\;=\;
\kappa_{\phi}\cdot\mu \cdot \mathbb{E}\big\|P(\mathrm{d}_i,\mathrm{d}_{-i}) - P(D/|I|,\mathrm{d}_{-i})\big\|_1.
\end{align}
\end{proof}

\paragraph*{The penalty trade-off for honest participants.}\label{trade-off}

From~\autoref{thm:main}, the prosumer's gain from misreporting is at most $\varepsilon(|I|) + \|w_i\|_q^2/(2\mu)$, which decreases as $\mu$ grows. From~\autoref{prop:tax-free}, the penalty at truthful reporting is at most $\kappa_{\phi}\cdot\mu\cdot\mathbb{E}\|P - \tilde{P}\|_1$, which increases as $\mu$ grows. The two bounds move in opposite directions: a large $\mu$ eliminates the prosumer's incentives, but when the same prosumer happens to report truthfully, the large $\mu$ amplifies the penalty they pay for honest participation. We now study this tradeoff and characterize the value of $\mu$ that best balances the two.
\smallskip

\noindent Consider a fixed prosumer $i$ with profit vector $w_i$ and a fixed penalty parameter $\mu$. The same $\mu$ governs the mechanism regardless of whether the prosumer underreports or reports truthfully. 
\smallskip

\noindent\textbf{The tradeoff.} Define the following shorthands:
\begin{align}
R(\mu) \;\equiv\; \frac{\|w_i\|_q^2}{2\mu}, \qquad T(\mu) \;\equiv\; \kappa_{\phi}\cdot\mu\cdot\mathbb{E}\big\|P - \tilde{P}\big\|_1,
\end{align}
where $\tilde{P} = P(D/|I|, \mathrm{d}_{-i})$. Here $R(\mu)$ is the residual gain from under-reporting (excluding $\varepsilon(|I|)$) and $T(\mu)$ is the penalty at truthful reporting. The function $R$ is decreasing in $\mu$ and $T$ is increasing in $\mu$. The designer faces a tradeoff: increasing $\mu$ tightens the bound on under-reporting but loosens the bound on truthful reporting. The value of $\mu$ that balances the two objectives is obtained by setting $R(\mu) = T(\mu)$. Solving for $\mu^{\star}$ yields
\begin{align}\label{eq:optimal-mu}
\mu^{\star} \;=\; \frac{\|w_i\|_q}{\sqrt{2\,\kappa_{\phi}\cdot\mathbb{E}\big\|P - \tilde{P}\big\|_1}}.
\end{align}
At this value, the residual gain from misreporting and the penalty at truthful reporting are both equal to
\begin{align}\label{eq:common-value}
R(\mu^{\star}) \;=\; T(\mu^{\star}) \;=\; \|w_i\|_q\;\sqrt{\frac{\kappa_{\phi}\cdot\mathbb{E}\big\|P - \tilde{P}\big\|_1}{2}}.
\end{align}

\noindent In summary, no fixed $\mu$ can achieve both objectives perfectly: deterring under-reporting requires large $\mu$, protecting truthful reporting requires small $\mu$. The optimal $\mu^{\star}$ balances the two, and the resulting common bound vanishes in large markets. 

\begin{table}[t]
\centering
\resizebox{1\textwidth}{!}{%
\setlength{\tabcolsep}{10pt}
\setlength{\extrarowheight}{6pt}
\renewcommand{\arraystretch}{1.5}
\begin{tabular}{|l|l|l|l|l|l|}
\hline
\textbf{Potential} $\phi(P;\gamma)$ &
\textbf{Score name} &
\textbf{Score} $S(P,i)$ &
$\mathbb{E}_{j\sim Q}[S(P,j)]$
&
\textbf{Norm} $p$
&
\textbf{Strong convexity} $\mu(\boldsymbol{\gamma})$ \\
\hline
$\sum_{i}\gamma_i P_i^2$ &
\shortstack{Quadratic\\(Brier)} &
$2\gamma_i P_i-\sum_k\gamma_k P_k^2$ &
$2\sum_{i}\gamma_iQ_iP_i-\sum_i\gamma_iP^2_i$
&
$2$
&
$2\gamma_{\min}$  \\
\hline
$\sum_{i}\gamma_i P_i\ln P_i$ &
\shortstack{Logarithmic\\(Log score)} &
$\gamma_i(\ln P_i+1)-\sum_k\gamma_k P_k$ &
$\sum_i \gamma_i Q_i\left(\ln P_i+1\right)-\sum_i \gamma_i P_i$&
$1$
&
$\gamma_{\min}$ \\
\hline
$\sum_{i}\gamma_i \dfrac{P_i^\alpha}{\alpha}$,\; $\alpha>1$ &
\shortstack{Power / Tsallis\\($\alpha$-score)} &
$\gamma_i P_i^{\alpha-1}-\dfrac{\alpha-1}{\alpha}\sum_k\gamma_k P_k^\alpha$ &
$\sum_i \gamma_i Q_i {P}_i^{\alpha-1}-\frac{\alpha-1}{\alpha} \sum_i\gamma_i {P}_i^\alpha$
&
$\alpha$
&
$\mu_{\alpha}=\begin{cases}(\alpha-1)\gamma_{\min} & \alpha \leq 2 \\ (\alpha-1)\gamma_{\min}\,\delta^{\alpha-2} & \alpha > 2\end{cases}$ \\
\hline
$\sum_{i}\gamma_i e^{P_i}$ &
\shortstack{Exponential\\score} &
$\gamma_i e^{P_i}+\sum_k\gamma_k e^{P_k}(1-P_k)$ &
$\sum_i \gamma_i Q_i e^{{P}_i}+\sum_i \gamma_i e^{{P}_i}\left(1-{P}_i\right)$
&
$1$
&
$\gamma_{\min}$ \\
\hline
\end{tabular}%
}
\caption{Strictly proper scoring rules with generating potentials and strong
convexity parameters $\mu(\boldsymbol{\gamma})$ with respect to $\Vert\cdot\Vert_p$.
For the power potential with $\alpha > 2$, the simplex is truncated to
$P_i \geq \delta > 0$ to ensure strong convexity, since $P_i^{\alpha-2} \to 0$
as $P_i \to 0$.
The logarithmic and exponential potentials achieve $\mu_1 = \gamma_{\min}$.
For the log score this is because $\gamma_i/P_i \geq \gamma_i$ with the infimum
achieved at $P_i = 1$; for the exponential score because $e^{P_i} \geq 1$ with
the infimum achieved at $P_i = 0$.}
\label{tab:scoring_rules}
\end{table}
\smallskip

\noindent\textbf{Large market regime.} The common value (\autoref{eq:common-value}) is the product of two factors: the prosumer's profit norm $\|w_i\|_q$ and $\sqrt{\kappa_{\phi}\cdot\mathbb{E}\|P - \tilde{P}\|_1/2}$. The first is a fixed property of the prosumer. The second vanishes in the large market regime where players are comparable in size, because $D/|I| \to \mathrm{d}_i$ implies $\mathbb{E}\|P - \tilde{P}\|_1 \to 0$. Therefore, at the optimal $\mu^{\star}$, both the gain from under-reporting and the penalty at truthful reporting vanish as the day-ahead market grows.

\subsubsection{Budget Balance.}
Baseline day-ahead and real-time energy settlement are jointly budget balanced by construction under uniform
pricing and physical balance constraints (see (\autoref{energy-budget-balanced})). The introduction of the settlement adjustment
$\Pi_i$ generally generates additional revenue that needs to be balanced. 
\smallskip

The TSOs are allowed to develop an additional settlement mechanism separate from imbalance settlement to recover procurement costs of balancing capacity, administraitive costs, and other balancing-related costs (subject to RNA approval) \citep{EU2017R2195_EEA}. On top of that, any net financial outcome must be passed on to network users (so extra revenue must ultimately be returned / offset somewhere) \citep{EU2017R2195_EEA}.
\smallskip

Thus, to preserve feasibility of the overall institution, we impose budget balance through a
neutrality account. Specifically, we introduce an account $0$ with transfer
\[
t_0\equiv \sum_{i\in I}\Pi_i
\]
and require the accounting identity
\begin{align}\label{o3}
    \sum_{i\in I}t_i^{DA,\text{new}}(\mathrm{d})
+\sum_{j\in\overline J}t_k^{RT}(D,\mathrm{d})
+t_0(\mathrm{d})=0
\end{align}
to hold for every realization-type profile $(D,\mathrm{d})$.

\section{Numerical Results}\label{numerical-simulations}

\subsection{Simulation Setup and Synthetic Results}\label{synthetic results}
We isolate the three qualitative phenomena presented in Sections \autoref{section-two}-\autoref{sec:mechanism-design} in a constructed synthetic market that focuses only on the incentive question.
\subsubsection{Simulation Setup}
 The day-ahead and real-time merit-order supply curves are fixed. Every player holds a uniform belief over private demand of common width $l$. We focus on player $i\in I$ with true demand $\mathrm{d}_i$, fixed across market sizes. This allows the comparisons to be made at a common baseline.
\smallskip

In every experiment we evaluate, on a grid of candidate reports $\mathrm{d}^{\prime}_i\in [0,\mathrm{d}_i]$, the expected gain from under-reporting, $\mathbb{E}_{\mathrm{d}_{-i}\sim P_{-i}} \left[
G_i(\mathrm{d}_i,\;\rM(\mathrm{d}^{\prime}_i,\mathrm{d}_{-i}))
\right],$ with the expectation taken by Monte Carlo. The four experiments differ only in which mechanism $\rM$ is in force and which structural parameters are varied. Each subsequent experiment adds one ingredient on top of the previous one, so that the figures read progressively.
\smallskip

\noindent
\textbf{Units.} Throughout, energy quantities ($\mathrm{d}_i$, $\mathrm{d}'_i$, $D$, $\overline g_j$) are in $\mathrm{MWh}$ delivered over the one-hour delivery period, and prices ($p^{DA}$, $p^{RT}$, $v_i$, $c_j$, $\overline c_j$) are in \$/$\mathrm{MWh}$. Each term of the per-period utility is therefore a product $\$/\mathrm{MWh} \times \mathrm{MWh} = \$$, which we read as a per-hour dollar amount, i.e.\ $\$/h$. Thus the gain and expected gain inherits the same units. That said, all payoff and gain axes in the figures below are reported in $\$/h$.
\subsubsection{Synthetic Results}

\noindent \textbf{\autoref{thm:eps-BIC} --- vanishing consumer gain.}
We consider a family of energy markets indexed by the number of consumers $|I|$. All players are consumers with $\alpha_j=0$ for all $j\in \overline{J}$. The demand of player $i\in I$, $\mathrm{d}_i$ is held fixed and the day-ahead nominal generation capacities are rescaled across market sizes so that the truthful clearing always lands on a common merit-order supply point. Recall (\autoref{eq:Delta-consumer-33}) under  \autoref{ass:certain-alloc}. The intuition is as follows: For an under-report to be profitable, it must shift the day-ahead clearing price which requires $\sum_{j\neq i}\mathrm{d}_j$ to fall within a strip of width $\mathrm{d}_i-\mathrm{d}^{\prime}_i$ around a supply-curve step.  \autoref{thm:eps-BIC} bounds the probability of that event by $O(1/\sqrt{|I|})$, so the expected gain inherits the same rate.

\begin{figure}[htbp]
\centering
\includegraphics[width=\linewidth]{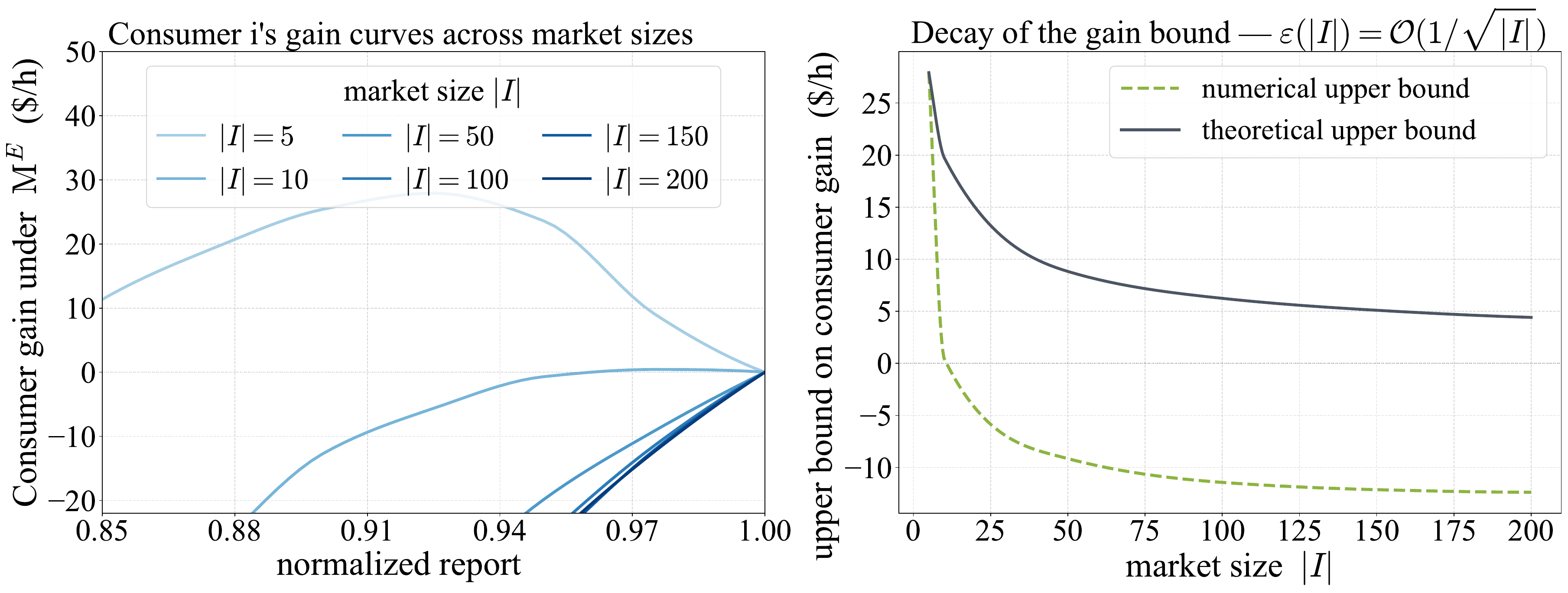}

\caption{\textbf{Vanishing consumer gain.} Left: expected gain of the consumer $i$ under mechanism $\rM^{E}$ as a function of the normalised report $\mathrm{d}^{\prime}_i/\mathrm{d}_i$, one curve per market size $|I|$. Each curve passes through zero at truthful reporting and decreaases toward negative values as the under-report grows. Right: maximum of the gain over $\mathrm{d}^{\prime}_i$ as a function of $|I|$, together with a reference of the theoretical bound. The numerical envelope starts at a small positive level, decays at the predicted rate, and crosses below zero once $|I|$ is large enough that the day-ahead-real-time residual strictly dominates the day-ahead price gap term.}
\label{fig:synth_thm2}
\end{figure}
\autoref{fig:synth_thm2} confirms the predicted decay. The left
panel shows that the sign of the gain is driven by the under-report cost: every curve rises toward zero from below as $\mathrm{d}^{\prime}_i$ tends to $\mathrm{d}_i$ ($\mathrm{d}^{\prime}_i/\mathrm{d}_i\to 1$), and the family shifts upward as $|I|$ grows. The right panel removes the report dimension by taking for each $|I|$, the maximum of the gain over non-truthful reports and comparing it against the theoretical bound reference. For small markets the maximum is positive but small: a narrow interior band of under-reports still yields profit because the probability of shifting the clearing price is high enough to beat the residual. As $|I|$ grows, that band shrinks at rate $1/\sqrt{|I|}$ and eventually disappears for $|I|=50$; the curve crosses below the reference and into the negative region. Intuitively, a single consumer's ability to move the market price decreases with scale, and once it has decreased enough, the residual forces every misreport into strict loss.
\smallskip

\noindent \textbf{\autoref{thm:not-eps-BIC} --- persistence of the prosumer gain.} We repeat the previous experiment, but now player $i$ is endowed with an ownership vector $\alpha_i$ on the simplex of real-time producers and with nominal capacities $\overline{g}$. The remaining players stay consumers with $\alpha_j=0$. The gain under $\rM^{E}$ conists of the previous consumer term and an real-time-share term (see (\autoref{shares})). The explanation is that, by under-reporting, the prosumer pushes a positive imbalance into real time. Once that imbalance crosses the next capacity threshold of the real-time merit order, her own asset is activated and collects the cost gap between it and the preceding unit. This profit depends only on the real-time supply curve. The convex combination $\alpha_i$, and the nominal generation capacities does not vanish with $|I|$. In every market size, one therefore expects a band of interior under-reports in which the real-time-share channel dominates the under-report cost and yields a strictly positive gain.

\begin{figure}[htbp]
    \centering
\includegraphics[width=\linewidth]{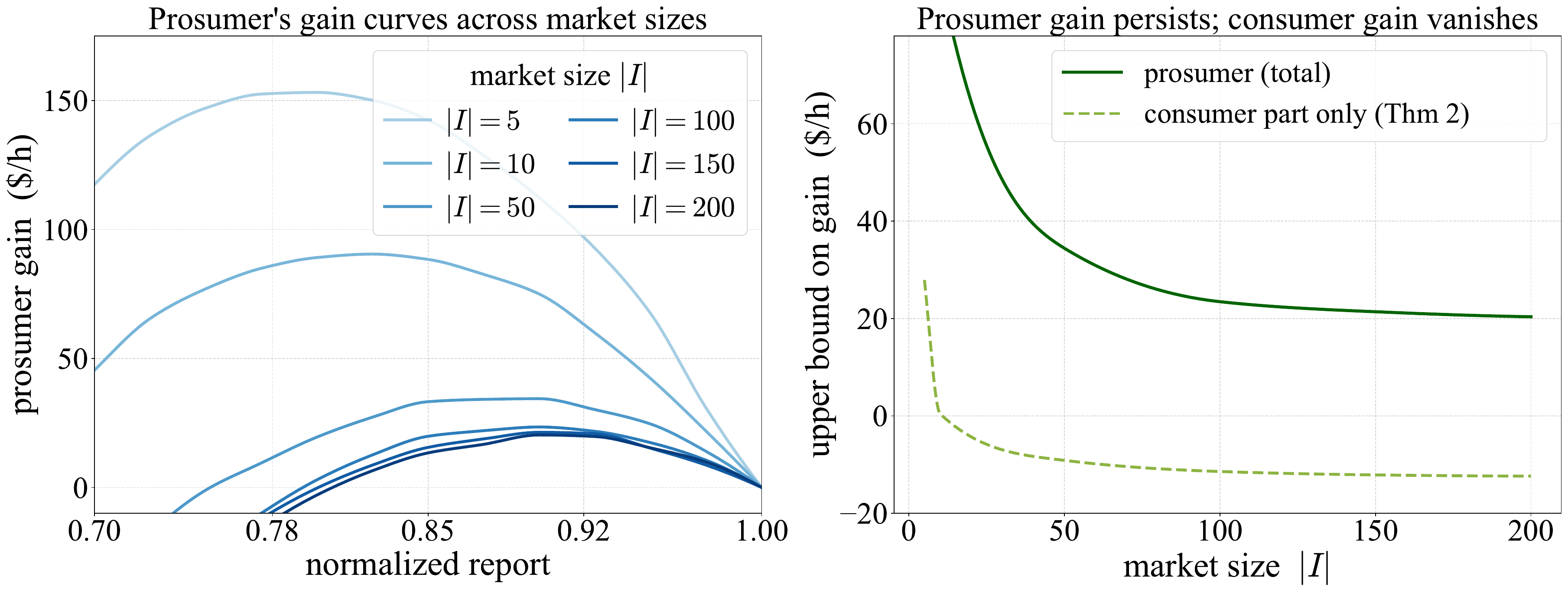}\\
\caption{\textbf{Persistent prosumer gain.} Left: expected gain under $\rM^{E}$ across the same market family of market sizes $|I|=5, 10, 50, 100, 150, 200$ as \autoref{fig:synth_thm2}. At the full display scale the faimily again looks dominated by the under-report cost, so the structural feature of interest, a positive interior humb around the reports that activate the real-time-share, is shown in the zoomed sub-plot. Right: maximum prosumer gain over reports $\mathrm{d}^{\prime}_{i}\in [0,\mathrm{d}_i)$ as a function of $|I|$ mapped onto the consumer bound of \autoref{fig:synth_thm2}. The consumer bound decays to zero and below while the prosumer bound stabilizes at a strictly positive level (around 50 $\$/h$ ;see also the level of the humb in the left panel) determined by the real-time supply structure, the nominal generation capacities and $\alpha_i$.}
\label{fig:synth_thm3}
\end{figure}
\autoref{fig:synth_thm3}  reproduces  \autoref{thm:not-eps-BIC} visually. The right panel contrasts the two maxima: the consumer curve drops as before, while the prosumer curve stabilizes at a strictly positive level. The left panel explains why: even for large $|I|$ there is a band of under-reports, located close to truthful (in a neighborhood of $\mathrm{d}^{\prime}_i/\mathrm{d}_i=0.89$), on which the prosumer's real-time-share revenue increases the under-report cost and produces a strictly positive gain. The height of that hump depends on the real-time supply curve, generation capacities and $\alpha_i$, and hence, it does not shrink with the market size. No single tolerance $\varepsilon$ can serve consumers and prosumers simultaneouly once $|I|$ is large enough, and $\rM^{E}$ therefore fails to be $\varepsilon-$BIC
 with vanishing $\varepsilon$ over the whole player set.
\bigskip

\noindent \textbf{\autoref{thm:main} --- scoring rule restoration.} We reuse the market of the previous experiments and augment the day-ahead settlement with the contrastive scoring-rule payment $\Pi_{\gamma_i}$ of \autoref{sec:mechanism-design}. The payment is parameterized by the strong convexity parameter $\mu$ of the underlying potential function. We sweep $\mu$ over a set of values while holding all other primitives constant. The prosumer's gain under the extended mechanism $\rM^S$ equals her gain under $\rM^E$ minus the scoring-rule payment (\autoref{incentive-structure-scoring-rules}). 
\smallskip

\autoref{fig:thm4} shows how the penalty closes the prosumer gap. The interior hump of \autoref{fig:synth_thm3} which was responsible for the persistent prosumer advantage under $\rM^E$, is pushed into the negative region once $\mu$ is large enough. The right panel gives the quantitative story: the family of green curves traces the  \autoref{thm:main} bound $\varepsilon(|I|)+\Vert w_i\Vert^{2}_2/(2\mu)$. The first term decays with $|I|$ at rate $1/\sqrt{|I|}$ established for consumer, and the second term is a free parameter that the designer can shrink by increasing $\mu$. For large $\mu$, the prosumer bound coincides with the consumer bound, restoring a uniform $\varepsilon-$BIC guarantee over the full player set.

\begin{figure}[h]
    \centering
\includegraphics[width=\linewidth]{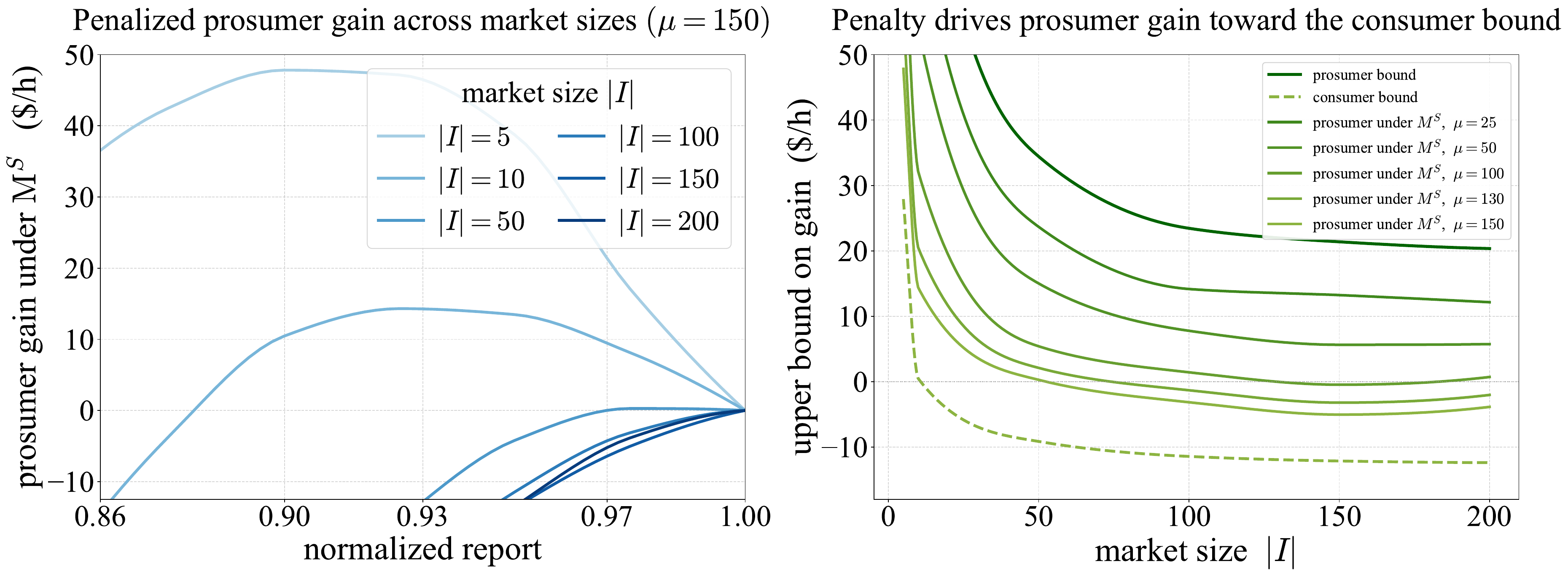}
    \caption{\textbf{Scoring-rule mechanism drives the prosumer gain toward the consumer bound} 
    Left: prosumer gain under $\rM^{S}$ across market sizes $|I|=5, 10, 50, 100, 150, 200$ at a single representative penalty strength $\mu=150$. Every curve still touches zero at truth but exhibits a sharp downward step at the under-report where the scoring rule fires: after that threshold the penalty adds a strictly negative contribution that suppress the real-time-share humb of \autoref{fig:synth_thm3}. Right: maximum gain under mechanism $\rM^{S}$ as a function of $|I|$ for a family of penalty strengths $\mu = 50, 100, 130, 150$. As before, the dark red solid curve corresponds to the prosumer's bound and the light green dashed curve to the consumer curve of~\autoref{fig:synth_thm3}. Between the two sits a family of prosumer bounds under $\rM^S$ at varying strong-convexity parameter $\mu$ (dark green = small $\mu$, light green = large $\mu$).
 As $\mu$ grows the red family descents from the dark-red reference toward the dashed green one, in agreement with the bound $\varepsilon(|I|)+\Vert w_i\Vert^{2}_{2}/(2\mu)$.}
    \label{fig:thm4}
\end{figure}

\noindent 
\textbf{Incentive deterence vs. cost of honest reporting.} 
If under the value of $\mu$ such that prosumers' gain drops to the consumer bound (\autoref{fig:thm4}),
the prosumer plays honestly under mechanism $\rM^S$, what does it cost them and how much they give up when under-reporting?
\smallskip

\autoref{fig:tradeoff-focus} makes this comparison at every penalty strength a designer might consider, on the same synthetic market that we used in \autoref{thm:eps-BIC}-\autoref{thm:main}. The horizontal axis is the penalty strength $\mu$ which varies from $0$ up to $150$ -- that is the range that matters for this market. Let $U^{0}$ be the prosumer's average payoff (measured in $\$/h$) on this market when they report their true demand under mechanism $\rM^{E}$. 
\smallskip

It makes sense to use $U^{0}$ as a reference since it expresses what the prosumer would have earned if the penalty had never been introduced. Hence, any cost or gain the penalty creates is naturally read against it. On top of that, $U^0$ adapts to the size of the market: if we double every price, every demand, and every producer capacity, $U^0$ doubles too, and so do the penalty and the under-reporting gain. The percentage of $U^0$ stays the same, so the numbers we report applies to a market that is either ten times larger or ten times smaller. Lastly, $U^0$ is the quantity the prosumer actually cares about. The penalty strength $\mu$ is also in dollars per hour, but it is a setting the designer picks. Comparing $\mu$ to $U^0$ via the percentage axis is what makes the figure say something concrete for the prosumer.

\begin{figure}[htbp]
\centering
\includegraphics[scale=0.25]{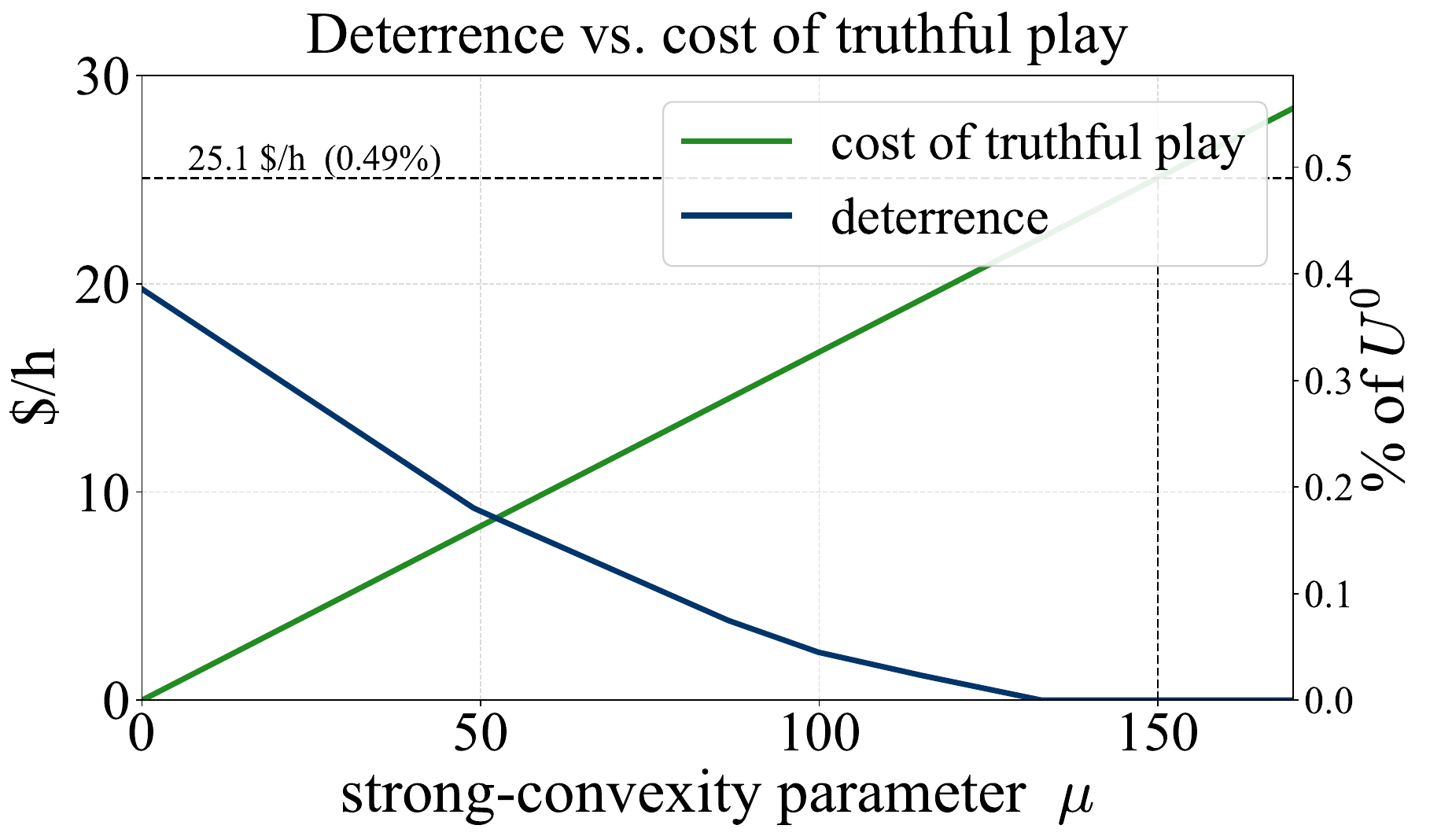}
\caption{Two curves the designer should look at. One expressed as a percentage of the prosumer's baseline hourly earnings $U^0\approx 5127.95 ~\$/h$ and the other absolute dollars for the particular hour slot. \emph{Blue}: the most the prosumer is expected to earn by under-reporting on this market, after the penalty activates; it shrinks to zero as the penalty parameter $\mu$ is turned up. \emph{Green}: the most the mechanism is expected to charge the honest prosumer, in this market; it grows in direct proportion to $\mu$. At $\mu\approx 150$ the two plots summarise the calibration we use for \autoref{thm:eps-BIC}---\autoref{thm:main}:under-reporting gain of $0\%$ and a worst-case honest cost of $0.5\%$ of $U^0$ (respectively $25\$/h$).}
\label{fig:tradeoff-focus}
\end{figure}

\subsection{Empirical Evaluation on Real-Market Data}\label{real market results}
This section demonstrates the performance of our mechanism by applying it to the strategic component that we infer from real market data. In particular, we work with hourly time series of low, medium and high voltage reported loads (interpreted as three players' day-ahead reports) found in \citep{HEnEx2025} together with the corresponding realized demand, extracted from an ENTSO-E-based dataset \citep{entsoe2025greece} and preprocessed to retain only under-estimated hours in which the aggregate report does not exceed actual total consumption. 
\smallskip

On this set we rationalize the day-ahead market reports as subjective equilibria by means of the prosumers valuations, demands, real-time share vectors and real-time nominal capacity bounds that make the observed day-ahead reports approximately optimal with respect to a finite set of breakpoint deviations. 
\smallskip

We formulate this inverse problem as a convex feasibility program, relaxed with slacks and recast it as a regularized convex optimization program (inner level), while the nominal real-time capacities are treated as outer-level decision variables tuned by a derivative-free covariance matrix evolution strategy (CMA-ES) search rather than by a large mixed-integer formulation. 
\smallskip

We detail the key modeling and algorithmic parameters (valuation floor, slack penalties, regularization weights, quantile supports, and solver choices), and finally present numerical results that compare observed and predicted best-response reports, day-ahead/real-time prices, imbalances, and player utilities under the calibrated model, with and without the effect of our mechanism.

\subsubsection{Time Resolution, Units, and Sample Period}\label{subsec:time-resolution}

\textbf{Day-ahead time unit (MTU) in our sample.}
In our empirical sample (prior to October 2025), the coupled EU day-ahead market (SDAC) clears
at an \emph{hourly} MTU, i.e.\ a delivery day is partitioned into 24 hourly MTUs.
SDAC transitioned from hourly to 15-minute MTUs on trading day 30 September 2025 for delivery day
1 October 2025; this occurs after our sample period. \citep{EC2025_DA15min,SDAC2025_15minMTU}
\smallskip

\noindent \textbf{Imbalance settlement time unit (ISP).}
At the EU level, EBGL sets the harmonisation objective for a 15-minute ISP. \citep{EU2017_EBGL}
\smallskip

\noindent \textbf{Model period and units.}
Because our dataset reports realized load only at the hourly level, we conduct the theoretical and
empirical analysis at the \emph{hourly} resolution. We study a single delivery hour and suppress
the hour index whenever no confusion arises.
Let $\tau>0$ denote the duration of one model period in hours; in the baseline specification,
$\tau=1$. All quantity variables are \emph{energy} quantities in MWh delivered over that hour.
For example, $100$ MW sustained over 15 minutes corresponds to $25$ MWh, while $100$ MW sustained
over an hour corresponds to $100$ MWh.

\subsubsection{Data and Pre-Processing}
The dataset combines day-ahead reported quantities at three aggregation levels: low-voltage LV, medium-voltage MV, high-voltage  HV, and realized total load (actual total system demand) $D$. In the implementation, these quantities are fetched from an annual ENTSO-E \citep{entsoe2025greece} "Total Load" time series and a preprocessed day-ahead market dataset containing the voltage levels forecasts \citep{henex}. The data are stored in an array of shape $(\text{days}, 4, 24)$ containing LV, MV, HV, and actual total consumption for each hour of the day, from which we extract a single hour across all days. For that fixed hour, each selected slot $t$ refers to a day of the year. We define three "players" corresponding to the LV, MV, and HV aggregates, and denote the collection $ \mathrm{d}_{i}^t$ as the observed day-ahead report of player $i$ at day $t$, for $i \in {I}$, $|I|=3$. The realized total consumption is denoted $D^t$. We restrict attention to under-estimated days where the reported day-ahead quantities do not exceed the realized average demand. Specifically, we retain only those tuples $\left(\mathrm{d}_1^t, \mathrm{d}_2^t, \mathrm{d}_3^t, D^t\right)$ that satisfy $\sum_{i}\mathrm{d}^{t}_{i}\leq D^{t}$
and discard all remaining days. This ensures that the implied day-ahead schedule is physically feasible without upward curtailment and simplifies the real-time modeling: imbalances are always nonnegative, $\rho^t\geq 0$ and therefore, the real-time market always operates in a "short" regime that activates positive real-time capacities. After this filtering step, we obtain a dataset $\mathrm{dataset}=\left\{\left(\mathrm{d}_1^t, \mathrm{d}_2^t, \mathrm{d}_3^t, D^t\right)\right\}_{t \in {T}}$, where ${T}$ indexes the remaining days for the chosen time slot.

\subsubsection{Inverse Best-Response Model}
Fixing a particular time slot of the day-ahead, we are provided with a sequence of reported profiles and total consumptions $\{(\mathrm{d}^{t}_{1},\ldots,\mathrm{d}^{t}_{|I|},D^{t})\}^{T}_{t=1}$ over different days $t$. Such sequences are publicly available by e.g., the Hellenic Energy Exchange \citep{HEnEx2025} and the ENTSO-E Transparency Platform \citep{entsoe2025greece}. We assume that each player decides upon the same historical horizon $H>0$ and maintains beliefs $P^{t}_{-i}$, and $Q$. For each $t> H$, these histograms are constructed by aggregating previously observed samples over a rolling window of length equal to the horizon $H$, $(\mathrm{d}^{\tau}_{i},\ldots,\mathrm{d}^{\tau}_{|I|},D^{\tau})$, $\tau=t-H,\ldots, t-1$. 
\smallskip

Given observed (from the data) profiles $\{(\mathrm{d}^{t}_{1},\ldots,\mathrm{d}^{t}_{|I|}),D^{t}\}^{T}_{t=H+1}$, we want to determine, for each player and each day, those valuations, demands and incentive strength (convex combinations) that rationalize the players' reports as subjective equilibria of $G^{t}$.
Let $\boldsymbol{v}^{t}=[v^{t}_{1},\ldots v^{t}_{|I|}]$, $\boldsymbol{\mathrm{d}}^{t}=[\mathrm{d}^{t}_{1},\ldots \mathrm{d}^{t}_{|I|}]$, $\boldsymbol{A}^{t}=[\boldsymbol{\alpha}^{t}_{1};\ldots;\boldsymbol{\alpha}^{t}_{|I|}]$, where $\boldsymbol{\alpha}^{t}_{i}=[\alpha^{t}_{i,1},\ldots,\alpha^{t}_{i,|\overline{J}|}]$ is the row vector of shares for player $i\in I$. Lastly, $\overline{\boldsymbol{g}}^{t}=[g^{t}_{1},\ldots,g^{t}_{|\overline{J}|}]$ is the vector of nominal capacities of the real time market. The problem we are interested in solving is: For each $t=H+1,\ldots, T$
\begin{equation}
    \begin{aligned}
& \textbf { Find } \qquad\qquad\left(\boldsymbol{\mathrm{d}}^t, \boldsymbol{v}^t, \boldsymbol{A}^t, \boldsymbol{\overline{g}}^{t}\right) \\
& \textbf { subject to :} \qquad \mathrm{d}_{i}^t \in \arg \max _{x \in[0, \mathrm{d}_i^t]} U_i^t\left(x, ; v_i^t, \boldsymbol{\alpha}_i^t,\overline{\boldsymbol{g}}^{t}\right), \\
& \qquad\qquad\qquad\quad \sum_{i} \mathrm{d}^{t}_{i} = D^t, 
\\
& \qquad\qquad\qquad\quad \sum_{j\in\overline{J}} \alpha^{t}_{i,j} = 1,\\
& \qquad\qquad\qquad\quad v^{t}_{i}\geq 0,
\\
& \qquad\qquad\qquad\quad \overline{\boldsymbol{\alpha}}^{t}\geq 0,
\\
& \qquad\qquad\qquad\quad \overline{\boldsymbol{g}}^{t}\geq 0,
\end{aligned}\quad \forall i\in I.\label{ibr}
\end{equation}
In practice of course, the dominance of the observed report is approximated by enforcing it on a finite grid of deviations $A_{i,t}$ around the observation.
\smallskip

The feasibility problem (\autoref{ibr}) is infinite-dimensional: the best-response conditions involve expectations with respect to empirical distributions and an $\arg\max$ over a continuum of actions $x\in[0,\mathrm{d}_i^t]$, and the real-time capacities $\overline{\boldsymbol{g}}^{t}$ themselves are additional unknowns. In order to implement it numerically we proceed in two steps:
\begin{enumerate}
    \item[(i)] for fixed real-time capacities $\overline{\boldsymbol{g}}^{t}$ we build a finite set of breakpoint deviations $A_{i,t}$ and formulate an \emph{inner} convex program that enforces approximate best-response inequalities and infers $(\boldsymbol{\mathrm{d}}^t,\boldsymbol{v}^t,\boldsymbol{A}^t)$;
    \item[(ii)] treating the capacities $\overline{\boldsymbol{g}}^{t}$ as \emph{outer} variables, we tune them so that the model's forward best responses match the observed reports. This outer problem is tackled with a derivative-free CMA-ES search rather than a large mixed-integer program.
\end{enumerate}

\subsubsubsection{Inner Level: Finite Deviations and Linear Inequalities.}
Fix a finite set of calibration days ${T}^{\mathrm{fit}}\subset\{H+1,\dots,T\}$ and assume that the real-time capacities $\overline{\boldsymbol{g}}^{t}$ for $t\in{T}^{\mathrm{fit}}$ are given. For each such day $t$ and each player $i\in I$, we construct a finite deviation set $A_{i,t}\subset [0,D^t]$ which contains the endpoints $0$ and $D^t$, day-ahead and real-time \emph{breakpoints} (values of $x$ where the day-ahead marginal unit or the real-time marginal genset changes, computed from the day-ahead/real-time stacks and the empirical supports of $P^t_{-i}$ and $Q^t$), and the observed report $ \mathrm{d}_{i}^t$ (used to form utility differences but not treated as a deviation).
\smallskip

For each $(i,t)$ and each deviation $x\in A_{i,t}$, we compare the expected utility of the observed action $ \mathrm{d}_{i}^t$ and the deviation $x$, with $U_i^t(\cdot)$ as in~(\autoref{prosumers-utility}). The approximate best-response condition is:
\begin{equation}
    U_i^t( \mathrm{d}_{i}^t; v_i^t,\boldsymbol{\alpha}_i^t,\overline{\boldsymbol{g}}^{t})
    \geq
    U_i^t(x; v_i^t,\boldsymbol{\alpha}_i^t,\overline{\boldsymbol{g}}^{t})
    - s_{i,t}(x),
    \qquad\forall x\in A_{i,t},
    \label{eq:br-ineq-raw}
\end{equation}
for some slack $s_{i,t}(x)\geq 0$ that allows for small violations and model error. By construction, and for fixed capacities and empirical distributions, $U_i^t(x)$ is affine in the unknowns $(v_i^t,\boldsymbol{\alpha}_i^t)$ 
and the real-time term is a weighted sum $\boldsymbol{\alpha}_i^{t\top}\text{(expected RT reward vector at $x$)}$. Hence there exist coefficients
\[
A_{i,t}(x)\in\mathbb{R},\quad
\boldsymbol{B}_{i,t}(x)\in\mathbb{R}^{|\overline{J}|},\quad
b_{i,t}(x)\in\mathbb{R},
\]
computed from the data and $\overline{\boldsymbol{g}}^{t}$, such that 
\[
U_i^t( \mathrm{d}^{\prime t}_{i}) - U_i^t(x)
= A_{i,t}(x) v_i^t + \boldsymbol{B}_{i,t}(x)^{\!\top}\boldsymbol{\alpha}_i^t - b_{i,t}(x).
\]
Substituting this into (\autoref{eq:br-ineq-raw}) yields the linear constraint
\begin{equation}
A_{i,t}(x) v_i^t + \boldsymbol{B}_{i,t}(x)^{\!\top}\boldsymbol{\alpha}_i^t
\geq
b_{i,t}(x) - s_{i,t}(x),
\qquad \forall x\in A_{i,t},~\forall i\in I,~\forall t\in{T}^{\mathrm{fit}}.
\label{eq:br-ineq-linear}
\end{equation}
We retain the feasibility constraints from (\autoref{ibr}) in a slightly strengthened form:
\begin{align}
&\sum_{i\in I} \mathrm{d}_i^t = D^t,
\qquad 
\mathrm{d}_i^t \geq {\mathrm{d}}_{i}^t,
\qquad 
\mathrm{d}_i^t\geq 0,
\quad\forall\, t\in T^{\mathrm{fit}},\label{eq:cap-constraints}\\
&\sum_{j\in\overline{J}} \alpha_{i,j}^t = 1,
\qquad 
\alpha_{i,j}^t\geq 0,
\quad\forall\, i\in I,~\forall\, t\in T^{\mathrm{fit}},\label{eq:alpha-constraints}\\
&v_i^t \geq v_{\min},
\quad\forall\, i\in I,~\forall\, t\in T^{\mathrm{fit}},\label{eq:vfloor}
\end{align}
where $v_{\min}$ is a valuation floor chosen strictly above the largest marginal cost in the day-ahead stack.

\subsubsubsection{Inner Regularized Convex Program.}
To select a particular solution and keep the inferred parameters economically reasonable, we minimize the total slack together with small quadratic regularization terms. Introducing positive weights $\gamma_{\mathrm{slack}},\gamma_{\alpha},\gamma_{v}>0, \gamma_{\mathrm{d}}>0$, we solve the following \emph{inner} inverse best-response problem:
\begin{equation}
\label{eq:inner-opt}
\begin{aligned}
\min_{\substack{v_i^t,\mathrm{d}_i^t,\boldsymbol{\alpha}_i^t\\ s_{i,t}(x)\geq 0}}
\quad &
\sum_{i\in I}
\biggl[
\gamma_{\mathrm{slack}}\sum_{x\in A_{i,t}} s_{i,t}(x)
+
\gamma_{\alpha} \|\boldsymbol{\alpha}_i^t\|_2^2
+
\gamma_{v}(v_i^t)^2 + \gamma_{\mathrm{d}} (\mathrm{d}_i^t)^2\bigr)
\biggr]
\\[2mm]
\text{s.t.}\quad &
\text{BR inequalities (\autoref{eq:br-ineq-linear}) for all }i,t,x,\\
&
\text{demand constraints (\autoref{eq:cap-constraints}), share constraints (\autoref{eq:alpha-constraints}),}\\
&\text{and valuation floor (\autoref{eq:vfloor}).}
\end{aligned}
\end{equation}
All constraints in (\autoref{eq:inner-opt}) are affine and the objective is a strictly convex quadratic function. Thus, for fixed capacities $\overline{\boldsymbol{g}}^{t}$ and empirical distributions $(P^t_{-i},Q^t)$, problem (\autoref{eq:inner-opt}) is a convex quadratic program which we solve with off-the-shelf solvers (OSQP, with ECOS and SCS as fallbacks), obtaining inferred primitives $(\widehat{v}^{t}_{i},\widehat{\mathrm{d}}^{t}_{i},\widehat{\alpha}_{i}^{t})$ for each $(i,t)\in I\times{T}^{\mathrm{fit}}$.

\paragraph{Outer Level: Calibration of Real-Time Capacities.}
In the inner problem, we treated the real-time capacities $\overline{\boldsymbol{g}}^{t}$ as fixed parameters. The second step is to adjust these capacities so that the model's predicted best-response bids reproduce the observed ones. Given a candidate collection of real-time capacities $\overline{\boldsymbol{g}} = \bigl(\overline{\boldsymbol{g}}^{t}\bigr)_{t\in{T}^{\mathrm{fit}}}$, we keep the inner primitives $(\widehat v_i^t,\widehat{\mathrm{d}}_i^t,\widehat{\boldsymbol{\alpha}}_i^t)$ fixed and simulate the subjective games $G^t$ forward. For each day $t$ and player $i$ we compute a \emph{soft} (quantal) best response $ \mathrm{d}_{i}^{t,\mathrm{BR}}(\overline{\boldsymbol{g}}^t)$ by evaluating $U_i^t(x)$ on a refined breakpoint grid and applying a softmax over utilities:
\[ \mathrm{d}_{i}^{t,\mathrm{BR}}(\overline{\boldsymbol{g}}^t)
=
\sum\nolimits_{x\in A^{\mathrm{sim}}_{i,t}}
\pi_i^t(x;\overline{\boldsymbol{g}}^t) x,
\qquad
\pi_i^t(x;\overline{\boldsymbol{g}}^t)
=
\frac{\exp\bigl(U_i^t(x; v_i^t,{\boldsymbol{\alpha}}_i^t,\overline{\boldsymbol{g}}^t)/\tau\bigr)}
{\sum_{\tilde x\in A^{\mathrm{sim}}_{i,t}}\exp\bigl(U_i^t(\tilde x;\widehat v_i^t,\widehat{\boldsymbol{\alpha}}_i^t,\overline{\boldsymbol{g}}^t)/\tau\bigr)},
\]
where $A^{\mathrm{sim}}_{i,t}$ is a simulation breakpoint set (constructed in the same spirit as $A_{i,t}$), and $\tau>0$ is a temperature parameter. As $\tau\to 0$ the soft best response approaches the hard $\operatorname{argmax}$, whereas larger $\tau$ yields smoother responses and better numerical robustness.

We then measure the fit between the model and the data by the mean squared deviation between observed and predicted bids:
\begin{equation}
\label{eq:outer-loss}
L(\overline{\boldsymbol{g}})
:=
\frac{1}{|I|}
\sum\nolimits_{t\in{T}^{\mathrm{fit}}}
\sum\nolimits_{i\in I}
\bigl( \mathrm{d}_{i}^{t,\mathrm{BR}}(\overline{\boldsymbol{g}}^t)
-
 \mathrm{d}_{i}^t
\bigr)^2.
\end{equation}
The \emph{outer} problem is then
\begin{equation}
\label{eq:outer-opt}
\begin{aligned}
\min_{\overline{\boldsymbol{g}}}\quad & L(\overline{\boldsymbol{g}})\\
\text{s.t.}\quad &
\underline g_{j}^t \le \overline g_j^t \le \overline g_{j,\max}^t,
\qquad \forall j\in\overline{J},~\forall t\in{T}^{\mathrm{fit}},
\end{aligned}
\end{equation}
where $\underline g_j^t>0$ and $\overline g_{j,\max}^t$ are simple box bounds (in the implementation, they are built from a heuristic ``suggested'' real-time capacity based on the empirical distribution of positive imbalances).
\smallskip

The mapping $\overline{\boldsymbol{g}}\mapsto L(\overline{\boldsymbol{g}})$ is evaluated via simulation and is nonconvex and nonsmooth. We therefore solve (\autoref{eq:outer-opt}) with a covariance-matrix adaptation evolution strategy (CMA-ES), acting on the flattened vector of capacities. In the numerical experiments we use a two-stage procedure: a first ``cheap'' stage with a very short belief horizon (used to warm-start the capacities), followed by a second stage with a longer horizon and refined CMA-ES search.

\subsubsection{Results and Discussion}

In this section we answer, slot by slot and player by player, the following question: if we keep the same fitted primitives and switch from the mechanism $\rM^E$ to mechanism $\rM^S$, does the best response move towards truthfulness? If yes, by how much, at what cost, and on which slots? Under $\rM^S$ each player maximizes the same $\rM^E$ utility minus the contrastive penalty $\Pi_{\gamma_i}$ defined in (\autoref{incentive-structure-scoring-rules}).
\smallskip

The first thing to look at is what the players actually report. \autoref{fig:reports-obs-br}  shows that under $\rM^E$ we see the observed report $\mathrm{d}_i^{\,t}$ and the fitted true demand $\widehat{\mathrm{d}}_i^{\,t}$ sitting strictly above it on most slots. That vertical gap is the empirical evidence of strategic under-reporting: the fitted player knows their true demand is $\widehat{\mathrm{d}}_i^{\,t}$, but, given the real-time profit channel, it is cheaper for them to drop the day-ahead report down to $\mathrm{d}_i^{\,t}$. 

\begin{figure}[htbp]
    \centering
    \begin{subfigure}{0.32\textwidth}
        \centering
        \includegraphics[width=\linewidth]{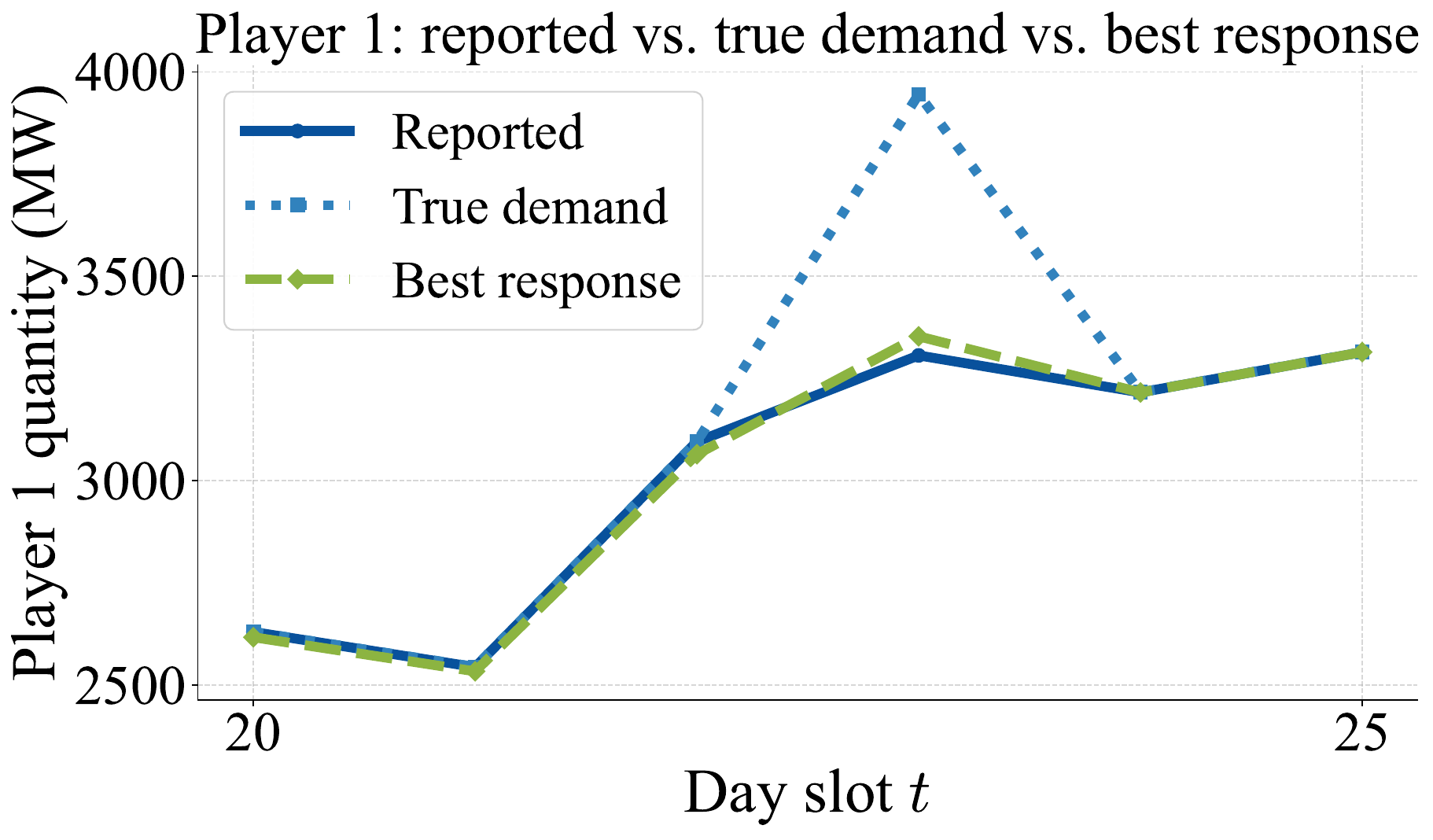}
    \end{subfigure}
    \hfill
    \begin{subfigure}{0.32\textwidth}
        \centering
        \includegraphics[width=\linewidth]{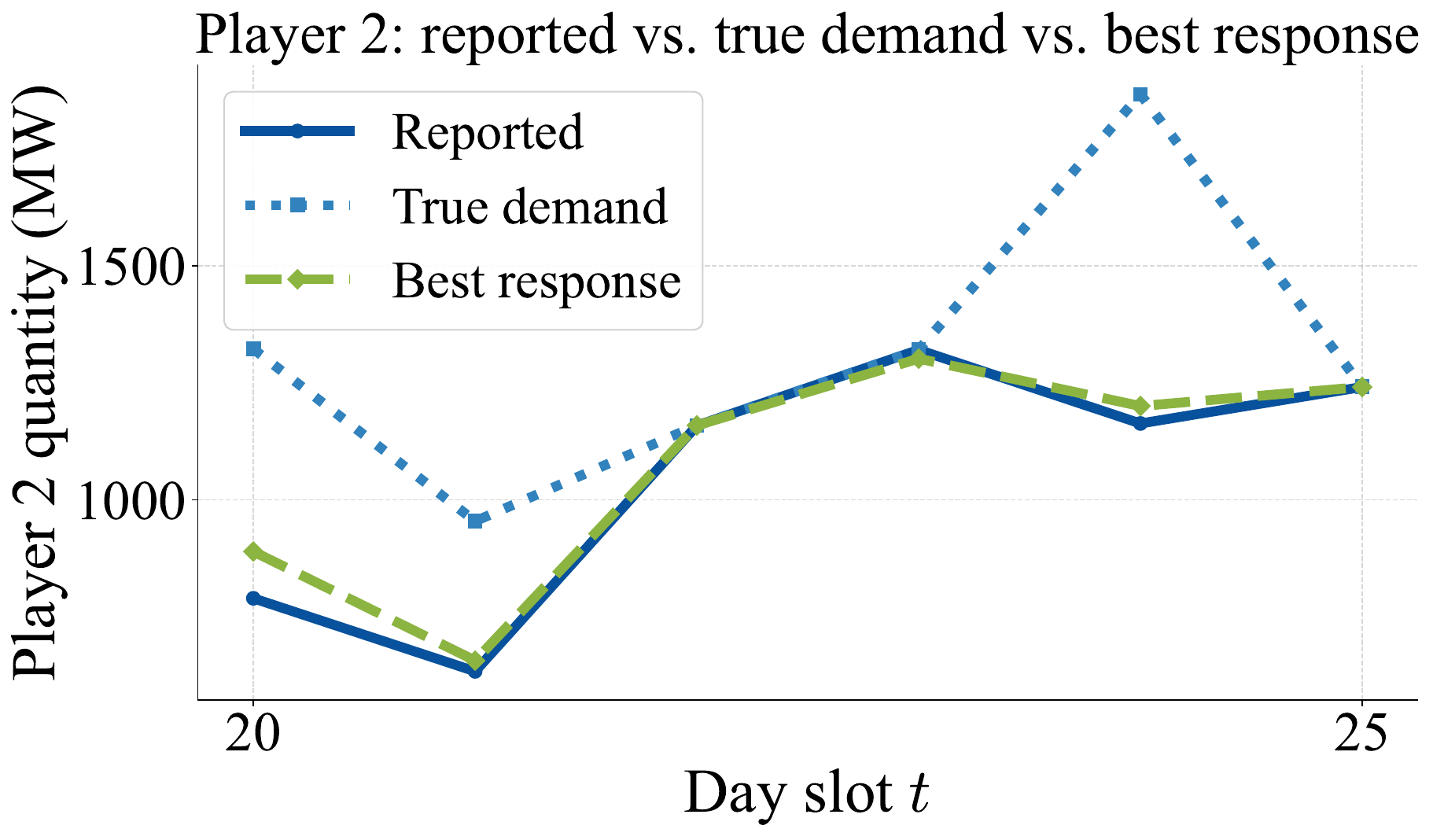}
    \end{subfigure}
    \hfill
    \begin{subfigure}{0.32\textwidth}
        \centering
        \includegraphics[width=\linewidth]{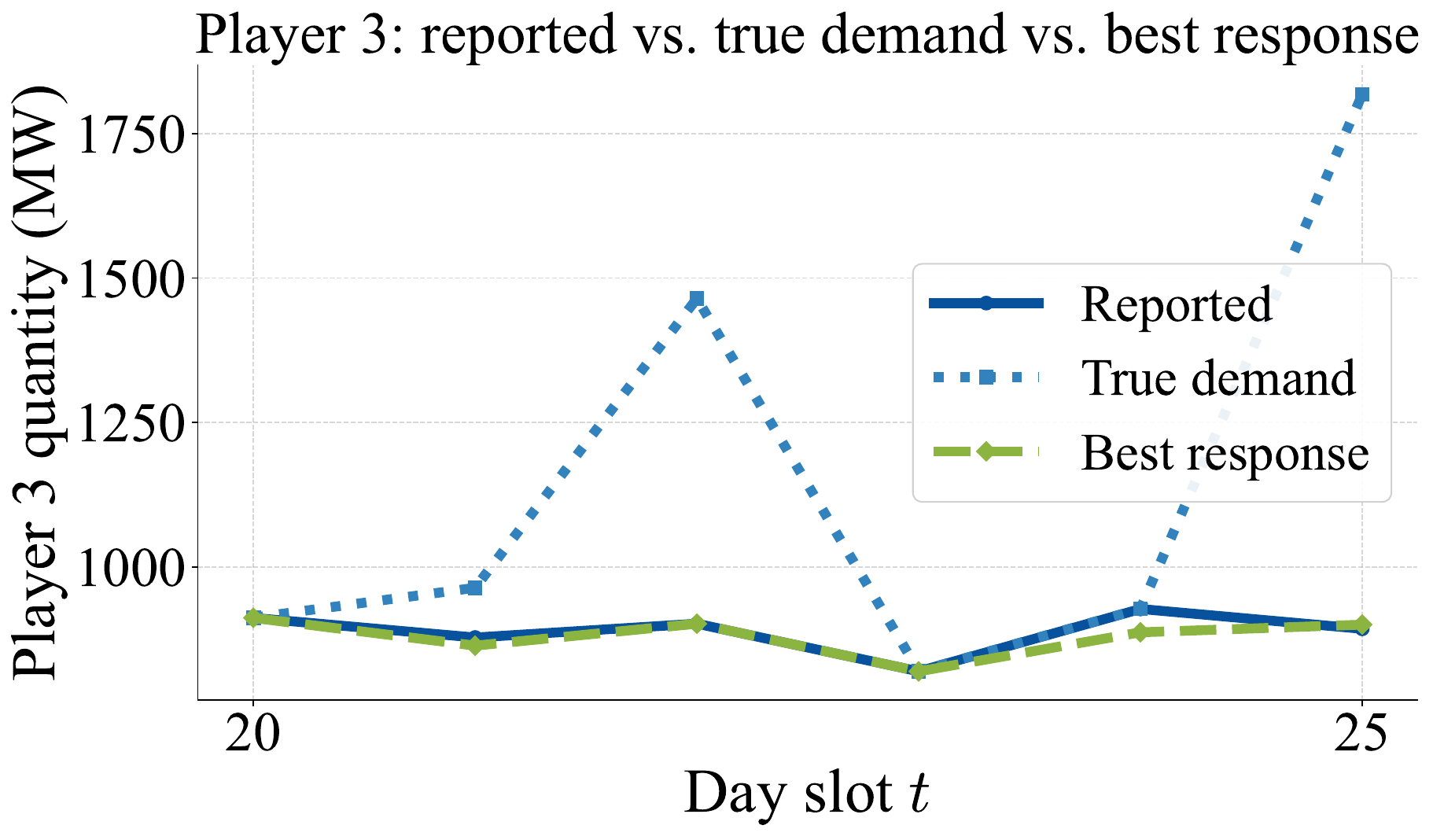}
    \end{subfigure}
    \vspace{5pt}
    \caption{Observed report $\mathrm{d}_i^{\,t}$ (solid blue), fitted truthful cap $\widehat{\mathrm{d}}_i^{\,t}$ (dashed cyan) and best response under $\rM^E$ (dashed green). The $\rM^E$ best response reproduces the observed report.}
    \label{fig:reports-obs-br}
\end{figure}
Moving forward, we set the strong-convexity modulus per player to $\boldsymbol \mu=(10^{5},10^{5},10^{6})$, mapped to $\boldsymbol \gamma_i$ via Table \autoref{tab:scoring_rules} under the log potential (logarithmic scoring rule). The spread is  asymmetric: a large penalty for LV and MV, and a strong one for HV. This way a single plot shows, what ``penalty in between'', and ``penalty on'' look like on the same real data.

\autoref{fig:reports-obs-br-pen} shows that under $\rM^S$, on exactly those gapped slots, the best response lifts off the observed curve and travels towards the fitted cap. The mechanism  makes the profitable under-report expensive enough that the player prefers to increase their report. On slots where the two $\rM^E$ curves were already close, $\rM^S$ changes nothing, which is the correct behavior of a mechanism that is supposed to activate only when there is something to correct.

\begin{figure}[htbp]
    \centering
    \begin{subfigure}{0.32\textwidth}
        \centering
        \includegraphics[width=\linewidth]{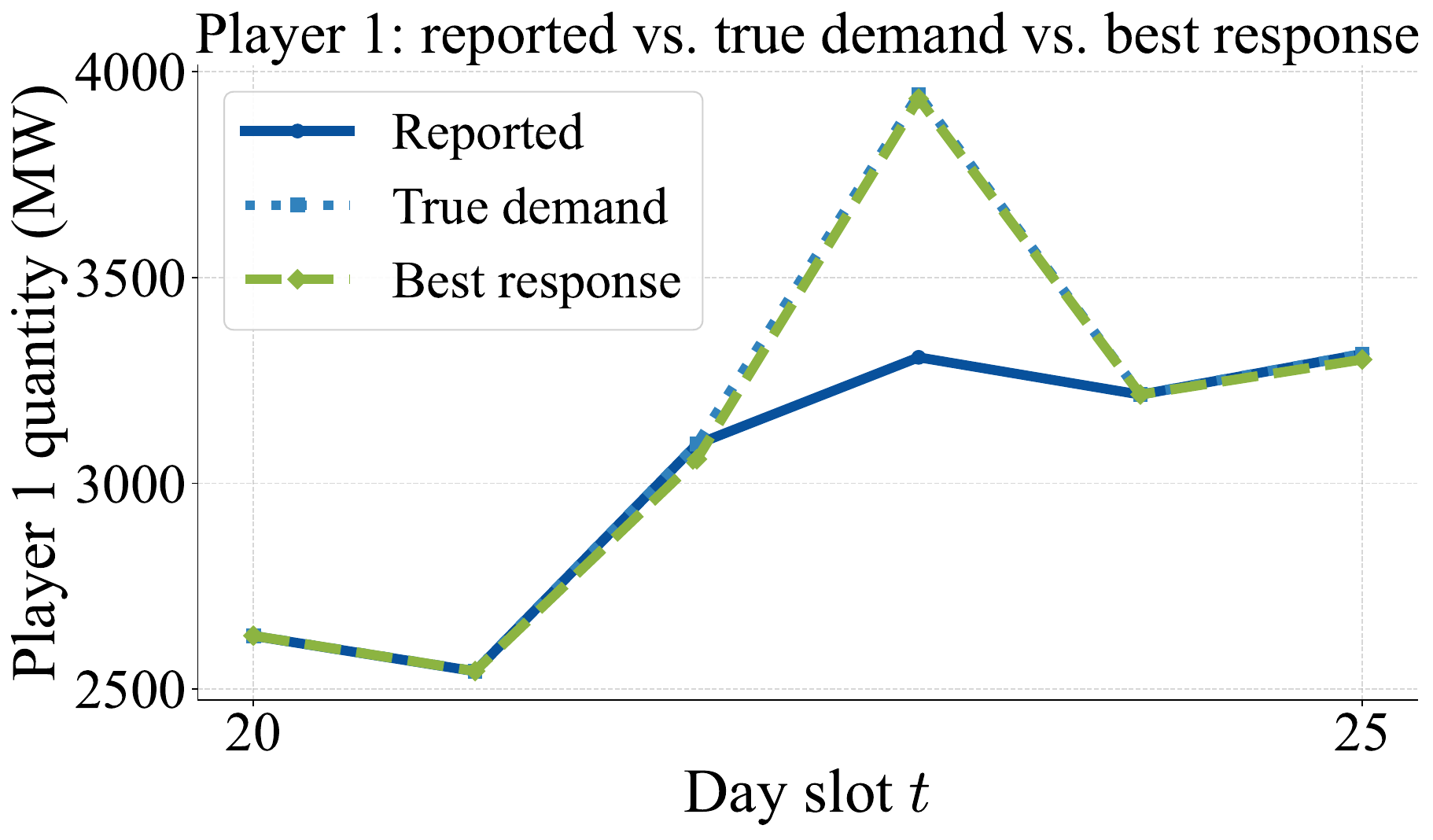}
    \end{subfigure}
    \hfill
    \begin{subfigure}{0.32\textwidth}
        \centering
        \includegraphics[width=\linewidth]{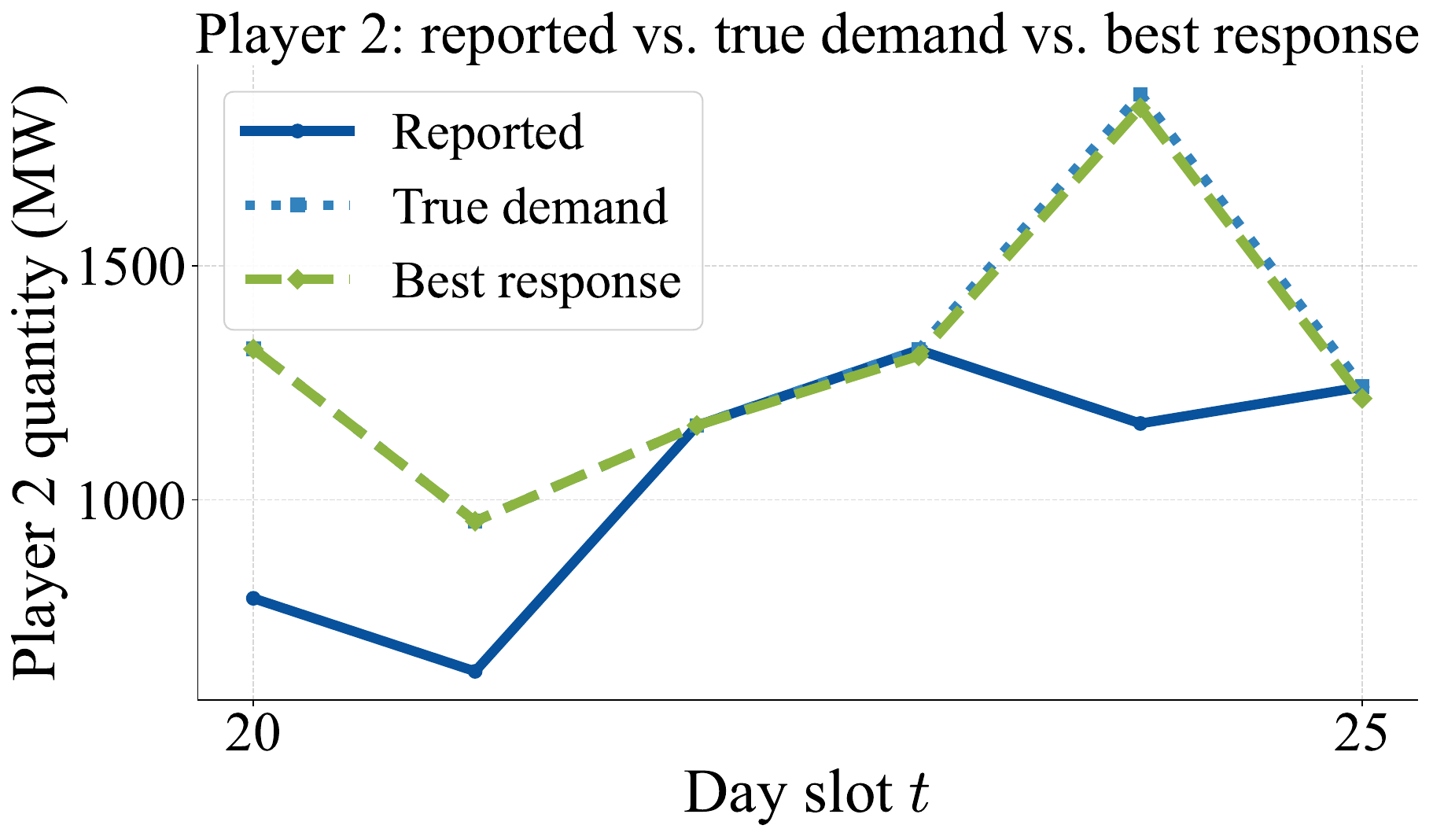}
    \end{subfigure}
    \hfill
    \begin{subfigure}{0.32\textwidth}
        \centering
        \includegraphics[width=\linewidth]{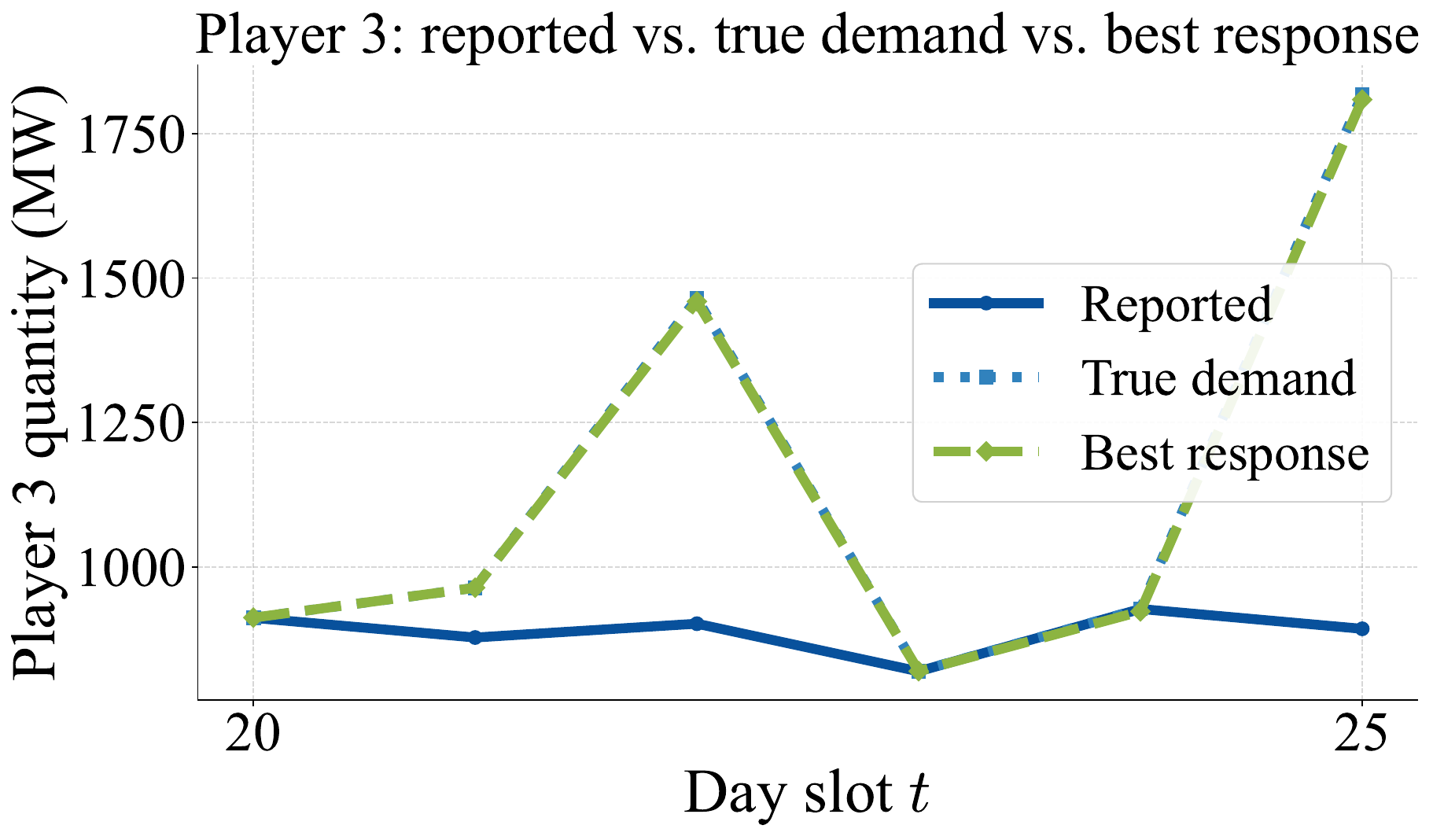}
    \end{subfigure}
    \vspace{5pt}
    \caption{The $\rM^S$ best response (dashed green) moves towards the true demand (dashed cyan) on exactly the slots where the two $\rM^E$ curves were apart.}
    \label{fig:reports-obs-br-pen}
\end{figure}
\begin{table}[htbp]
\centering
\small
\setlength{\tabcolsep}{8pt}
\begin{tabular}{lcc}
\toprule
Quantity                                          & under $\rM^E$                       & under $\rM^S$                              \\
\midrule
Day-ahead clearing price (\$/MWh)        & $80$                                & $80$                                        \\
Real-time clearing price (\$/MWh)        & $200$                               & $20$ \;(min $0$, max $40$)                  \\
Real-time imbalance (MWh)                  & $614$ \;(min $408$, max $919$)      & $24$ \;(min $0$, max $47$)                  \\
Total cleared demand (MWh)       & $4977$                              & $5547$                                      \\
\addlinespace[2pt]
Day-ahead marginal producer                             & bin $3$ of $|J|\!=\!4$              & bin $3$ of $|J|\!=\!4$                      \\
Real-time marginal producer                             & bin $|\bar J|\!=\!4$ (most expensive) & bin $1$ (cheapest), or inactive            \\
\bottomrule
\end{tabular}
\caption{Day-ahead and real-time clearing prices, real-time imbalance, total
cleared day-ahead demand, and the marginal producer on each side under the two
mechanisms. Reported quantities are means over the simulated slots, with the
per-slot minimum and maximum given in parentheses where dispersion is
non-trivial. Both cleared day-ahead totals fall inside the same day-ahead capacity bin (bin
$3$ of $|J|\!=\!4$), so the day-ahead clearing price is unaffected. Under
$\rM^S$ the residual the TSO must cover in real time shrinks by roughly an order
of magnitude, the marginal real-time producer drops from the most expensive bin to
the cheapest (or none at all), and the real-time clearing price collapses
correspondingly.}
\label{tab:prices}
\end{table}

\noindent The summary in~\autoref{tab:prices} tells the same story. Under both mechanisms the day-ahead market clears inside the
same generation capacity bin, and hence the day-ahead clearing price does not change: the larger reports
under $\rM^S$ consume more but remain inside the same bin without crossing the upper step of the bin. The
real-time market behaves differently. Under $\rM^E$ the residual that the TSO
must cover in real time is large enough to activate the most expensive producer on
the real-time supply stack, while under $\rM^S$ that residual is much smaller and either the
cheapest real-time marginal cost is activated or no real-time producer is dispatched at all and the real-time
clearing price collapses correspondingly. Both effects are consistent with what
the theory predicts.
\smallskip

\autoref{fig:imbalance} completes the picture by showing the real-time imbalance, i.e.\ how much demand had to be absorbed in the real-time market. Under $\rM^E$ this quantity is positive: that is the residual real-time profit channel~\autoref{thm:not-eps-BIC} tells us the mechanism is meant to close. Under $\rM^S$ the same slots show a substantially smaller residual. Price and imbalance curves show that the mechanism does what it is supposed to do overall. 
\begin{figure}[h]
\centering
\centering
\includegraphics[scale=0.25]{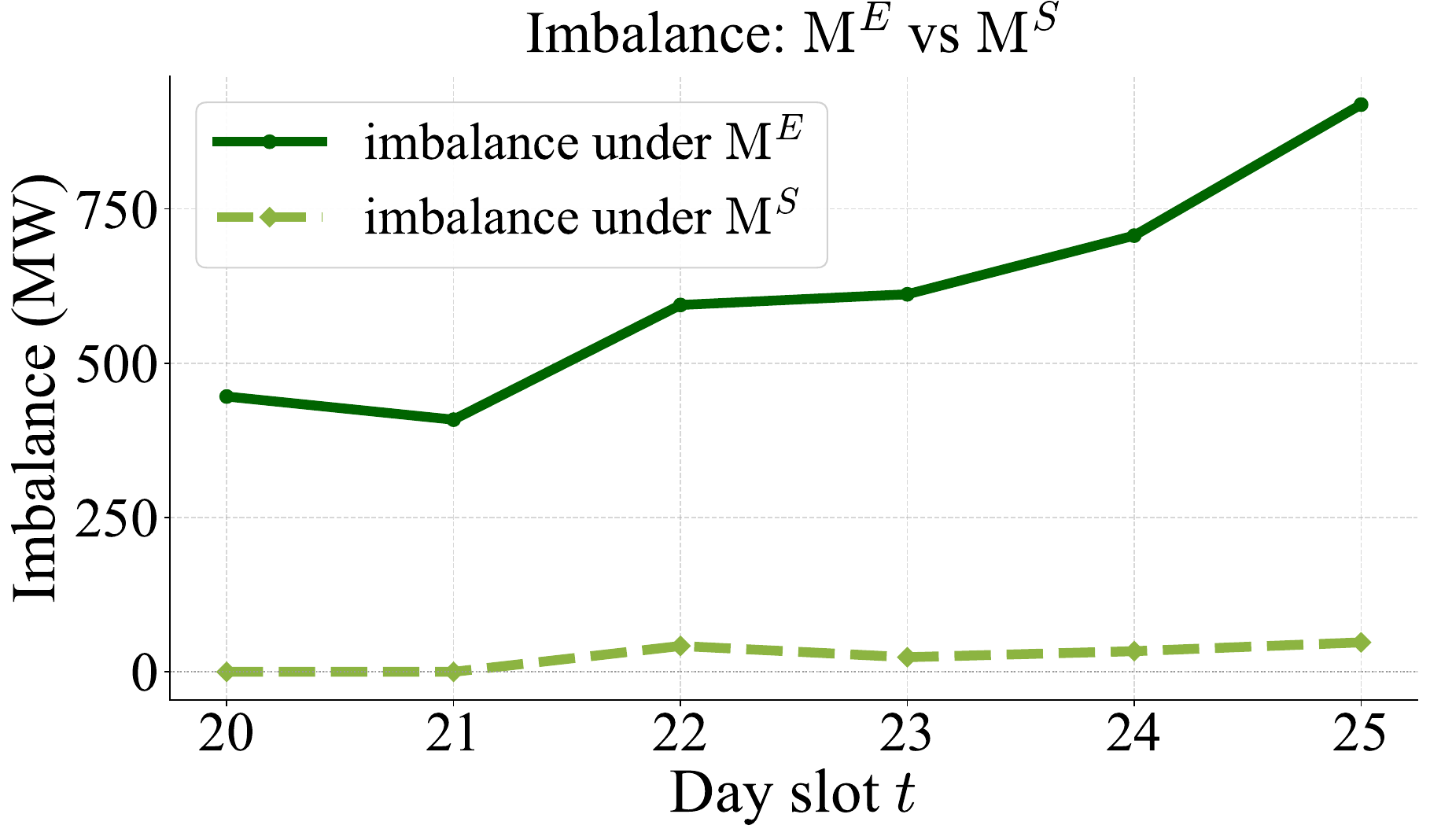}
\caption{Imbalance (day-ahead-cleared demand subtracted from the realized total consumption) under $\rM^E$ (dark green) and under $\rM^S$ (light green).}
\label{fig:imbalance}
\end{figure}

Lastly, the purpose of~\autoref{fig:penalty-and-utility} is to illustrate when mechanism $\rM^{S}$ activates and how much cost it incurs. In order to ease the comparison the top row of~\autoref{fig:penalty-and-utility} reproduces the content of \autoref{fig:reports-obs-br}. The two rows at the bottom of~\autoref{fig:penalty-and-utility} check the property that matters most for a well-posed mechanism: when a player is honest, the mechanism does not charge them. For each player we plot the per-slot penalty, and we look at what happens on the slots where the observed report already agrees with the fitted truthful demand (so the player is, in effect, honest) versus the slots where the fit detected a meaningful under-report (so the player is not honest). On the honest slots the penalty is flat and near zero across all three players. On the dishonest slots the penalty activates based on how strong the player's $\mu$ is: player~1 (LV, $\mu=10^{5}$), player~2 (HV, $\mu=10^{6}$), and player~3 (MV, $\mu=10^5$). The last row shows for each player, how the realized payoff diminishes under the equilibrium profiles played w.r.t.\ mechanism $\rM^S$ for each time slot. In short, the plots show that dishonest play is what turns the penalty on, and honest play keeps it off and how the payoffs reduce under our mechanism design.

\begin{figure}[h]
    \centering
    \captionsetup[subfigure]{labelformat=empty}
    \makebox[\linewidth][s]{%
        \begin{minipage}{0.32\linewidth}\centering
            \includegraphics[width=\linewidth]{experiments-figures/stage2_best_reported_truedemand_br_player1.pdf}
        \end{minipage}\hfill
        \begin{minipage}{0.32\linewidth}\centering
            \includegraphics[width=\linewidth]{experiments-figures/stage2_best_reported_truedemand_br_player2.pdf}
        \end{minipage}\hfill
        \begin{minipage}{0.32\linewidth}\centering
            \includegraphics[width=\linewidth]{experiments-figures/stage2_best_reported_truedemand_br_player3.pdf}
        \end{minipage}%
    }\\[2pt]
    {\footnotesize\itshape Reproduced from Fig.~\autoref{fig:reports-obs-br}, aligned column-wise with the panels below.}\\[8pt]
    \rule{0.95\linewidth}{0.3pt}\\[8pt]
    %
    \begin{minipage}{0.32\linewidth}\centering
        \includegraphics[width=\linewidth]{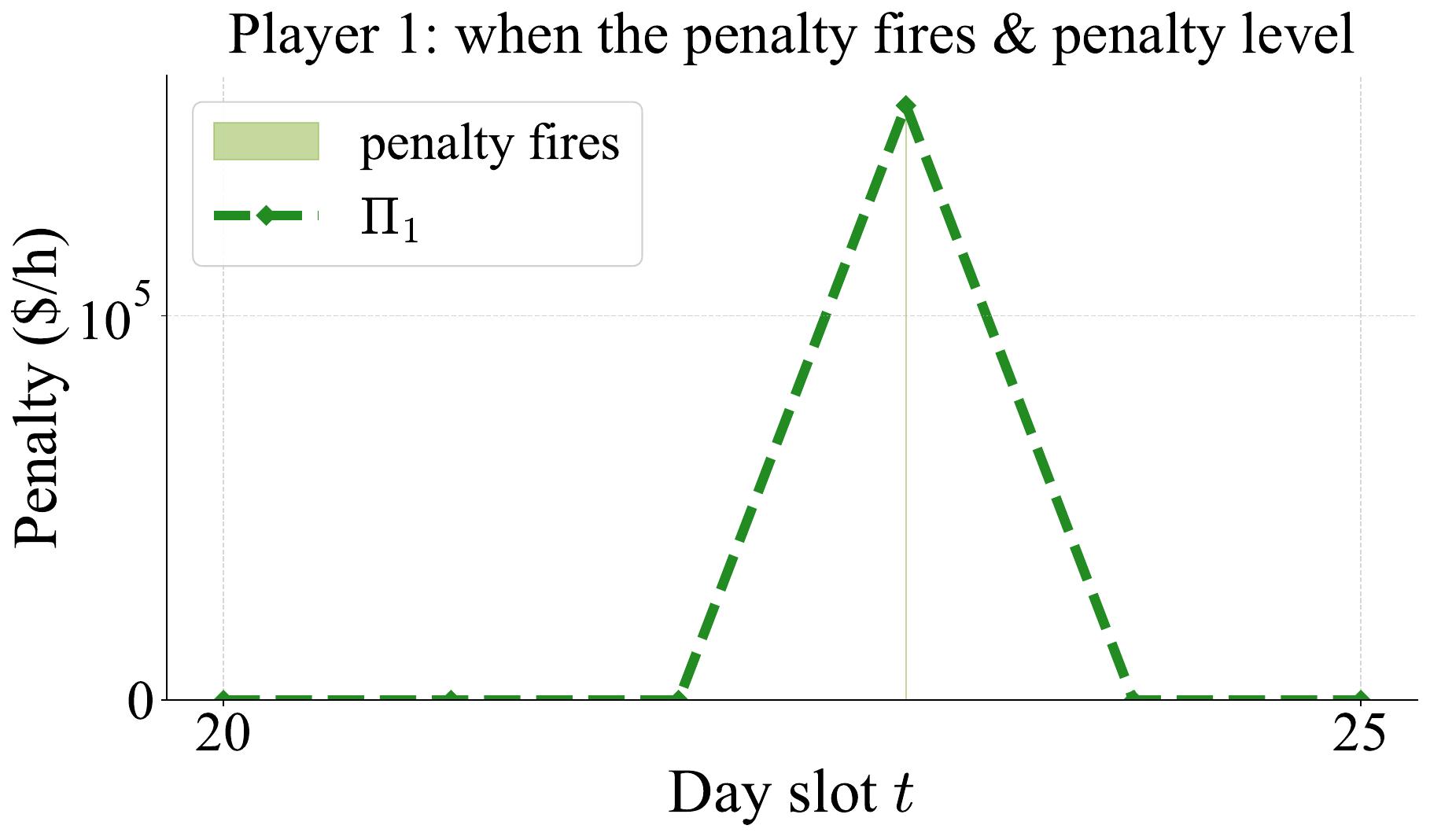}\\[-2pt]
        \includegraphics[width=\linewidth]{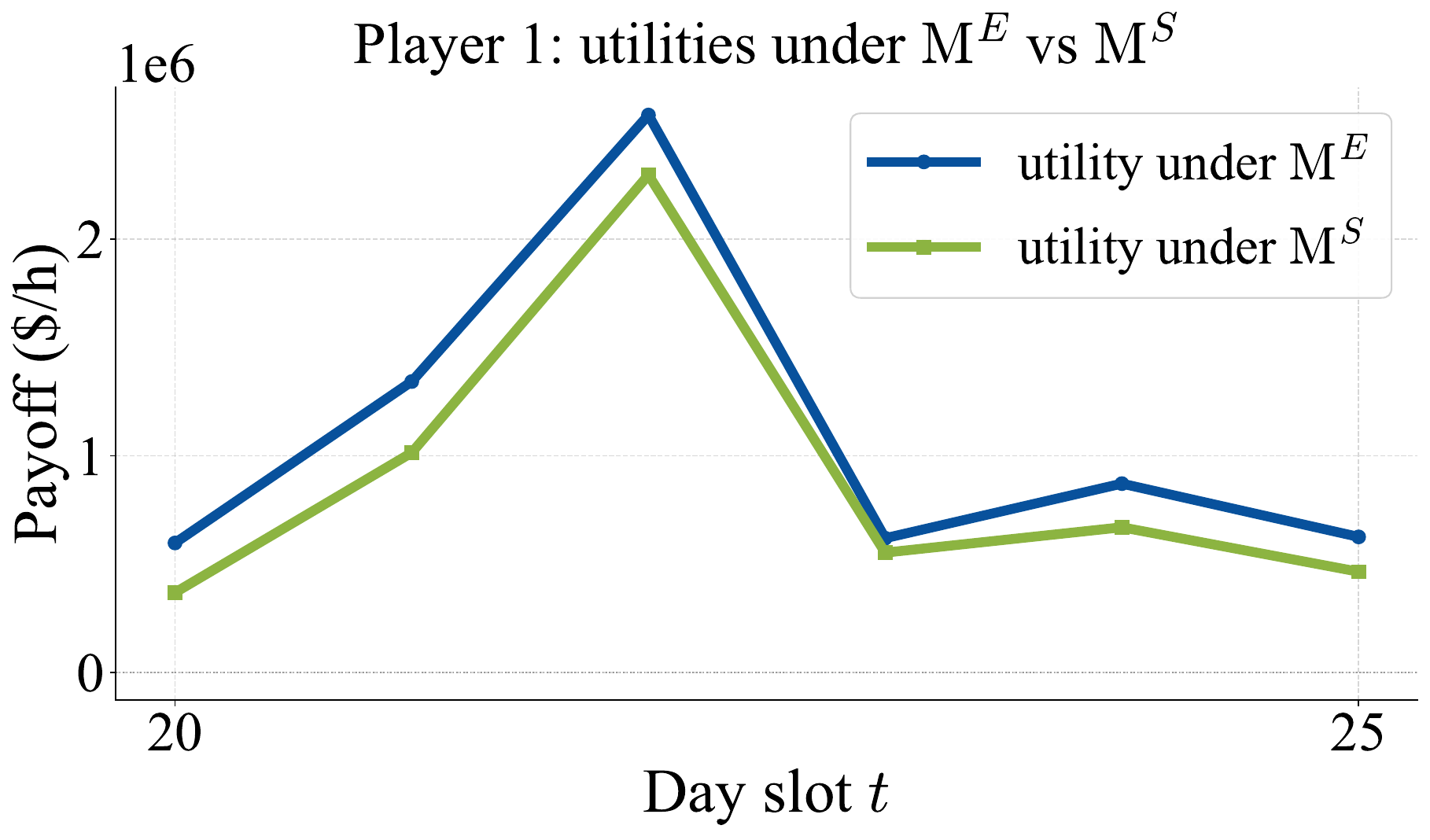}\\[1pt]
        {\small Player 1 (LV, $\mu=10^{5}$)}
    \end{minipage}\hfill
    \begin{minipage}{0.32\linewidth}\centering
        \includegraphics[width=\linewidth]{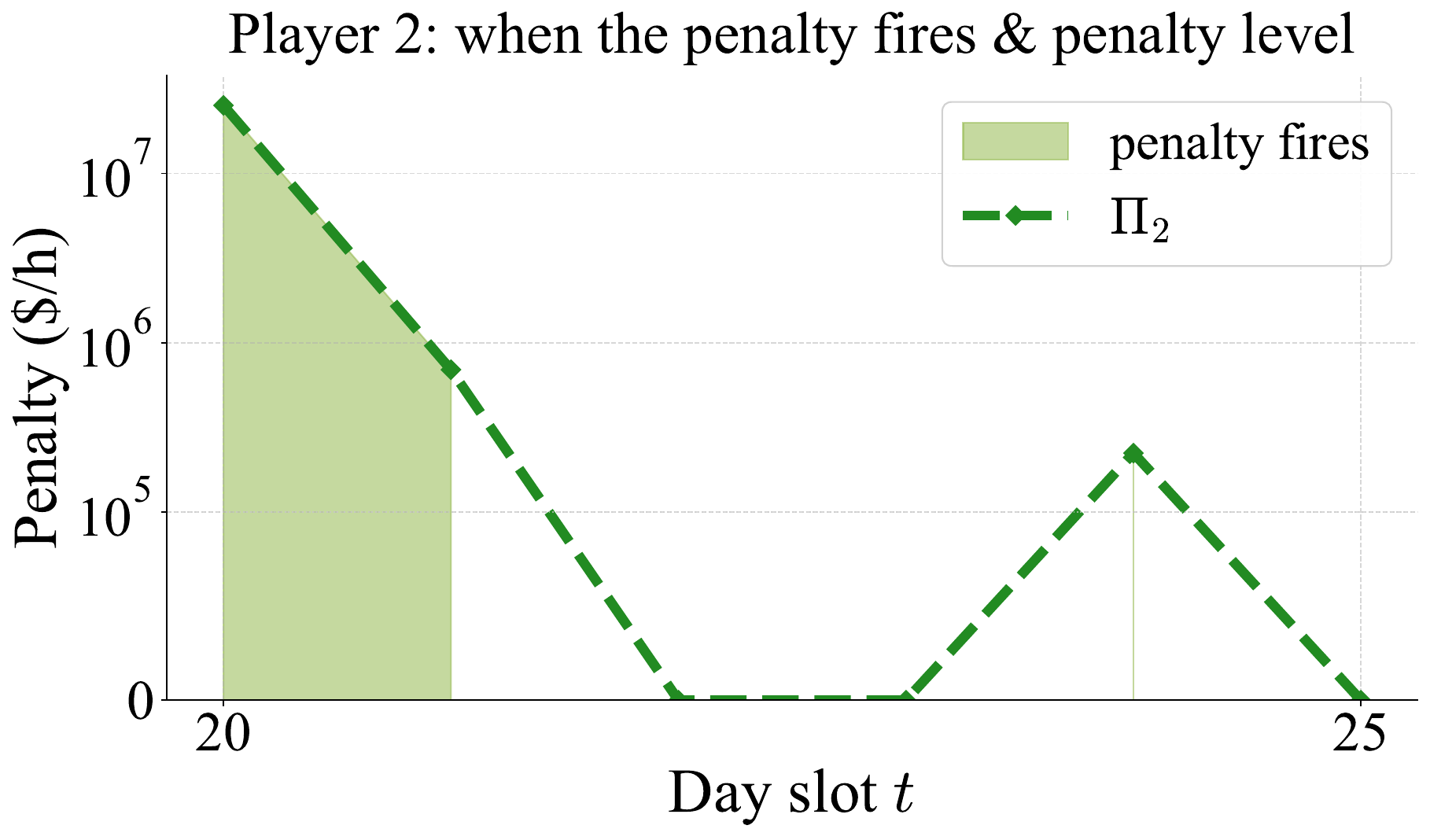}\\[-2pt]
        \includegraphics[width=\linewidth]{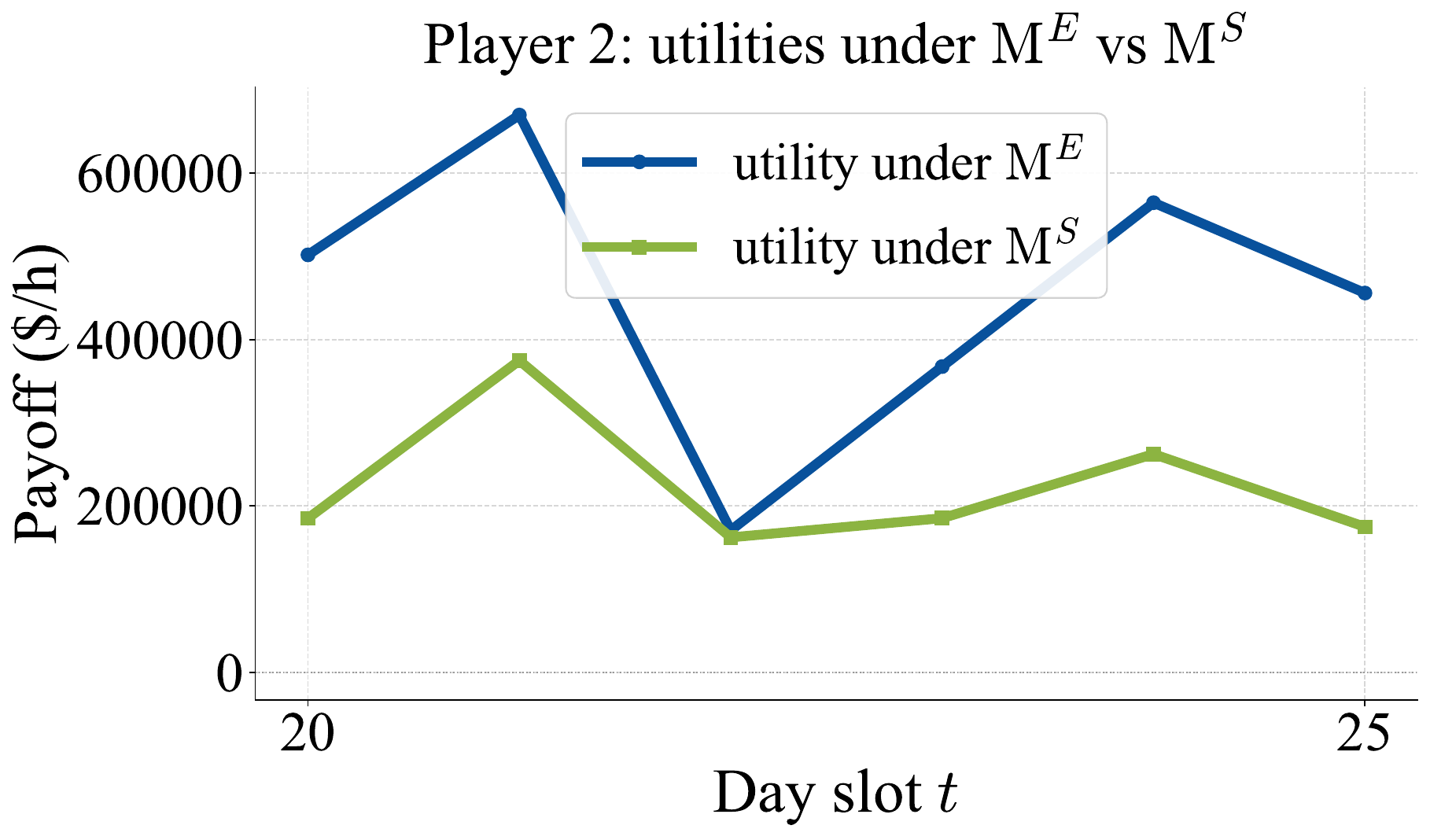}\\[1pt]
        {\small Player 2 (MV, $\mu=10^5$)}
    \end{minipage}\hfill
    \begin{minipage}{0.32\linewidth}\centering
        \includegraphics[width=\linewidth]{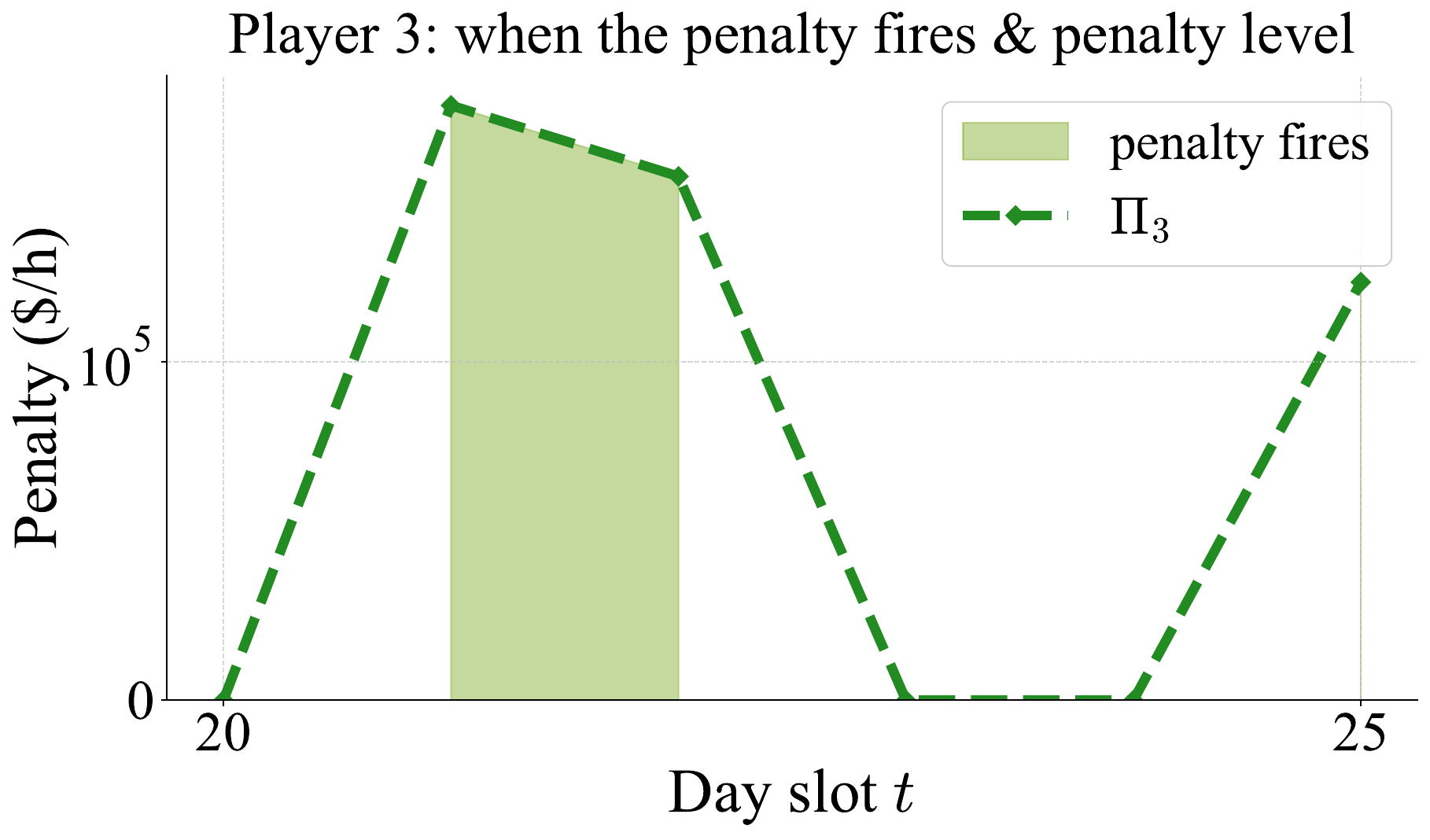}\\[-2pt]
        \includegraphics[width=\linewidth]{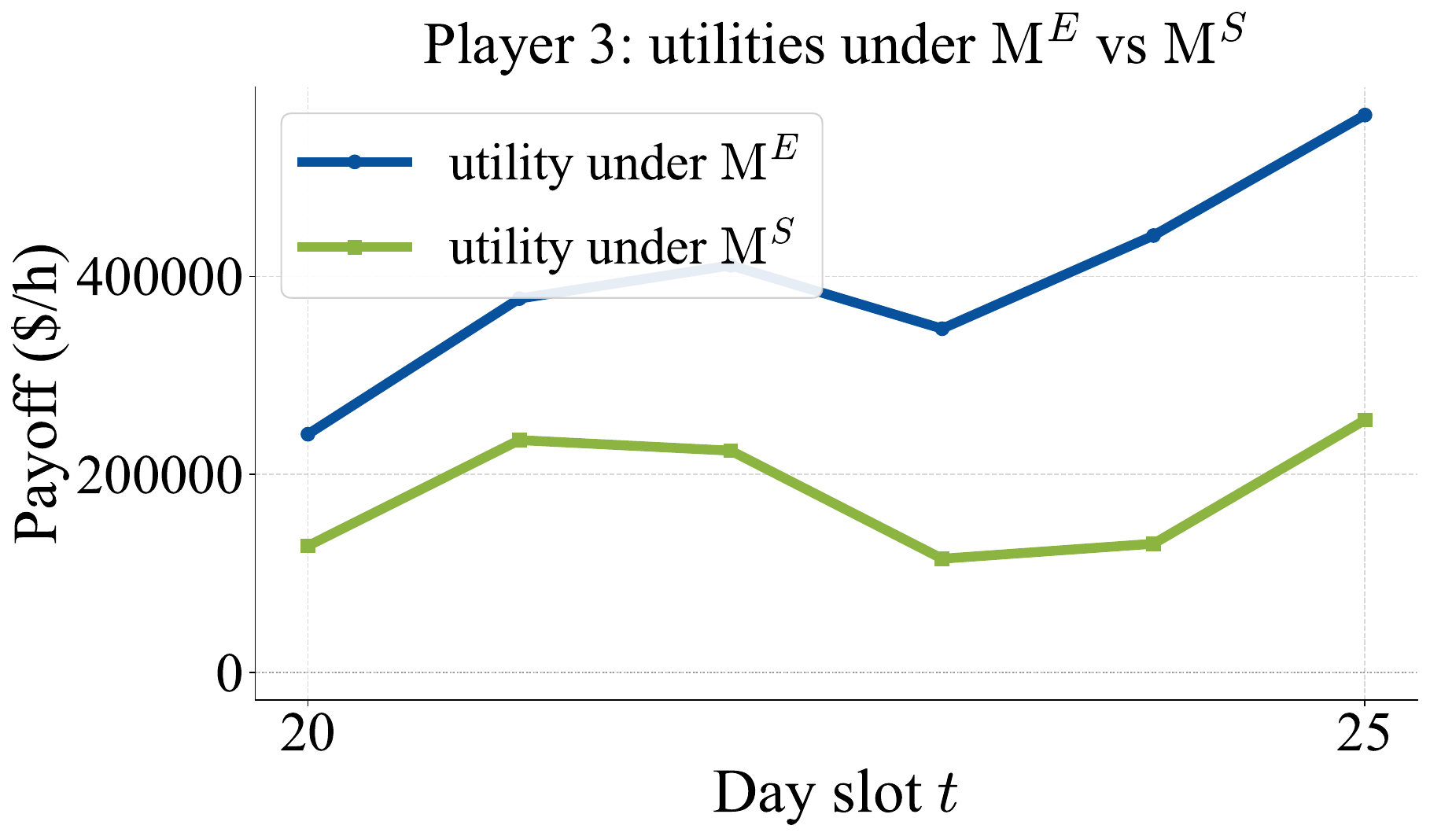}\\[1pt]
        {\small Player 3 (HV, $\mu=10^{6}$)}
    \end{minipage}
    \caption{Mechanism at work, per player. The top row reproduces~\autoref{fig:reports-obs-br} so the reader can line up the deterred slots with the penalty and utility rows below. Middle row: penalty under $\rM^S$; the light green band is where the penalty activates. Bottom row: realized utility (payoff) under $\rM^E$ (no penalty) vs.\ under $\rM^S$ (paid penalty already subtracted); the vertical gap on deterred slots is the money paid to the mechanism.}
    \label{fig:penalty-and-utility}
\end{figure}

\section{Conclusion}

We studied a strategic estimation problem in which an institution---the day-ahead market---aggregates reports from self-interested players in order to recover a latent population-level object: the true distribution of next-day demand. Classical statistical estimation treats data as drawn from nature; in our setting the data are reports chosen by players who respond to the very rules that will process them, and the estimator is therefore endogenous to its own data-generating process.
\smallskip

Empirically, this circularity has a visible signature. Using publicly available data from three European wholesale markets, we documented a persistent, one-sided gap between aggregate day-ahead reports and realized consumption. The gap survives long observation windows and changes in conditions, ruling out prediction noise as the explanation. We argued that it is best understood as the equilibrium response of strategic prosumers to the prevailing settlement rules.
\smallskip

Our theoretical results explain the asymmetry of this response. For players participating on the demand side only---\emph{consumers} in our model---the gain from misreporting shrinks at rate $1/\sqrt{|I|}$ in the size $|I|$ of the day-ahead market, so honesty is an approximate best response in large day-ahead markets. For consumers who also hold real-time generation---\emph{prosumers}---the gain is bounded below by a strictly positive constant that depends only on the prosumer's portfolio and the real-time supply stack, and does not vanish as the day-ahead market grows. The two-stage institution therefore cannot be made approximately truthful for everyone simultaneously.
\smallskip

To close this gap we introduced an imbalance-settlement mechanism built on a strictly proper, leave-one-out contrastive scoring rule. The auctioneer constructs two population-level subjective beliefs about the real-time outcome---one using the player's report and one omitting it---and charges each player according to her marginal contribution to the fidelity of the market's belief. The mechanism has a single tuning parameter $\mu$, the strong-convexity parameter of the strongly convex function associated with the scoring rule, which the designer can set large enough to dominate the prosumer's persistent strategic gain. Increasing $\mu$ also charges honest players whose private demand happens to lie far from the population, and we characterized the value $\mu^\star$ that balances these two effects: in a large market the residual incentive to misreport and the residual cost paid by honest participants vanish at the same rate. The mechanism is informationally light: the auctioneer needs nothing beyond what is already observable in the day-ahead stage, and the existing allocation and pricing rules are left untouched---only the settlement transfer is modified. 
\smallskip

We then tested the mechanism on real day-ahead and real-time data from a European wholesale market. Once the penalty is in place, reports rise toward true demand, more energy clears in the day-ahead stage, and the real-time imbalance shrinks by roughly an order of magnitude. In our model, the day-ahead clearing price is unaffected because the larger reports remain within the same step of the supply curve; the real-time clearing price however drops sharply because the smaller imbalance no longer requires the most expensive producers to be activated. The empirical picture matches the theory.


\appendix

\newpage
\bibliographystyle{plainnat}
\bibliography{refs}

\begin{thebibliography}{67}
\providecommand{\natexlab}[1]{#1}
\providecommand{\url}[1]{\texttt{#1}}
\expandafter\ifx\csname urlstyle\endcsname\relax
  \providecommand{\doi}[1]{doi: #1}\else
  \providecommand{\doi}{doi: \begingroup \urlstyle{rm}\Url}\fi

\bibitem[EU2(2019)]{EU2019_943}
{Regulation (EU) 2019/943 of the European Parliament and of the Council of 5 June 2019 on the internal market for electricity (recast)}.
\newblock {Official Journal of the European Union, L158, 14 June 2019, pp.54-124}, June 2019.
\newblock Applied from 1 January 2020.

\bibitem[Aumann(1992)]{aumann1992perspectives}
Robert Aumann.
\newblock Perspectives on bounded rationality.
\newblock In \emph{Proceedings of the 4th conference on Theoretical aspects of reasoning about knowledge}, pages 108--117, 1992.

\bibitem[Brier(1950)]{glenn1950verification}
W~Glenn Brier.
\newblock Verification of forecasts expressed in terms of probability.
\newblock \emph{Monthly weather review}, 78\penalty0 (1):\penalty0 1--3, 1950.

\bibitem[Cai et~al.(2015)Cai, Daskalakis, and Papadimitriou]{cai2015optimum}
Yang Cai, Constantinos Daskalakis, and Christos Papadimitriou.
\newblock Optimum statistical estimation with strategic data sources.
\newblock In \emph{Conference on Learning Theory}, pages 280--296. PMLR, 2015.

\bibitem[Chen and Teboulle(1993)]{chen1993convergence}
Gong Chen and Marc Teboulle.
\newblock Convergence analysis of a proximal-like minimization algorithm using bregman functions.
\newblock \emph{SIAM Journal on Optimization}, 3\penalty0 (3):\penalty0 538--543, 1993.

\bibitem[de~Finetti(1937)]{definetti1937foresight}
Bruno de~Finetti.
\newblock La pr{\'e}vision: ses lois logiques, ses sources subjectives.
\newblock \emph{Annales de l'Institut Henri Poincar{\'e}}, 7\penalty0 (1):\penalty0 1--68, 1937.
\newblock Translated as 'Foresight: Its Logical Laws, Its Subjective Sources' in Kyburg and Smokler (1964).

\bibitem[Dekel et~al.(2010)Dekel, Fischer, and Procaccia]{dekel2010incentive}
Ofer Dekel, Felix Fischer, and Ariel~D. Procaccia.
\newblock Incentive compatible regression learning.
\newblock \emph{Journal of Computer and System Sciences}, 76\penalty0 (8):\penalty0 759--777, 2010.
\newblock \doi{10.1016/j.jcss.2010.03.003}.

\bibitem[Edelman et~al.(2007)Edelman, Ostrovsky, and Schwarz]{Edelman}
Benjamin Edelman, Michael Ostrovsky, and Michael Schwarz.
\newblock Internet advertising and the generalized second-price auction: Selling billions of dollars worth of keywords.
\newblock \emph{American Economic Review}, 97\penalty0 (1):\penalty0 242–259, March 2007.
\newblock \doi{10.1257/aer.97.1.242}.
\newblock URL \url{https://www.aeaweb.org/articles?id=10.1257/aer.97.1.242}.

\bibitem[{ENTSO-E}({\natexlab{a}})]{ENTSOE_MissionStatement}
{ENTSO-E}.
\newblock {ENTSO-E Mission Statement}, {\natexlab{a}}.
\newblock URL \url{https://www.entsoe.eu/about/inside-entsoe/mission-statement/}.
\newblock About page (mission statement).

\bibitem[{ENTSO-E}({\natexlab{b}})]{entsoe2025greece}
{ENTSO-E}.
\newblock Total load -- greece (25 may 2025), {\natexlab{b}}.
\newblock \url{https://transparency.entsoe.eu/load-domain/r2/totalLoadR2/show?name=&defaultValue=false&viewType=GRAPH&areaType=CTA&atch=false&dateTime.dateTime=25.05.2025+00:00%7CCET%7CDAY&biddingZone.values=CTY%7C10YGR-HTSO-----Y!CTA%7C10YGR-HTSO-----Y&dateTime.timezone=CET_CEST&dateTime.timezone_input=CET+(UTC+1)+/+CEST+(UTC+2)}.

\bibitem[{ENTSO-E}(2026)]{entsoe}
{ENTSO-E}.
\newblock Entso-e website.
\newblock \url{https://www.entsoe.eu}, 2026.

\bibitem[EU(2019)]{EU2019Directive944}
EU.
\newblock Directive (eu) 2019/944 of the european parliament and of the council of 5 june 2019 on common rules for the internal market for electricity and amending directive 2012/27/eu (recast).
\newblock Official Journal of the European Union, 2019.
\newblock Available at \url{https://eur-lex.europa.eu/legal-content/EN/TXT/?uri=CELEX:32019L0944}.

\bibitem[Euphemia(2025)]{euphemia}
Euphemia.
\newblock Pan-european hybrid electricity market integration algorithm, 2025.
\newblock URL \url{https://www.nemo-committee.eu/}.
\newblock Accessed: 2025-01-07.

\bibitem[{European Commission}(2009)]{ThirdEnergyPackage2009}
{European Commission}.
\newblock Third energy package.
\newblock Legislative package, Directives 2009/72/EC and 2009/73/EC; Regulations (EC) No 713/2009, 714/2009, and 715/2009, 2009.
\newblock Official Journal of the European Union, L 211, 14 August 2009.

\bibitem[{European Commission}(2015)]{EU2015R1222}
{European Commission}.
\newblock Commission regulation (eu) 2015/1222 of 24 july 2015 establishing a guideline on capacity allocation and congestion management, 2015.
\newblock URL \url{https://eur-lex.europa.eu/legal-content/EN/TXT/PDF/?uri=CELEX:32015R1222}.
\newblock OJ L 197, 25.7.2015, pp.~24--72.

\bibitem[{European Commission}(2017{\natexlab{a}})]{EC2017_Regulation_2017_2195}
{European Commission}.
\newblock Commission regulation (eu) 2017/2195 of 23 november 2017 establishing a guideline on electricity balancing.
\newblock Official Journal of the European Union, November 2017{\natexlab{a}}.
\newblock URL \url{https://eur-lex.europa.eu/legal-content/EN/TXT/PDF/?uri=CELEX:32017R2195}.

\bibitem[{European Commission}(2017{\natexlab{b}})]{EU2017R2195}
{European Commission}.
\newblock {Commission Regulation (EU) 2017/2195 of 23 November 2017 establishing a guideline on electricity balancing}.
\newblock Official Journal of the European Union, L 312, pp. 6--53, November 2017{\natexlab{b}}.
\newblock URL \url{https://eur-lex.europa.eu/legal-content/EN/TXT/HTML/?uri=CELEX%3A32017R2195}.
\newblock Text with EEA relevance.

\bibitem[{European Commission}(2017{\natexlab{c}})]{EU2017R2195_EEA}
{European Commission}.
\newblock {Commission Regulation (EU) 2017/2195 of 23 November 2017 establishing a guideline on electricity balancing (Text with EEA relevance)}.
\newblock Official Journal of the European Union, L 312, 28 November 2017, pp. 6--53, November 2017{\natexlab{c}}.
\newblock URL \url{https://eur-lex.europa.eu/legal-content/EN/TXT/HTML/?uri=CELEX%3A32017R2195}.
\newblock Text with EEA relevance; Regulation (EU) 2017/2195 on electricity balancing guideline.

\bibitem[{European Commission}(2017{\natexlab{d}})]{EU2017Regulation2195}
{European Commission}.
\newblock {Commission Regulation (EU) 2017/2195 of 23 November 2017 establishing a guideline on electricity balancing}.
\newblock Commission Regulation / Technical Report 2017/2195, {Official Journal of the European Union}, 2017{\natexlab{d}}.
\newblock URL \url{https://eur-lex.europa.eu/legal-content/EN/TXT/PDF/?uri=CELEX%3A32017R2195}.
\newblock Accessed: 2025-09-09.

\bibitem[{European Commission}(2017{\natexlab{e}})]{EU2017_EBGL}
{European Commission}.
\newblock Commission regulation (eu) 2017/2195 of 23 november 2017 establishing a guideline on electricity balancing (ebgl).
\newblock Official Journal of the European Union, L 312, 28.11.2017, 2017{\natexlab{e}}.
\newblock URL \url{https://eur-lex.europa.eu/legal-content/EN/TXT/PDF/?uri=CELEX:32017R2195}.

\bibitem[{European Commission}(2019)]{CleanEnergyPackage2019}
{European Commission}.
\newblock Clean energy for all {Europeans} package.
\newblock Legislative package, Directives ({EU}) 2018/844, 2018/2001, 2018/2002, 2019/944; Regulations ({EU}) 2018/1999, 2019/941, 2019/942, 2019/943, 2019.
\newblock Official Journal of the European Union.

\bibitem[{European Commission}(2025)]{EC2025_DA15min}
{European Commission}.
\newblock Eu electricity trading in the day-ahead markets becomes more dynamic, 2025.
\newblock URL \url{https://energy.ec.europa.eu/news/eu-electricity-trading-day-ahead-markets-becomes-more-dynamic-2025-10-01_en}.
\newblock News item, 1 Oct 2025.

\bibitem[{European Parliament and Council of the European Union}(2019{\natexlab{a}})]{EU2019L0944}
{European Parliament and Council of the European Union}.
\newblock Directive (eu) 2019/944 of 5 june 2019 on common rules for the internal market for electricity and amending directive 2012/27/eu, 2019{\natexlab{a}}.
\newblock URL \url{https://eur-lex.europa.eu/legal-content/EN/TXT/PDF/?uri=CELEX:32019L0944}.
\newblock OJ L 158, 14.6.2019, pp.~125--199.

\bibitem[{European Parliament and Council of the European Union}(2019{\natexlab{b}})]{EU2019R0943}
{European Parliament and Council of the European Union}.
\newblock Regulation (eu) 2019/943 of 5 june 2019 on the internal market for electricity, 2019{\natexlab{b}}.
\newblock URL \url{https://eur-lex.europa.eu/legal-content/EN/TXT/PDF/?uri=CELEX:32019R0943}.
\newblock OJ L 158, 14.6.2019, pp.~54--124.

\bibitem[{Federal Energy Regulatory Commission}(2023)]{FERC_ReliabilityExplainer_2023}
{Federal Energy Regulatory Commission}.
\newblock {Reliability Explainer}, 2023.
\newblock URL \url{https://www.ferc.gov/reliability-explainer}.
\newblock Webpage last updated August 16, 2023.

\bibitem[{Federal Energy Regulatory Commission}(2025{\natexlab{a}})]{FERC_ISONE_Guide_2025}
{Federal Energy Regulatory Commission}.
\newblock {An Introductory Guide for Participation in ISO New England Processes}, 2025{\natexlab{a}}.
\newblock URL \url{https://www.ferc.gov/introductory-guide-participation-iso-new-england-processes}.
\newblock Webpage published April 11, 2025.

\bibitem[{Federal Energy Regulatory Commission}(2025{\natexlab{b}})]{FERC_IntroGuide_2025}
{Federal Energy Regulatory Commission}.
\newblock {An Introductory Guide to Electricity Markets Regulated by the Federal Energy Regulatory Commission}, 2025{\natexlab{b}}.
\newblock URL \url{https://www.ferc.gov/introductory-guide-electricity-markets-regulated-federal-energy-regulatory-commission}.
\newblock Webpage last updated April 3, 2025.

\bibitem[Gasparella et~al.(2023)Gasparella, Koolen, and Zucker]{Gasparella2023}
Andrea Gasparella, Derck Koolen, and Andreas Zucker.
\newblock The merit order and price-setting dynamics in {E}uropean electricity markets.
\newblock Technical Report JRC134300, European Commission, Joint Research Centre, Petten, 2023.

\bibitem[Ghorbani and Zou(2019)]{ghorbani2019data}
Amirata Ghorbani and James Zou.
\newblock Data shapley: Equitable valuation of data for machine learning.
\newblock In \emph{International conference on machine learning}, pages 2242--2251. PMLR, 2019.

\bibitem[Gneiting and Raftery(2007{\natexlab{a}})]{GneitingRaftery2007_StrictlyProperScoringRules}
Tilmann Gneiting and Adrian~E. Raftery.
\newblock Strictly proper scoring rules, prediction, and estimation.
\newblock \emph{Journal of the American Statistical Association}, 102\penalty0 (477):\penalty0 359--378, 2007{\natexlab{a}}.
\newblock \doi{10.1198/016214506000001437}.

\bibitem[Gneiting and Raftery(2007{\natexlab{b}})]{gneiting2007strictly}
Tilmann Gneiting and Adrian~E. Raftery.
\newblock Strictly proper scoring rules, prediction, and estimation.
\newblock \emph{Journal of the American Statistical Association}, 102\penalty0 (477):\penalty0 359--378, 2007{\natexlab{b}}.

\bibitem[Good(1952)]{good1952rational}
I.~J. Good.
\newblock Rational decisions.
\newblock \emph{Journal of the Royal Statistical Society. Series B (Methodological)}, 14\penalty0 (1):\penalty0 107--114, 1952.
\newblock URL \url{https://www.jstor.org/stable/2984087}.

\bibitem[Good et~al.(1964)Good, Mayne, Smith, and Page]{good1964scientist}
Irving~John Good, AJ~Mayne, JM~Smith, and Thornton Page.
\newblock The scientist speculates: an anthology of partly-baked ideas.
\newblock \emph{American Journal of Physics}, 32\penalty0 (4):\penalty0 322--322, 1964.

\bibitem[Hanson(2003)]{hanson2003combinatorial}
Robin Hanson.
\newblock Combinatorial information market design.
\newblock \emph{Information Systems Frontiers}, 5\penalty0 (1):\penalty0 107--119, 2003.

\bibitem[Hanson(2007)]{Hanson2007_LogarithmicMSR}
Robin Hanson.
\newblock Logarithmic market scoring rules for modular combinatorial information aggregation.
\newblock \emph{Journal of Prediction Markets}, 1\penalty0 (1):\penalty0 3--15, 2007.
\newblock \doi{10.5750/jpm.v1i1.417}.

\bibitem[Hardt et~al.(2016)Hardt, Megiddo, Papadimitriou, and Wootters]{hardt2016strategic}
Moritz Hardt, Nimrod Megiddo, Christos Papadimitriou, and Mary Wootters.
\newblock Strategic classification.
\newblock In \emph{Proceedings of the 2016 ACM Conference on Innovations in Theoretical Computer Science ({ITCS})}, pages 111--122. ACM, 2016.
\newblock \doi{10.1145/2840728.2840730}.

\bibitem[{HENEX}()]{HEnEx2025}
{HENEX}.
\newblock Hellenic energy exchange (henex): Day-ahead market report.
\newblock \url{https://www.enexgroup.gr/web/guest/markets-publications-el-day-ahead-market}.

\bibitem[HEnEx(2025)]{henex}
HEnEx.
\newblock Hellenic energy exchange, 2025.
\newblock URL \url{https://www.enexgroup.gr/}.
\newblock Accessed: 2025-01-07.

\bibitem[{ISO New England Inc.}(n.d.)]{ISO_NE_Regulation_Item}
{ISO New England Inc.}
\newblock {Regulation (Settlement item description)}, n.d.
\newblock URL \url{https://www.iso-ne.com/markets-operations/settlements/understand-bill/item-descriptions/regulation}.
\newblock Webpage describing regulation as an ancillary service provided based on regulation supply offers.

\bibitem[Kalai and Lehrer(1995)]{kalai1995subjective}
Ehud Kalai and Ehud Lehrer.
\newblock Subjective games and equilibria.
\newblock \emph{Games and economic behavior}, 8\penalty0 (1):\penalty0 123--163, 1995.

\bibitem[Kirschen and Strbac(2018)]{Kirschen2018}
Daniel Kirschen and Goran Strbac.
\newblock \emph{Fundamentals of Power System Economics}.
\newblock Wiley, 2nd edition, 2018.
\newblock See Chapter~3 (generating electricity, marginal cost of generation, merit order dispatch).

\bibitem[Lee et~al.(1997)Lee, Padmanabhan, and Whang]{bullwhip_effect}
Hau~L. Lee, V.~Padmanabhan, and Seungjin Whang.
\newblock Information distortion in a supply chain: The bullwhip effect.
\newblock \emph{Management Science}, 43\penalty0 (4):\penalty0 546--558, 1997.
\newblock ISSN 00251909, 15265501.
\newblock URL \url{http://www.jstor.org/stable/2634565}.

\bibitem[L\'{e}vy(1937)]{Levy1937}
Paul L\'{e}vy.
\newblock \emph{Th\'{e}orie de l'addition des variables al\'{e}atoires}.
\newblock Gauthier-Villars, Paris, 1937.

\bibitem[McCarthy(1956)]{mccarthy1956measures}
John McCarthy.
\newblock Measures of the value of information.
\newblock \emph{Proceedings of the National Academy of Sciences}, 42\penalty0 (9):\penalty0 654--655, 1956.

\bibitem[Meir et~al.(2012)Meir, Procaccia, and Rosenschein]{meir2012algorithms}
Reshef Meir, Ariel~D. Procaccia, and Jeffrey~S. Rosenschein.
\newblock Algorithms for strategyproof classification.
\newblock \emph{Artificial Intelligence}, 186:\penalty0 123--156, 2012.
\newblock \doi{10.1016/j.artint.2012.03.008}.

\bibitem[{Monitoring Analytics}(2024)]{MonitoringAnalytics2023}
{Monitoring Analytics}.
\newblock 2023 {S}tate of the {M}arket {R}eport for {PJM}, {S}ection 3: {E}nergy {M}arket.
\newblock Technical report, March 2024.
\newblock URL \url{https://www.monitoringanalytics.com/reports/PJM_State_of_the_Market/2023/2023-som-pjm-sec3.pdf}.
\newblock See pp.~130--132, Figs.~3-1 and~3-2 (day-ahead supply curves); pp.~143--148 (LMP fuel cost component decomposition).

\bibitem[Myerson(1981)]{Myerson-optimal}
Roger~B. Myerson.
\newblock Optimal auction design.
\newblock \emph{Mathematics of Operations Research}, 6\penalty0 (1):\penalty0 58--73, 1981.
\newblock ISSN 0364765X, 15265471.
\newblock URL \url{http://www.jstor.org/stable/3689266}.

\bibitem[Nesterov et~al.(2018)]{nesterov2018lectures}
Yurii Nesterov et~al.
\newblock \emph{Lectures on convex optimization}, volume 137.
\newblock Springer, 2018.

\bibitem[{North American Electric Reliability Corporation (NERC)}(2025{\natexlab{a}})]{NERCGlossary2025}
{North American Electric Reliability Corporation (NERC)}.
\newblock Glossary of terms used in nerc reliability standards, October 2025{\natexlab{a}}.
\newblock URL \url{https://www.nerc.com/globalassets/standards/reliability-standards/glossary_of_terms.pdf}.
\newblock Updated October 1, 2025.

\bibitem[{North American Electric Reliability Corporation (NERC)}(2025{\natexlab{b}})]{NERC_Glossary_2025}
{North American Electric Reliability Corporation (NERC)}.
\newblock \emph{{Glossary of Terms Used in NERC Reliability Standards}}.
\newblock {NERC}, November 2025{\natexlab{b}}.
\newblock URL \url{https://www.nerc.com/globalassets/standards/reliability-standards/glossary_of_terms.pdf}.
\newblock Updated November 5, 2025.

\bibitem[Perdomo et~al.(2020)Perdomo, Zrnic, Mendler-D{\"u}nner, and Hardt]{perdomo2020performative}
Juan~C. Perdomo, Tijana Zrnic, Celestine Mendler-D{\"u}nner, and Moritz Hardt.
\newblock Performative prediction.
\newblock In \emph{Proceedings of the 37th International Conference on Machine Learning ({ICML})}, volume 119 of \emph{Proceedings of Machine Learning Research}, pages 7599--7609, 2020.

\bibitem[Petrov(1975)]{Petrov1975}
Valentin~V. Petrov.
\newblock \emph{Sums of Independent Random Variables}, volume~82 of \emph{Ergebnisse der Mathematik und ihrer Grenzgebiete}.
\newblock Springer-Verlag, Berlin, 1975.
\newblock Translated from the Russian by A. A. Brown.

\bibitem[{PJM Interconnection}(2023)]{PJM_Manual11_Redline_2023}
{PJM Interconnection}.
\newblock {Manual 11 revisions (redline): Energy Market Business Rules and Day-Ahead timeline (committee material)}, 2023.
\newblock URL \url{https://www.pjm.com/-/media/DotCom/committees-groups/committees/mic/2023/20230412/20230412-item-02-3---real-time-values-manual-11-revisions---redline.ashx}.
\newblock Committee document dated April 12, 2023 (PDF).

\bibitem[{PJM Interconnection}(2024{\natexlab{a}})]{PJM_Manual11}
{PJM Interconnection}.
\newblock \emph{Manual 11: {E}nergy \& {A}ncillary {S}ervices {M}arket {O}perations}, revision 129 edition, February 2024{\natexlab{a}}.
\newblock URL \url{https://www.pjm.com/library/manuals}.
\newblock See Section~5, pp.~63--80 (day-ahead market scheduling and clearing procedure).

\bibitem[{PJM Interconnection}(2024{\natexlab{b}})]{PJM_Planning2023}
{PJM Interconnection}.
\newblock {PJM} publishes annual planning report for 2023.
\newblock PJM Inside Lines, 2024{\natexlab{b}}.
\newblock URL \url{https://insidelines.pjm.com/pjm-publishes-annual-planning-report-for-2023/}.
\newblock Documents over 1,400 generators and approximately 65,000~MW of natural gas capacity interconnection rights in PJM.

\bibitem[{PJM Interconnection}(2025)]{PJM_AncillaryServices_FactSheet_2025}
{PJM Interconnection}.
\newblock {Ancillary Services Fact Sheet}, January 2025.
\newblock URL \url{https://www.pjm.com/-/media/DotCom/about-pjm/newsroom/fact-sheets/ancillary-services-fact-sheet.pdf}.
\newblock Fact sheet dated Jan. 29, 2025.

\bibitem[Rogozin(1961)]{Rogozin1961}
B.~A. Rogozin.
\newblock An estimate for concentration functions.
\newblock \emph{Theory of Probability \& Its Applications}, 6\penalty0 (1):\penalty0 94--97, 1961.
\newblock \doi{10.1137/1106007}.

\bibitem[Rothschild and Stiglitz(1976)]{rothy}
Michael Rothschild and Joseph Stiglitz.
\newblock Equilibrium in competitive insurance markets: An essay on the economics of imperfect information.
\newblock \emph{The Quarterly Journal of Economics}, 90\penalty0 (4):\penalty0 629--649, 1976.
\newblock ISSN 00335533, 15314650.
\newblock URL \url{http://www.jstor.org/stable/1885326}.

\bibitem[Savage(1971)]{savage1971elicitation}
Leonard~J Savage.
\newblock Elicitation of personal probabilities and expectations.
\newblock In V.~P. Godambe and D.~A. Sprott, editors, \emph{Foundations of statistical inference}, pages 159--175. Holt, Rinehart and Winston, 1971.

\bibitem[{SDAC Project Parties}(2025)]{SDAC2025_15minMTU}
{SDAC Project Parties}.
\newblock 15-minute mtu in sdac was implemented (press release), 2025.
\newblock URL \url{https://eepublicdownloads.blob.core.windows.net/public-cdn-container/clean-documents/Network%20codes%20documents/NC%20CACM/SDAC%202025/SDAC_15-Minute_MTU_Was_Implemented.pdf}.
\newblock Press release dated 01 Oct 2025; go-live on trading day 30 Sep 2025 for delivery day 1 Oct 2025.

\bibitem[Sorin(1992)]{sorin1992information}
Sylvain Sorin.
\newblock Information and rationality: some comments.
\newblock \emph{Annales d'Economie et de Statistique}, pages 315--325, 1992.

\bibitem[Stoft(2002)]{stoft2002power}
Steven Stoft.
\newblock Power system economics.
\newblock \emph{Journal of Energy Literature}, 8:\penalty0 94--99, 2002.

\bibitem[{U.S. Department of Energy}(2022)]{DOE_BA_Backgrounder_2022}
{U.S. Department of Energy}.
\newblock {How it Works: The Role of a Balancing Authority}, 2022.
\newblock URL \url{https://www.energy.gov/sites/default/files/2023-08/Balancing%20Authority%20Backgrounder_2022-Formatted_041723_508.pdf}.
\newblock Backgrounder (document label: 2022; formatted 2023-04-17 per filename).

\bibitem[{U.S. Energy Information Administration}(2021)]{EIA_PJM_NewerTech}
{U.S. Energy Information Administration}.
\newblock Newer-technology natural gas-fired generators are utilized more than older units in {PJM}.
\newblock Today in Energy, April 2021.
\newblock URL \url{https://www.eia.gov/todayinenergy/detail.php?id=47556}.
\newblock Documents heat rates by turbine class: B/D/E-class ({$>$}8,000~BTU/kWh), F-class (7,000--7,300~BTU/kWh), H/J-class (6,700~BTU/kWh).

\bibitem[{U.S. Energy Information Administration}(2023)]{EIA_CCGT_Utilization}
{U.S. Energy Information Administration}.
\newblock Natural gas combined-cycle power plants increased utilization with improved technology.
\newblock Today in Energy, December 2023.
\newblock URL \url{https://www.eia.gov/todayinenergy/detail.php?id=60984}.
\newblock Documents fleet-wide average heat rate of 6,960~BTU/kWh for CCGT plants built 2010--2022.

\bibitem[Vickrey(1961)]{vickrey1961counterspeculation}
William Vickrey.
\newblock Counterspeculation, auctions, and competitive sealed tenders.
\newblock \emph{The Journal of finance}, 16\penalty0 (1):\penalty0 8--37, 1961.

\bibitem[Zakeri et~al.(2023)Zakeri, Staffell, Dodds, Grubb, Ekins, J{\"a}{\"a}skel{\"a}inen, Cross, Helin, and Castagneto~Gissey]{Zakeri2023}
Behnam Zakeri, Iain Staffell, Paul~E. Dodds, Michael Grubb, Paul Ekins, Jaakko J{\"a}{\"a}skel{\"a}inen, Samuel Cross, Kristo Helin, and Giorgio Castagneto~Gissey.
\newblock The role of natural gas in setting electricity prices in {E}urope.
\newblock \emph{Energy Reports}, 10:\penalty0 2778--2792, 2023.
\newblock \doi{10.1016/j.egyr.2023.09.069}.
\newblock See p.~2778 (Abstract: headline results); p.~2780, Section~3 (methodology based on day-ahead market data); p.~2785, Fig.~5 (country-level marginal shares in 2021).

\end{thebibliography}

\end{document}